\documentclass[11pt]{article}  

\usepackage{graphicx}

\usepackage{a41}
\usepackage{color}
\usepackage[rflt]{floatflt}

\usepackage{array}

\usepackage[english]{babel}
\usepackage[latin1]{inputenc}
\usepackage[T1]{fontenc}
\usepackage{ae}

\usepackage{url}

\usepackage{amsmath, amsthm, amssymb}
\newtheorem{thm}{Theorem}[section]

\newtheorem{lemma}[thm]{Lemma}
\newtheorem{prop}[thm]{Proposition}

\newtheorem{definition}[thm]{Definition}

\newtheorem{remark}[thm]{Remark}

\newcommand{\NN}{\nonumber}
\newcommand{\ve}[1]{\textit{\textbf{#1}}}
\newcommand{\abs}[1]{\left|{#1}\right|}

\newcommand{\sign}[1]{\textnormal{sign}\hspace{-0.1em}\left(#1\right)}

\renewcommand{\H}[2]{\textnormal{H}_{#1}\hspace{-0.2em}\left(#2\right)}
\newcommand{\Hma}[2]{\textnormal{H}_{#1}\left(#2\right)}

\newcommand{\M}[2]{\textnormal{M}\hspace{-0.2em}\left(#1,#2\right)}
\newcommand{\Mma}[2]{\textnormal{M}\left(#1,#2\right)}
\newcommand{\Mp}[2]{\textnormal{M}^+\hspace{-0.2em}\left(#1,#2\right)}
\renewcommand{\S}[2]{\textnormal{S}_{#1}\hspace{-0.2em}\left(#2\right)}
\renewcommand{\SS}[3]{\textnormal{S}_{#1}\hspace{-0.2em}\left(#2;#3\right)}

\newcommand{\R}{\mathbb R}

\newcommand{\C}{\mathbb C}

\newcommand{\Z}{\mathbb Z}
\newcommand{\N}{\mathbb N}

\newcommand{\Mvec}{{\bf \rm M}}

\newcommand{\ie}{i.e.,\ }
\newcommand{\eg}{e.g.,\ }
\renewcommand{\Re}{\operatorname{Re}}
\renewcommand{\Im}{\operatorname{Im}}

\newcommand{\ep}{\varepsilon}

\let\set\mathbb

\usepackage{rotating}
\newcommand{\shuffle}{\, \raisebox{1.2ex}[0mm][0mm]{\rotatebox{270}{$\exists$}} \,}
\usepackage{algorithm}
\usepackage{algorithmicx}
\usepackage{algpseudocode}
\usepackage{graphicx}

\newcounter{mmacnt}
\def\restartmma{\setcounter{mmacnt}{0}}
\restartmma \catcode`|=\active
\def|#1|{\mathrm{#1}}
\catcode`|=12
\newenvironment{mma}{
 \par\smallskip
 \catcode`|=\active
 \parskip=0pt\parindent=0pt 
 \small
 \def\In##1\\{%
   \def\linebreak{\hfill\break\null\qquad}%
   \refstepcounter{mmacnt}
   \hangindent=2.5em\hangafter=0
   \leavevmode
   \llap{\tiny\sffamily In[\arabic{mmacnt}]:=\kern.5em}%
   \mathversion{bold}\footnotesize$\displaystyle##1$\normalsize
   \mathversion{normal}\par
 }%
 \def\Print##1\\{%
   \def\linebreak{\hfill\break}%
   \hangindent=2.5em\hangafter=0
   \leavevmode ##1\par}%
 \def\Out##1\\{%
   \def\linebreak{$\hfill\break\null\hfill$}%
   \kern\abovedisplayskip\par
   \hangindent=2.5em\hangafter=0
   \leavevmode
   \llap{\tiny\sffamily Out[\arabic{mmacnt}]=\kern.5em}
   \footnotesize$\displaystyle##1$\normalsize\hfill\null\par
   \kern\belowdisplayskip
 }%
 \def\Warning##1##2\\{%
   \def\linebreak{\hfill\break}%
   \hangindent=2.5em\hangafter=0
   \leavevmode
   {\scriptsize##1 : ##2}\par}%
}{%
 \par\smallskip
}

\usepackage{color}

\newenvironment{fshaded}{%
\MakeFramed {\FrameRestore}
}%
{\endMakeFramed}

\usepackage{tikz}
\usetikzlibrary{matrix}

\allowdisplaybreaks[4]

\begin{document}
\setlength{\baselineskip}{0.515cm}
\sloppy
\thispagestyle{empty}
\begin{flushleft}
DESY 12--210
\hfill 
\\
DO--TH 13/01\\
SFB/CPP-12-91\\
LPN 12-125\\
January 2013\\
\end{flushleft}

\mbox{}
\vspace*{\fill}
\begin{center}

{\LARGE\bf Analytic and Algorithmic Aspects of\\[0.4cm]
Generalized\hspace*{-0.1cm} Harmonic\hspace*{-0.1cm} Sums and Polylogarithms}

\vspace{4cm}
\large
Jakob Ablinger$^a$,  Johannes Bl\"umlein$^b$, and Carsten Schneider$^a$

\vspace{1.5cm}
\normalsize
{\it $^a$~Research Institute for Symbolic Computation (RISC),\\
                          Johannes Kepler University, Altenbergerstra\ss{}e 69,
                          A--4040, Linz, Austria}\\

\vspace*{3mm}
{\it  $^b$ Deutsches Elektronen--Synchrotron, DESY,}\\
{\it  Platanenallee 6, D-15738 Zeuthen, Germany}
\\

\end{center}
\normalsize
\vspace{\fill}
\begin{abstract}
\noindent In recent three--loop calculations of massive Feynman integrals within
Quantum Chromodynamics (QCD) and, e.g., in recent combinatorial problems the
so-called generalized harmonic sums (in short $S$-sums) arise. They are characterized
by rational (or real) numerator weights also different from $\pm 1$. In this article
we explore the algorithmic and analytic properties of these sums systematically. We
work out the Mellin and inverse Mellin transform which connects the sums under
consideration with the associated Poincar\'{e} iterated integrals, also called
generalized harmonic polylogarithms. In this regard, we obtain explicit analytic
continuations by means of asymptotic expansions of the $S$-sums which started to
occur frequently in current QCD calculations. In addition, we derive algebraic and
structural relations, like differentiation w.r.t. the external summation index and
different multi-argument relations, for the compactification of $S$-sum
expressions. Finally, we calculate algebraic relations for infinite $S$-sums,
or equivalently for generalized harmonic polylogarithms evaluated at special values.
The corresponding algorithms and relations are encoded in the computer algebra
package {\tt HarmonicSums}.
\end{abstract}

\vspace*{\fill}
\noindent
\numberwithin{equation}{section}
\newpage

\section{Introduction}
\label{sec:1}

\vspace*{1mm}
\noindent
In loop calculations of Feynman Diagrams in Quantum Field Theories special classes
of multiply nested sums and iterated integrals appear both in
intermediate steps and in the final results. E.g., this applies both to the massless
case and diagrams characterized by a single mass scale, or in cases two sufficiently
different mass scales appear and one may consider precise power expansions in
their ratio. We will focus in the following on single differential quantities. They
can be characterized either by a dimensional variable $x$ with $x \in \mathbb{R}$,
$0 \leq x \leq 1$ with exponential growth, or an integer variable $N$. Both
representations are equivalent and are related by a Mellin
transform\footnote{Here $f:\set R\to\set R$ can be distribution valued function,
cf.~\cite{YOSIDA}.}
\begin{eqnarray}
\label{eq:MEL1}
\Mvec[f(x)](n) := \int_0^1x^{n}f(x)\,dx.
\end{eqnarray}
The $n$-dependence may occur generically for Feynman diagrams with local
operator insertions which appear, e.g., within the light--cone expansion
\cite{Wilson:1969zs,Zimmermann:1970,Brandt:1972nw,Frishman:1971qn,Blumlein:1996vs}
in deep--inelastic scattering. As has been shown in
Refs.~\cite{Blumlein:2009ta,Blumlein:2010zv,Weinzierl:2010ps}
these Feynman integrals can always be represented by multiply nested sums. However,
in general the specific type of these sums is not easily predicted. It is assumed
that the classes of sums do steadily extend going to higher loops and/or allowing for
more masses and external off-shellness for a growing number of legs.

Systematic analyses
for a large class of single-scale processes at the 2-loop level in the massless case
revealed \cite{Blumlein:2005im,Blumlein:2006rr} that the results for all processes
could be
expressed in terms of six basic functions in Mellin-space only. Here one allows
as well to use argument-duplication and differential relations.
The corresponding functions are multiple harmonic sums.\footnote{In the physics
literature their use dates back to \cite{GonzalezArroyo:1979df}, see also
\cite{Floratos:1981hs,Mertig:1995ny}.}

For a non-negative integer $n$ and nonzero integers $a_i$ $(1\leq
i\leq k)$ the harmonic sums~\cite{Blumlein:1998if,Vermaseren:1998uu} are defined as
\begin{eqnarray}
\label{Equ:HarmonicSumsIntro}
\S{a_1,\ldots ,a_k}n= \sum_{n \geq i_1 \geq i_2 \geq \cdots \geq i_k \geq 1}
\frac{\sign{a_1}^{i_1}} {i_1^{\abs {a_1}}}\cdots
\frac{\sign{a_k}^{i_k}}{i_k^{\abs {a_k}}}.	
\end{eqnarray}
If $a_1\neq1$, the limit $n\to\infty$ exists and the values can be expressed
in terms of multiple zeta values, resp.\ Euler-Zagier values \cite{EZ,Zagier1994};
for further
properties and an extensive list of literature we refer to~\cite{Blumlein:2009cf}.
Likewise, more general classes of sums also result into new classes of special
numbers.

In physics applications harmonic sums emerge for the first time in case of the
fermion-fermion anomalous dimension (the Mellin transform of the corresponding
splitting function) \cite{Fermi:1924tc} as the classical harmonic sum $S_1(n)$ (also called harmonic numbers).
A first non-trivial nested harmonic sum appears at the 2--loop level with
$S_{-2,1}(n)$, \cite{GonzalezArroyo:1979df}.

Harmonic sums satisfy a quasi-shuffle algebra~\cite{Hoffman} and can be formulated
in terms of indefinite nested integrals.
The Mellin transform (\ref{eq:MEL1})
enables one to map harmonic sums to their integral representation in terms of
weighted harmonic polylogarithms, and the inverse transform allows one
to transform these integral representations back to harmonic
sums~\cite{Remiddi:1999ew}. These properties give rise to the analytic continuation
\cite{Blumlein:2000hw,Alekhin:2003ev,Blumlein:2005jg,Kotikov:2005gr,Blumlein:2007dj,
Blumlein:2009fz,Blumlein:2009ta} of harmonic sums. In particular, one can calculate the asymptotic expansion of harmonic
sums~\cite{Blumlein:2009fz,Blumlein:2009ta}; for the special case $a_i\geq1$ see also~\cite{Costermans:05}. Moreover, various relations between
harmonic sums are derived which allows to compactify harmonic sum expressions
coming from QCD calculations. Here in particular the algebraic
relations~\cite{Moch:2001zr,Blumlein:2003gb,Vermaseren:1998uu}
induced by the quasi-shuffle algebra and the structural relations
\cite{Blumlein:2009ta,Blumlein:2009fz} based on the differentiation of harmonic sums
contribute substantially.

Harmonic sums play a prominent role, e.g., in the field of combinatorics or number theory. Here the summation package~\texttt{Sigma}~\cite{Schneider:2007} based on difference field
theory~\cite{Karr1981,Schneider:2005,Schneider:2005b,Schneider:2005c,Schneider:2008} contributed to discover and prove identities such as~\cite{PS:03,CS:04,DPS,Apery,PS:07a}; as a special case these algorithms contain hypergeometric summation~\cite{PWZ:96}.
Meanwhile these algorithms are combined with the \texttt{HarmonicSums} package~\cite{Ablinger:12} (containing the analytic and algebraic properties summarized above) and lead to new packages
\cite{Ablinger:2010pb,Blumlein:2010zv,Blumlein:2012hg,Ablinger:2012ph} dealing with definite multi-sums.

In recent calculations based on these packages
new classes of functions have been arising in intermediate results but also in the final
result of given Feynman integrals. During such calculation but also for further
considerations of the given result it is crucial that the knowledge of these new function
classes is explored to a similar advanced level as the harmonic sums. For instance, the new
class of cyclotomic harmonic sums has been analyzed in detail in~\cite{Ablinger:2011te}.

In particular, the generalized harmonic sums or so-called $S$-sums
\cite{Moch:2001zr} occur in higher order calculations
\cite{Vermaseren:2005qc,Ablinger:2010ha,Ablinger:2010ty,
Ablinger:2012qm,Ablinger:2012sm,ABSW13}.

These sums also may emerge in
connection with cyclotomic index alphabets \cite{Ablinger:2011te} and nested binomial
and inverse binomial weights, cf. \cite{ABSW13,ABRS13}. In the context of statistics or number theory typical calculations can be found, e.g., in~\cite{PS:03,PS:07b,OS:08}. In case of pure harmonic
weights they are defined by
\begin{equation}\label{Equ:SSumsIntro}
	\S{a_1,\ldots ,a_k}{x_1,\ldots ,x_k;n}= \sum_{n\geq i_1 \geq i_2 \geq
        \cdots \geq i_k \geq 1} \frac{x_1^{i_1}}{i_1^{a_1}}\cdots
	\frac{x_k^{i_k}}{i_k^{a_k}},
\end{equation}
with positive integers $a_i$ and (usually) real numbers $x_i\neq0$ $(1\leq i\leq k)$. Note
that restricting $x_i=\pm1$ gives back the class of harmonic
sums~\eqref{Equ:HarmonicSumsIntro}, or in case that $x_i \in \{u\in\set C | u^n  = 1, n
\in \mathbb{N}\}$ the complex representations of the cyclotomic sums
\cite{Ablinger:2011te}. There are mutual applications of these functions
in loop calculations, for a survey see e.g.~\cite{Weinzierl:2006qs}.
Special cases of the $S$-sums are,
cf.~\cite{Moch:2001zr}~: the harmonic sums \cite{Vermaseren:1998uu,Blumlein:1998if} and
the multiple zeta values resp. Euler-Zagier values, cf. e.g.~\cite{Blumlein:2009cf},
the multiple polylogarithms \cite{GON1,Borwein:1999js}, two-dimensional harmonic
polylogarithms including the letters $\{0,1,-1,1-z,-z\}$
\cite{Gehrmann:2000zt,Gehrmann:2001jv}, the harmonic polylogarithms over the alphabet
$\{0,1,-1\}$ \cite{Remiddi:1999ew}, the Nielsen integrals over the alphabets
$\{0,1\}$, resp. $\{0,-1\}$
\cite{Nielsen1909,Kolbig:1969zza,Devoto:1983tc,Kolbig:1983qt}, and the classical
polylogarithms \cite{LEWIN:1958,LEWIN:1981,Devoto:1983tc,Zagier:2007}. Furthermore
\cite{Moch:2001zr}, the generalized hypergeometric functions
$_{P+1}F_P$
\cite{Klein:1906,Bailey:1935,Slater:1966,Erdelyi:1953,Exton:1976,Exton:1978}
the Appell-functions $F_{1,2}$  and the Kamp\'e De F\'eriet functions
\cite{Appell:1925,Appell:1926,Bailey:1935,Slater:1966,Erdelyi:1953,Exton:1976,Exton:1978}
can be represented in terms of $S$-sums. Code-implementations of four general
algorithms being described in \cite{Moch:2001zr} are given in
\cite{Weinzierl:2002hv,Moch:2005uc,Vollinga:2004sn}. In field theoretic calculations
the $\varepsilon$-expansion in the dimensional parameter is of importance and needs
to be solved for the $S$-sums and the analytic continuation in their parameters.
For special classes the $\varepsilon$-expansion has been studied in
Refs.~\cite{Moch:2001zr,
Weinzierl:2004bn,Kalmykov:2006pu,Kalmykov:2007pf,Huber:2005yg,Huber:2007dx,
Ablinger:2010ty,Blumlein:2010zv,Bytev:2011ks}.

In~\cite{Moch:2001zr} it is heavily exploited that these sums satisfy a quasi-shuffle
algebra. $S$--sums form a Hopf algebra
\cite{Hopf1941,MILMO1965,Sweedler1969,Kreimer:1997dp,Connes:1998qv,Cartier2006}, as has
been shown in Ref.~\cite{Moch:2001zr}, see also \cite{Duhr:2012fh}. The quasi-shuffle
property leads to new techniques
to expand certain definite nested sums over hypergeometric terms in terms of $S$-sums.
In addition, it has been worked out in~\cite{Borwein:1999js} that the $S$-sums at infinity
can be expressed in terms of Goncharov's multiple polylogarithms~\cite{GON1}; see
also~\cite{Borwein:1999js}. Moreover, relations of infinite $S$-sums with $x_i$ being
roots of unity have been explored in~\cite{Ablinger:2011te}.

In this article we aim at extending these results further to important technologies that
are currently available for harmonic sums. First we observe that $S$-sums can be expressed
in terms of a variant of multiple polylogarithms involving also the parameter $n$;
sending $n$
to infinity gives back the standard generalized polylogarithms. As a consequence,
$S$-sums can be analytically continued up to countable many poles.
This gives rise to
introduce Poincar\'{e}  iterated integrals~\cite{Poincare1884,LAPPO,CHEN} which
we call generalized harmonic polylogarithms and which extend the harmonic
polylogarithms
from~\cite{Remiddi:1999ew}. It will turn out that an extended version of the Mellin
transform
applied to weighted generalized polylogarithms can be expressed in terms of $S$-sums.
Conversely, the inverse Mellin transform enables one to express a certain class of $S$-sums
in terms of weighted generalized polylogarithms. In addition, we extend the ideas
of~\cite{Remiddi:1999ew} to calculate the series expansion of generalized polylogarithms
where
the coefficients of the expansion are given in terms of $S$-sums. This enables one to
transform and back-transform generalized polylogarithms evaluated at constants in terms
of infinite $S$-sums. In addition, we obtain algorithms to calculate the differentiation of
$S$-sums and calculate the asymptotic expansion of such sums. In general, we obtain complete
algorithms for $S$-sums~\eqref{Equ:SSumsIntro} with $x_i\in[-1,1]$ and $x_i\neq0$.
For arising calculations in QCD we succeeded in extending the algorithms for larger ranges
of $x_i$ by using, e.g, the integral representation of $S$-sums. Finally, we calculate
algebraic, structural and duplication relations for $S$-sums, e.g., for the alphabet
$x_i\in\{-\tfrac{1}{2},\tfrac{1}{2},1,-2,2\}$ to assist the compactification of $S$-sum
expression. This alphabet is of special importance in case of heavy flavor corrections
to deep-inelastic scattering at the 3--loop level. Results for related alphabets
are obtained analogously. In addition, inspired by the work in
Ref.~\cite{Blumlein:2009cf}
for multiple zeta values we deal
with new relations for infinite $S$-sums coming from the relations of finite sums, shuffle
relations from the generalized polylogarithms and duality relations based on the argument
transformations of the generalized polylogarithms.

All the underlying algorithms have been implemented in J.~Ablinger's Mathematica
package {\tt HarmonicSums}. For the available commands illustrated by examples we
refer to Appendix~\ref{App:HarmonicSums}, see also~\cite{Ablinger:12}.

\smallskip

The outline of this article is as follows. In Section~\ref{Sec:BasicProp} we define the
basic properties of $S$-sums with their integral representation and the closely connected
generalized polylogarithms.
In Section~\ref{SSRelatedArguments} we state important argument transformations of
generalized polylogarithms that will be used throughout this article. In
Section~\ref{Sec:PowerSeriesMPL} we show how one can calculate the coefficients of the
series expansion of generalized polylogarithms, and in Section~\ref{SSinfval} we use
this result to derive an algorithm that transforms in both directions generalized
polylogarithms evaluated at constants and $S$-sums at infinity.
In Section~\ref{Sec:AnalyticCont} we elaborate new aspects on the analytic continuation of $S$-sums and the special case of harmonic sums. In Section~\ref{SSMellin}
we define an appropriate Mellin transform for generalized polylogarithms and derive
algorithms that calculate the Mellin transform in terms of $S$-sums and that perform the
inverse Mellin transform in terms of generalized polylogarithms. In
Section~\ref{SSdifferentiation} we show how one can differentiate $S$-sums with the
(inverse) Mellin transform or with the integral representation. In
Section~\ref{Sec:Relation} we explore the algebraic relations induced by the quasi-shuffle
algebra, the structural relations based on the differentiation of $S$-sums and duplication
relations. In Section~\ref{Sec:InfiniteSSums} we consider relations of infinite $S$-sums.
Finally, we present techniques to calculate the asymptotic expansion of $S$-sums in
Section~\ref{Sec:AsymptoticExp}. In conclusion we discuss in
Section~\ref{Sec:QCDExample} examples of generalized harmonic sums emerging in massless
and massive 3-loop calculations and their representation in the complex $n$-plane. In the Appendix we summarize the available commands of the \texttt{HarmonicSums} package which supplement the presented algorithms.

\section{Basic properties of $S$-Sums and generalized polylogarithms}
\label{Sec:BasicProp}

\vspace*{1mm}
\noindent
Throughout this article let $\set K$ be a field containing the rational numbers $\set Q$, let $\set R^+$ be the positive real numbers, and let $\set N=\{1,2,\dots\}$ be the positive integers. For $a_i\in \N$ and\footnote{For a
set $A$ we define $A^*:=A\setminus\{0\}$.} $x_i\in \set K^*$ we define the class of \textit{generalized harmonic sums}~\cite{Moch:2001zr}, in short they are called
\textit{$S$-sums}, by
\begin{equation}\label{Equ:SSumDef}
	\S{a_1,\ldots ,a_k}{x_1,\ldots ,x_k;n}= \sum_{i_1=1}^n\frac{x_1^{i_1}}{i_1^{a_1}}
\sum_{i_2=1}^{i_1}\frac{x_2^{i_2}}{i_2^{a_2}}\dots\sum_{i_k=1}^{i_{k-1}}\frac{x_k^{i_k}}
{i_k^{a_k}};
\end{equation}
for a different writing see~\eqref{Equ:SSumsIntro}. $k$ is called the depth and
$w=\sum_{i=0}^ka_i$ is called the weight of the $S$-sum $\S{a_1,\ldots ,a_k}{x_1,\ldots
,x_k;n}$.

Restricting the $S$-sums to the situation $x_i=\pm 1$ yields to the well understood class
of harmonic sums~\cite{Blumlein:1998if,Vermaseren:1998uu}. For harmonic sums the limit
$n\to\infty$ of~\eqref{Equ:SSumDef} exists if and only if $a_1\neq 1$ or $x_1=-1$;
see~\cite{Zagier1994}. For $S$-sums this result extends as follows; for a detailed proof
see~\cite{Ablinger:12}.
\begin{thm}
Let $a_1, a_2, \ldots a_k \in \N$ and $x_1, x_2, \ldots x_k \in \R^*$ for $k \in \N.$
The $S$-sum $\S{a_1,a_2,\ldots,a_k}{x_1,x_2,\ldots,x_k;n}$ is absolutely convergent,
when
$n\rightarrow \infty$, if and only if one of the following conditions holds:
\begin{itemize}
 \item [1.] $\abs{x_1}<1 \wedge \abs{x_1 x_2}\leq 1 \wedge \ldots \wedge \abs{x_1 x_2
\cdots x_k}\leq 1$
 \item [2.] $a_1>1 \wedge \abs{x_1}=1 \wedge \abs{x_2}\leq 1 \wedge \ldots \wedge
\abs{x_2 \cdots x_k}\leq 1.$
\end{itemize}
In addition the $S$-sum is conditional convergent (convergent but not absolutely
convergent)
if and only if
\begin{itemize}
 \item [3.] $a_1=1 \wedge x_1=-1 \wedge \abs{x_2}\leq 1 \wedge \ldots \wedge
\abs{ x_2 \cdots x_k}\leq 1.$
\end{itemize}
\label{SSconsumthm}
\end{thm}
\noindent
If a $S$-sum is finite, i.e., the limit $n\to\infty$ exists, we define
$$\S{a_1,\ldots ,a_k}{x_1,\ldots ,x_k;\infty}:=\lim_{n\rightarrow \infty}
\S{a_1,\ldots ,a_k}{x_1,\ldots ,x_k;n}.$$

\noindent An other crucial property is that for the product of two $S$-sums with the
same
upper summation limit the following quasi-shuffle relation
holds~(see~\cite{Moch:2001zr}): for $n\in \N,$
\begin{eqnarray}
	&&\S{a_1,\ldots ,a_k}{x_1,\ldots ,x_k;n}\S{b_1,\ldots ,b_l}{y_1,\ldots ,y_l;n}=\nonumber\\
	&&\hspace{2cm}\sum_{i=1}^n \frac{a_1^i}{i^{a_1}}\S{a_2,\ldots ,a_k}{x_2,\ldots ,x_k;i}\S{b_1,\ldots ,b_l}{y_1,\ldots ,y_l,i} \nonumber\\
	&&\hspace{2cm}+\sum_{i=1}^n \frac{b_1^i}{i^{\abs {b_1}}}\S{a_1,\ldots ,a_k}{x_1,\ldots ,x_k,i}\S{b_2,\ldots ,b_l}{y_2,\ldots ,y_l;i} \nonumber\\
	&&\hspace{2cm}-\sum_{i=1}^n \frac{(a_1\cdot b_1)^i}{i^{a_1+b_1}}\S{a_2,\ldots ,a_k}{x_2,\ldots ,x_k,i}\S{b_2,\ldots ,b_l}{y_2,\ldots ,y_l;i}.
	\label{SSsumproduct}
\end{eqnarray}

Since the right hand side consists of $S$-sums where at least one component of the product has smaller depth, this formula applied iteratively leads to a linear combination of $S$-sums with integer coefficients. The underlying quasi-shuffle algebra induced by this product property will be exploited further in Section~\ref{Sec:Relation}.

Finally, we will exploit heavily the fact that each $S$-sum can be written as an indefinite nested integral. Namely, for $\S{1}{b;n}$ with $b\in\R^*$ we have
$$\S{1}{b;n}=\int_0^{b}{\frac{x_1^n-1}{x_1-1}dx_1}$$
and by induction on $m\geq1$
we obtain
$$\S{m}{b;n}=\int_0^b{\frac{1}{x_m}\int_0^{x_m}{\frac{1}{x_{m-1}} \cdots \int_0^{x_3}{\frac{1}{x_2}\int_0^{x_2}{\frac{x_1^n-1}{x_1-1}dx_1}dx_2}\cdots dx_{m-1}}dx_m}.$$

Now consider, e.g., the sum $\S{2,1}{a,b;n}$ with $a,b\in\R^*.$ Due to the integral representation for $S$-sums with depth 1 we obtain
\begin{eqnarray*}
 \S{2,1}{a,b;n}&=&\sum_{j=1}^n\frac{a^j}{j^2}\S{1}{b;j}\\
	&=&\sum_{j=1}^n\frac{a^j}{j^2}\int_{0}^b\frac{x^j-1}{x-1}dx \\
	&=&\int_{0}^b{\frac{1}{x-1}\sum_{j=1}^n\frac{a^j}{j^2}\left(x^j-1\right)dx}\\
	&=&\int_{0}^b{\frac{1}{x-1}\left(\S{2}{a\;x;n}-\S{2}{a;k}\right)dx}\\	 &=&\int_{0}^b{\frac{1}{x-1}\int_a^{ax}{\frac{1}{y}\int_0^y{\frac{z^n-1}{z-1}dz}dy}dx}\\	 &=&\int_0^{ab}{\frac{1}{x-a}\int_a^{x}{\frac{1}{y}\int_0^y{\frac{z^n-1}{z-1}dz}dy}dx}.
\end{eqnarray*}

\noindent Looking at this example one arrives straightforwardly at the following theorem; the proof follows by induction on $k$ and the integral representations of $S$-sums with depth $k=1$.
\begin{thm}\label{SSintrep}
Let $m_i\in\N,$ $b_i\in\R^*$ and define
\begin{eqnarray*}
&&I_{m_1,m_2,\ldots,m_k}(b_1,b_2,\ldots,b_k;z)=\\
&&\hspace{0.4cm}\int_0^{b_1\cdots b_k}{\frac{dx_{k}^{m_k}}{x_{k}^{m_k}}\int_0^{x_{k}^{m_k}}{\frac{dx_{k}^{m_k-1}}{x_{k}^{m_k-1}} \cdots
\int_0^{x_{k}^3}{\frac{dx_{k}^2}{x_{k}^2}\int_0^{x_{k}^2}{\frac{dx_{k}^1}{x_{k}^1-b_1\cdots b_{k-1}}}}}}\\
&&\hspace{0.4cm}\int^{x_{k}^1}_{b_1\cdots b_{k-1}}{\frac{dx_{k-1}^{m_{k-1}}}{x_{{k-1}}^{m_{k-1}}}\int_0^{x_{{k-1}}^{m_{k-1}}}{\frac{dx_{{k-1}}^{m_{k-1}-1}}{x_{{k-1}}^{m_{k-1}-1}} \cdots
\int_0^{x_{{k-1}}^3}{\frac{dx_{{k-1}}^2}{x_{{k-1}}^2}}}}\int_0^{x_{{k-1}}^2}\hspace{-0.4em}{\frac{dx_{{k-1}}^1} {x_{{k-1}}^1-b_1\cdots b_{k-2}}}\\
&&\hspace{0.4cm}\int^{x_{{k-1}}^1}_{b_1\cdots b_{k-2}}{\frac{dx_{{k-2}}^{m_{k-2}}}{x_{{k-2}}^{m_{k-2}}}\int_0^{x_{{k-2}}^{m_{k-2}}}{\frac{dx_{{k-2}}^{m_{k-2}-1}}{x_{{k-2}}^{m_{k-2}-1}} \cdots
\int_0^{x_{{k-2}}^3}{\frac{dx_{{k-2}}^2}{x_{{k-2}}^2}}}}\int_0^{x_{{k-2}}^2}\hspace{-0.4em}{\frac{dx_{{k-2}}^1} {x_{{k-2}}^1-b_1\cdots b_{k-3}}}\\
&&\vspace{0.1cm}\\
&&\hspace{0cm}\hbox to 0.4\textwidth{}\vdots\\
&&\vspace{0.1cm}\\
&&\hspace{0.4cm}\int^{x_{4}^1}_{b_1b_2b_3}{\frac{dx_{3}^{m_3}}{x_{3}^{m_3}}\int_0^{x_{3}^{m_3}}{\frac{dx_{3}^{m_3-1}}{x_{3}^{m_3-1}} \cdots \int_0^{x_{3}^3}{\frac{dx_{3}^2}{x_{3}^2}\int_0^{x_{3}^2}{\frac{dx_{3}^1}{x_{3}^1-b_1b_2}}}}}\\
&&\hspace{0.4cm}\int^{x_{3}^1}_{b_1b_2}{\frac{dx_{2}^{m_2}}{x_{2}^{m_2}}\int_0^{x_{2}^{m_2}}{\frac{dx_{2}^{m_2-1}}{x_{2}^{m_2-1}} \cdots \int_0^{x_{2}^3}{\frac{dx_{2}^2}{x_{2}^2}\int_0^{x_{2}^2}{\frac{dx_{2}^1}{x_{2}^1-b_1}}}}}\\
&&\hspace{0.4cm}\int^{x_{2}^1}_{b_1}
{\frac{dx_{1}^{m_1}}{x_{1}^{m_1}}\int_0^{x_{1}^{m_1}}{\frac{dx_{1}^{m_1-1}}{x_{1}^{m_1-1}}
\cdots \int_0^{x_{1}^3}{\frac{dx_{1}^2}{x_{1}^2}\int_0^{x_{1}^2}{\frac{\left({x_{1}^1}
\right)^z-1}{x_{1}^1-1}dx_{1}^1}}}}.
\end{eqnarray*}
Then for all $n\in\set N$,
$$\S{m_1,m_2,\ldots,m_k}{b_1,b_2,\ldots,b_k;n}=I_{m_1,m_2,\ldots,m_k}(b_1,b_2,\ldots,b_k;n).$$
\end{thm}

\noindent Note that sending $n\to\infty$ we obtain back multiple polylogarithms stated,
e.g., in~\cite{GON1,Borwein:1999js,Moch:2001zr}.

For later constructions based on analytic reasoning (see also Section~\ref{Sec:AnalyticCont}), we will use the fact that one can transform $S$-sums to a particular form of integral representations, the so-called generalized harmonic polylogarithms using the inverse Mellin transform; see
Section~\ref{SSMellin}. Conversely, with the Mellin transform one can express a weighted generalized harmonic polylogarithm in terms of $S$-sums.
In other words, the $S$-sums are strongly twisted with generalized harmonic polylogarithms which can be defined as follows.

For $a\in\R$ and
\begin{eqnarray*}
q=\left\{
		\begin{array}{ll}
				a, &  \textnormal{if }a>0  \\
				\infty, & \textnormal{otherwise}
		\end{array} \right.
\end{eqnarray*}
we introduce the auxiliary function $f_a:(0,q)\mapsto \R$ by
$$f_a(x)=\left\{
		\begin{array}{ll}
				\frac{1}{x}, &  \textnormal{if }a=0  \\
				\frac{1}{\abs{a}-\sign{a}\,x}, & \textnormal{otherwise}.
		\end{array}
		\right.$$
Then the \textit{generalized harmonic polylogarithms} (in short generalized
polylogarithm) \textnormal{H} of the word $m_1,\dots,m_w$ (resp.\ vector
$\ve m=(m_1,\dots,m_w)$) with $m_i \in \R$ is defined as follows.
For $x\in (0,q)$ with $q:=\displaystyle{\min_{m_i>0}{m_i}}$,
\begin{eqnarray}
\H{}{x}&=&1,\nonumber\\
\H{m_1,m_{2},\ldots,m_w}{x} &=&\left\{
		  	\begin{array}{ll}
			\frac{1}{w!}(\log{x})^w,&
\textnormal{if }(m_1,\ldots,m_w)=(0,\ldots,0)\\
\int_0^x{f_{m_1}(y) \H{m_{2},\ldots,m_w}{y}dy},& \textnormal{otherwise}.
					\end{array} \right.  \nonumber
\end{eqnarray}
Generalized polylogarithms are also dealt with in Refs.~\cite{
Brown:2008um,Brown:2009ta,Bogner:2012dn,Ablinger:2012qm,Ablinger:2012ej} as
hyperlogarithms.
The length $w$ of the vector $\ve m$ is called the weight of the generalized polylogarithm $\H{\ve m}x.$

Restricting $\H{m_1,m_{2},\ldots,m_w}{x}$ to $m_i\in\{-1,0,1\}$ gives the class of
harmonic polylogarithms~\cite{Remiddi:1999ew}.
Typical examples are
\begin{eqnarray*}
\H{1}x&=&\int_{0}^x\frac{1}{1-x_1}dx_1=-\log(1-x)\\
\H{-2}x&=&\int_{0}^x\frac{1}{2+x_1}dx_1=\log(2+x)-\log(2)\\
\H{2,0,-\frac{1}{2}}x&=&\int_{0}^x\frac{1}{2-x_1}\int_{0}^{x_1}\frac{1}{x_2}\int_{0}^{x_2}\frac{1}{\frac{1}{2}+x_3}dx_3dx_2dx_1.
\end{eqnarray*}

A generalized polylogarithm $\H{\ve m}x=\H{m_1,\ldots,m_w}x$ with $q:=\min_{m_i>0}{m_i}$ is an analytic function for $x\in (0,q).$ For the limits $x\rightarrow 0$ and  $x\rightarrow q$ the following holds.
From the definition we have that $\H{\ve m}0~=~0$ if and only if $\ve m\neq \ve 0_w.$
In addition, $\H{\ve m}q$ is finite if and only if one of the following cases holds
    \begin{itemize}
      \item $m_1\neq q$
      \item $w\geq2$, $m_1=1$ and $m_v=0$ for all $v$ with $1<v\leq w.$
    \end{itemize}
We define $\H{\ve m}0:=\lim_{x\rightarrow 0^+} \H{\ve m}x$ and $\H{\ve m}1:=\lim_{x\rightarrow 1^-} \H{\ve m}x$ if the limits exist.\\
Note that for generalized polylogarithms of the form $\H{m_1,\ldots,m_k,1,0\ldots,0}x$ with $c:=\min_{m_i>0}{m_i}$ (and $c>1$) we can extend the definition range from
$x\in (0,1)$ to $x \in (0,c).$ These generalized polylogarithms are analytic functions for $x \in (0,c).$ The limit $x\rightarrow c$ exists if and only if $m_1\neq c.$

Summarizing, a generalized polylogarithm $\H{m_1,m_2,\ldots,m_w}x$ with $q:=\min_{m_i>0}{m_i}$ is finite at $c\in\set R$ for $c\geq0$ if one of the following cases hold:
\begin{enumerate}
\item $c<q$
\item $c=q$ and $m_1\neq q$
\item $(m_1,\ldots,m_w)=(m_1,\ldots,m_k,1,0,\dots,0)$ with $q':=\min_{m_i>0}{m_i}$ and $c<q'$.
\end{enumerate}
Throughout this article, whenever we state that $\H{m_1,m_2,\ldots,m_w}c$ is finite or well defined, we assume that one of the three conditions from above holds.

For the derivatives we have for all $x\in (0,\min_{m_i>0}{m_i})$ that $$ \frac{d}{d x} \H{\ve m}{x}=f_{m_1}(x)\H{m_{2},m_{3},\ldots,m_w}{x}. $$
Besides that, the product of two generalized polylogarithms of the same argument can be expressed using the formula
\begin{equation}
\label{SShpro}
\H{\ve p}x\H{\ve q}x=\sum_{\ve r= \ve p \shuffle \ve q}\H{\ve r}x;
\end{equation}
here the shuffle operation $\ve p \shuffle \ve q$ denotes all merges of $\ve p$ and $\ve q$ in which the relative orders of the elements of $\ve p$ and $\ve q$ are preserved. Thus, similarly to $S$-sums, one can write any product of generalized polylogarithms by a linear combination of generalized polylogarithms over the integers.

As utilized in~\cite{Remiddi:1999ew} for harmonic polylogarithms, we use this shuffle
product to
remove trailing or leading constants in the vector of the arising generalized polylogarithms. E.g., for $a\in\set R$ the shuffle product gives
\begin{align*}
\H{a}x\H{m_1,\ldots,m_w}{x}=&\H{a,m_1,\ldots,m_w}{x} + \H{m_1,a,m_{2},\ldots,m_w}{x}\\
&+\H{m_1,m_{2},a,m_{3},\ldots,m_w}{x} + \cdots + \H{m_1,\ldots,m_w,a}{x},
\end{align*}
or equivalently
\begin{align*}
\H{m_1,\ldots,m_w,a}{x} =& \H{a}{x}\H{m_1,\ldots,m_w}{x} - \H{a,m_1,\ldots,m_w}{x}\\
&-\H{m_1,a,m_{2},\ldots,m_w}{x} -\cdots- \H{m_1,\ldots,a,m_w}{x}.
\end{align*}
If $m_w$ is $a$ as well, we can move the last term to the left and can divide by two. This leads to
\begin{eqnarray}
\H{m_1,\ldots,a,a}{x} &=& \frac{1}{2}\left(\H{a}{x}\H{m_1,\ldots,m_{w-1},a}{a} - \H{a,m_1,\ldots,m_{w-1},a}{x}\right.\nonumber\\
&&\left.-\H{m_1,a,m_{2},\ldots,m_{w-1},a}{x} -\cdots- \H{m_1,\ldots,a,m_{w-1}}{x}\right).\nonumber
\end{eqnarray}
Applying this tactic iteratively leads to a polynomial representation in $H_a(x)$ with coefficients in terms of generalized polylogarithms with no trailing $a$'s. In particular, if a generalized polylogarithm is analytic for $(0,q)$ then also the occurring generalized polylogarithms of the linear combination are analytic for $(0,q)$. A typical example for $a=0$ is, e.g.,
\begin{eqnarray*}
\textnormal{H}_{2,-3,0,0}(x)&=&\frac{1}{2} \textnormal{H}_0(x){}^2 \textnormal{H}_{2,-3}(x)-\textnormal{H}_0(x) \textnormal{H}_{0,2,-3}(x)-\textnormal{H}_0(x) \textnormal{H}_{2,0,-3}(x)\\
&&+\textnormal{H}_{0,0,2,-3}(x)+\textnormal{H}_{0,2,0,-3}(x)+\textnormal{H}_{2,0,0,-3}(x).
\end{eqnarray*}
Note that there is one extra complication: given a
generalized polylogarithm in the form $\H{m_1,\ldots,m_k,1,0\ldots,0}x$
with $1<q:=\min_{m_i>0}{m_i}$ then we cannot use the strategy to extract trailing
zeroes if $x\geq1$ since this would lead
to infinities: The shuffle algebra would suggest to rewrite $\H{2,1,0}{x}$ being analytic for $x\in(0,2)$ as
$\H{0}x\H{2,1}{x}-\H{0,2,1}{x}-\H{2,0,1}{x},$
however for $x>1$ for instance $\H{2,0,1}{x}$ is not defined.  In later considerations these difficulties have to be handled separately.

Completely analogously, leading constants can be removed in generalized polylogarithms.
In addition, using the quasi shuffle algebra~\eqref{SSsumproduct} it is possible to
extract leading (or trailing) indices of a $S$-sum. For example we have for $a\in\N$ and $b\in\set K^*$
\begin{align*}
\textnormal{S}_{a,2,1}\left(b,\tfrac{1}{2},\tfrac{1}{3};n\right)=&\textnormal{S}_a(b;n) \textnormal{S}_{2,1}\left(\tfrac{1}{2},\tfrac{1}{3};n\right)+\textnormal{S}_{2+a,1}\left(\tfrac{b}{2},\tfrac{1}{3};n\right)+\textnormal{S}_{2,1+a}\left(\tfrac{1}{2},\tfrac{b}{3};n\right)\\	 &-\textnormal{S}_{2,a,1}\left(\tfrac{1}{2},b,\tfrac{1}{3};n\right)-\textnormal{S}_{2,1,a}\left(\tfrac{1}{2},\tfrac{1}{3},b;n\right).
\end{align*}
For $S$-sums $\S{a_1,\ldots,a_k}{b_1,\ldots,b_k;n}$ with $a_i\in\N$ and $b_i\in
[-1,1]\setminus\{0\}$ this can be used to get a polynomial representation in $\S{1}n$
with coefficients in terms of convergent $S$-sums, like, e.g.,
\begin{align*}
 \textnormal{S}_{1,1,2}\left(1,1,\tfrac{1}{2};n\right)=&\tfrac{1}{2} \textnormal{S}_1(n){}^2 \textnormal{S}_2\left(\tfrac{1}{2};n\right)+\textnormal{S}_1(n) \left(\textnormal{S}_3\left(\tfrac{1}{2};n\right)-\textnormal{S}_{2,1}\left(\tfrac{1}{2},1;n\right)\right)\\
					&+\tfrac{1}{2} \textnormal{S}_4\left(\tfrac{1}{2};n\right)-\tfrac{1}{2} \textnormal{S}_{2,2}\left(\tfrac{1}{2},1;n\right)+\tfrac{1}{2}\textnormal{S}_{2,2}\left(1,\tfrac{1}{2};n\right)\\
&-\textnormal{S}_{3,1}\left(\tfrac{1}{2},1;n\right)+\textnormal{S}_{2,1,1}\left(\tfrac{1}{2},1,1;n\right).
\end{align*}
For general $S$-sums this strategy usually fails to extract the divergency, \eg in
\begin{align*}
\textnormal{S}_{1,2,1}\left(5,\tfrac{1}{2},\tfrac{1}{3};n\right)=&\textnormal{S}_1(5;n) \textnormal{S}_{2,1}\left(\tfrac{1}{2},\tfrac{1}{3};n\right)+\textnormal{S}_{2,2}\left(\tfrac{1}{2},\tfrac{5}{3};n\right)+\textnormal{S}_{3,1}\left(\tfrac{5}{2},\tfrac{1}{3};n\right)\\	 &-\textnormal{S}_{2,1,1}\left(\tfrac{1}{2},\tfrac{1}{3},5;n\right)-\textnormal{S}_{2,1,1}\left(\tfrac{1}{2},5,\tfrac{1}{3};n\right),
\end{align*}
the sums $\textnormal{S}_1(5;n),\textnormal{S}_{3,1}\left(\tfrac{5}{2},\tfrac{1}{3};n\right)$ and $\textnormal{S}_{2,1,1}\left(\tfrac{1}{2},5,\tfrac{1}{3};n\right)$ are divergent.

\section{Identities between Multiple Polylogarithms of Related Arguments}
\label{SSRelatedArguments}

\vspace*{1mm}
\noindent
In the following we want to look at several special transforms of the argument of generalized polylogarithms which will be useful for later considerations.
In general, we start with a defined generalized polylogarithm $f(x)=\H{m_1,\ldots, m_w}{\frac{a x+b}{cx+d}}$ with particular chosen integers $a,b,c,d$ and $x\in\set R$ and derive representations in terms of generalized polylogarithms in $x$. Since such a transformation holds not only for a point $x\in\set R$, but in a region around $x$, all the transformations have the following additional property: if the input expression $f(x)$ is analytic in a certain region around $x$, also the arising generalized polylogarithms of the transformed representation are analytic in this region.

\subsection{$x+b\rightarrow x$}
\label{SStransformplusxplusb}

\begin{lemma}
Let $m_i \in \R\setminus (0,1)$ and $x>0$ such that the generalized polylogarithm $\H{m_1,\ldots, m_l}{x+1}$ is defined. If $(m_1,\ldots,m_l)\neq (1,0,\ldots,0),$
\begin{align*}
\H{m_1,\ldots, m_l}{x+1}=&
\H{m_1,\ldots, m_l}{1}+\H{m_1-1}x\H{m_2,\ldots, m_l}{1}+\\
&+\cdots+\H{m_1-1,\ldots, m_{l-1}-1}{x}\H{m_l}1+\H{m_1-1,\ldots, m_l-1}{x}.
\intertext{If $(m_1,\ldots,m_l)=(1,0,\ldots,0),$}
\H{m_1,\ldots,m_l}{x+1}=&\H{1,0,\ldots, 0}{1}-\H{0,-1,\ldots, -1}{x}.
\end{align*}
\label{SSeinplustrafo}
\end{lemma}
\begin{proof}
We proceed by induction on $l.$ For $l=1$  and $m_1>1$ we have:
\begin{eqnarray*}
\H{m_1}{x+1}&=&\int_0^{x+1}\frac{1}{m_1-u}du=\H{m_1}1+\int_1^{x+1}\frac{1}{m_1-u}du\\
&=&\H{m_1}1+\int_0^{x}\frac{1}{m_1-(u+1)}du=\H{m_1}1+\H{m_1-1}x.
\end{eqnarray*}
For $l=1$  and $m_1<0$ we have:
\begin{eqnarray*}
\H{m_1}{x+1}&=&\int_0^{x+1}\frac{1}{\abs{m_1}+u}du=\H{m_1}1+\int_1^{x+1}\frac{1}{\abs{m_1}+u}du\\
&=&\H{m_1}1+\int_0^{x}\frac{1}{\abs{m_1}+(u+1)}du=\H{m_1}1+\H{m_1-1}x.
\end{eqnarray*}
For $l=1$  and $m_1 = 0$ we have:
\begin{eqnarray*}
\H{0}{x+1}&=&\H{0}1+\int_1^{x+1}\frac{1}{u}du=\int_0^{x}\frac{1}{1+u}du=\H{0}1+\H{-1}{x}.
\end{eqnarray*}
Suppose the theorem holds for $l-1.$ If $m_1\neq 0$ we get
\begin{eqnarray*}
\H{m_1,\ldots, m_l}{x+1}&=&\H{m_1,\ldots, m_l}{1}+\int_1^{x+1}{\frac{\H{m_2,\ldots,m_l}{u}}{\abs{m_1}-\sign{m_1}u}}du\\
&=&\H{m_1,\ldots, m_l}{1}+\int_0^{x}{\frac{\H{m_2,\ldots,m_l}{u+1}}{\abs{m_1}-\sign{m_1}{(u+1)}}}du\\
&=&\H{m_1,\ldots, m_l}{1}+\int_0^{x}\frac{1}{\abs{m_1}-\sign{m_1}{(u+1)}}\biggl(\H{m_2,\ldots, m_l}{1}\biggr.\\
&&+\H{m_2-1}x\H{m_3,\ldots, m_l}{1}+\H{m_2-1,m_3-1}x\H{m_4,\ldots, m_l}{1}+\cdots\\
&&\biggl.+\H{m_2-1,\ldots, m_{l-1}-1}{x}\H{m_l}1+\H{m_2-1,\ldots, m_l-1}{x}\biggr)du\\
&=&\H{m_1,\ldots, m_l}{1}+\H{m_1-1}x\H{m_2,\ldots, m_l}{1}\\
&&+\H{m_1-1,m_2-1}x\H{m_3,\ldots, m_l}{1}+\cdots+\\
&&+\H{m_1-1,\ldots, m_{l-1}-1}{x}\H{m_l}1+\H{m_1-1,\ldots, m_l-1}{x}.
\end{eqnarray*}
If $m_1= 0$ we get
\begin{eqnarray*}
\H{0,m_2,\ldots, m_l}{x+1}&=&\H{0,m_2,\ldots, m_l}{1}+\int_1^{x+1}{\frac{\H{m_2,\ldots,m_l}{u}}{u}}du\\
&=&\H{0,m_2,\ldots, m_l}{1}+\int_0^{x}{\frac{\H{m_2,\ldots,m_l}{u+1}}{(1+u)}}du\\
&=&\H{0,m_2,\ldots, m_l}{1}+\int_0^{x}\frac{1}{1+u}\biggl(\H{m_2,\ldots, m_l}{1}+\biggr.\\
&&+\H{m_2-1}x\H{m_3,\ldots, m_l}{1}+\H{m_2-1,m_3-1}x\H{m_4,\ldots, m_l}{1}+\\
&&\biggl.+\cdots+\H{m_2-1,\ldots, m_{l-1}-1}{x}\H{m_l}1+\H{m_2-1,\ldots, m_l-1}{x}\biggr)du\\
&=&\H{0,m_2,\ldots, m_l}{1}+\H{-1}x\H{m_2,\ldots, m_l}{1}\\
&&+\H{-1,m_2-1}x\H{m_3,\ldots, m_l}{1}+\cdots+\\
&&+\H{-1,\ldots, m_{l-1}-1}{x}\H{m_l}1+\H{-1,\ldots, m_l-1}{x}.
\end{eqnarray*}
\end{proof}

The proofs of the following three lemmas are similar to the proof of the previous one, hence we will omit them here.
\begin{lemma}
Let $b>0$, $m_1\in \R \setminus (0,b]$, $m_i \in \R \setminus (0,b)$ for $i\in \{2,3,\ldots l\}$ such that $(m_j,\ldots,m_l)\neq(1,0,\ldots,0)$ for all $j\in \{1,\ldots l\}.$
If $\H{m_1,\ldots, m_l}{x+b}$ is defined for $x>0$,
\begin{eqnarray*}
\H{m_1,\ldots, m_l}{x+b}&=&\H{m_1,m_2,\ldots, m_l}{b}+\H{m_1-b}x\H{m_2,\ldots, m_l}{b}\\
&&+\H{m_1-b,m_2-b}x\H{m_3,\ldots, m_l}{b}+\cdots+\\
&&+\H{m_1-b,\ldots, m_{l-1}-b}{x}\H{m_l}b+\H{m_1-b,\ldots, m_l-b}{x}.
\end{eqnarray*}
\label{SStransformplusa1}
\end{lemma}

\begin{lemma}
Let $(m_1,\ldots,m_l)=(1,0,\ldots,0)$ and let $x>0$ such that $\H{m_1,\ldots, m_l}{x+b}$ is defined. For $1>b>0$ and $1-b>x>0$,
\begin{align*}
\H{m_1,\ldots, m_l}{x+b}=&\H{m_1,m_2,\ldots, m_l}{b}+\H{m_1-b}x\H{m_2,\ldots, m_l}{b}+\H{m_1-b,m_2-b}x\H{m_3,\ldots, m_l}{b}\\
&+\cdots+\H{m_1-b,\ldots, m_{l-1}-b}{x}\H{m_l}b+\H{m_1-b,\ldots, m_l-b}{x}.
\intertext{For $x>0$ and $b\geq 1$,}
\H{m_1,\ldots, m_l}{x+b}
=&\H{m_1,m_2,\ldots, m_l}{b}-\H{m_1-b}x\H{m_2,\ldots, m_l}{b}-\H{m_1-b,m_2-b}x\H{m_3,\ldots, m_l}{b}\\
&-\cdots-\H{m_1-b,\ldots, m_{l-1}-b}{x}\H{m_l}b-\H{m_1-b,\ldots, m_l-b}{x}.
\end{align*}
\label{SStransformplusa2}
\end{lemma}

\begin{lemma}
Let $b>0$, $m_1\in \R \setminus (0,b]$, $m_i \in \R \setminus (0,b)$ for $i\in \{2,3,\ldots k\}$ and $(m_{k+1},\ldots,m_l)= (1,0,\ldots,0).$ Let $x>0$ such that $\H{m_1,\ldots, m_l}{x+b}$ is defined. For $1-b>x>0$ and $1>b>0$,
\begin{align*}
H_{m_1,\ldots, m_l}&(x+b)=\H{m_1,m_2,\ldots, m_l}{b}+\H{m_1-b}x\H{m_2,\ldots, m_l}{b}+\cdots+\\
&+\H{m_1-b,\ldots, m_{l-1}-b}{x}\H{m_l}b+\H{m_1-b,\ldots, m_l-b}{x}.
\intertext{For $b\geq 1$,}
H_{m_1,\ldots, m_l}&(x+b)=
\H{m_1,m_2,\ldots, m_l}{b}+\H{m_1-b}x\H{m_2,\ldots, m_l}{b}+\cdots+\\
&+\H{m_1-b,\ldots, m_{k}-b}{x}\H{m_{k+1},\ldots,m_l}b-\Big[
\H{m_1-b,\ldots, m_{k+1}-b}{x}\H{m_{k+2},\ldots,m_l}b+\\
&+\cdots+\H{m_1-b,\ldots, m_{l-1}-b}{x}\H{m_l}b-\H{m_1-b,\ldots, m_l-b}{x}\Big].
\end{align*}
\label{SStransformplusa3}
\end{lemma}

\subsection{$b-x\rightarrow x$}
\label{SSbxx}

In this Subsection we assume that $b>0$ and we consider indices $m_i\in \R\setminus
(0,b)$. The following method transforms a defined generalized polylogarithm $H_{m_1,\dots,m_w}(b-x)$ for some $x\in\set R$ in terms of generalized polylogarithm with argument $x$ and evaluations at constants.
Proceeding recursively on the weight $w$ of the generalized polylogarithm we have
\begin{equation}\label{SStrafobx1}
\begin{split}
 \H{m_1}{b-x}&=\H{m_1}b-\H{b-m_1}x\quad \text{for }m_1\neq b,\\
 \H{b}{b-x}&=\H{0}x-\H{0}b.
\end{split}
\end{equation}
Now let us look at higher weights $w>1$. We consider $\H{m_1,m_2,\ldots,m_w}{b-x}$ with $m_i\in \R\setminus(0,b)$ and suppose that we can already apply the transformation for generalized polylogarithms
 of weight $<w$. If $m_1=b,$ we can remove leading $b's$ (see
Section~\ref{Sec:BasicProp}) and end up with generalized polylogarithms without leading $b's$ and powers of $\H{b}{b-x}.$ For the powers of $\H{b}{b-x}$ we
 can use (\ref{SStrafobx1}); therefore, only the cases in which the first index $m_1\neq b$ are to be considered. For $b\neq 0$:
\begin{eqnarray*}
\H{m_1,m_2,\ldots,m_w}{b-x}&=&\H{m_1,m_2,\ldots,m_w}b-\int_0^x\frac{\H{m_2,\ldots,m_w}{b-t}}{\abs{-m_1+b}-\sign{-m_1+b}t}dt.
\end{eqnarray*}
Since we know the transform for weights $<w,$ we can apply it to $\H{m_2,\ldots,m_w}{b-t}$. Linearizing the expression with the shuffle product using the definition of the generalized polylogarithm we finally obtain the required weight $w$ identity.
With the steps above one obtains, e.g.,
\begin{equation*}
\textnormal{H}_{3,-2,-1}(2-x) =-\textnormal{H}_{-1}(x) \textnormal{H}_{-2,-1}(2)+\textnormal{H}_{-1}(2) \textnormal{H}_{-1,4}(x)-\textnormal{H}_{-1,4,3}(x)+\textnormal{H}_{3,-2,-1}(2).
\end{equation*}

\begin{remark}
Analyzing the underlying algorithm we see that in the trans\-form of $\H{m_1,m_2,\ldots,m_w}{b-x}$ we find only one generalized polylogarithm of weight $w$ and argument $x$, namely $\H{-(m_1-b),-(m_2-b),\ldots.-(m_w-b)}{x}$.
All other polylogarithms at argument $x$ that appear have lower weights.
In general the indices of the generalized polylogarithms at argument $x$ that appear in the transformed expression are in the set $\{-(m_1-b),-(m_2-b),\ldots.-(m_w-b)\}.$
Note that if we look at the reverse transform \ie $x\rightarrow b-x$ we can derive a similar property.
\label{SSbxxRemark}
\end{remark}

\subsection{$\frac{1-x}{1+x} \rightarrow x$}
\label{SS1x1x}

\noindent
Next, we present an algorithm that performs the transformation $\frac{1-x}{1+x}\rightarrow x$, like, e.g.,
\begin{align*}
\textnormal{H}_{-2,1}\left(\frac{1-x}{1+x}\right)=&\textnormal{H}_{-1}(1) \left(-\textnormal{H}_{-3}(x)+\textnormal{H}_{-1}(x)\right)+\textnormal{H}_{-3,-1}(x)-\textnormal{H}_{-3,0}(x)+\textnormal{H}_{-2,1}(1)\\
	&-\textnormal{H}_{-1,-1}(x)+\textnormal{H}_{-1,0}(x).
\end{align*}
Here we consider indices with $m_i\in \R\setminus(0,1)$.
Proceeding recursively on the weight $w$ of the generalized polylogarithm we have for $0<\frac{1-x}{1+x}<1$, $y_1<-1$, $-1<y_2<0$ and $y_3>1:$
\begin{eqnarray}
\H{y_1}{\frac{1-x}{1+x}}&=&-\H{0}{-y_1}-\H{-1}{x}+\H{0}{1-y_1}+\H{\frac{1-y_1}{1+y_1}}x\nonumber\\
\H{-1}{\frac{1-x}{1+x}}&=&\H{-1}{1}-\H{-1}x\nonumber\\
\H{y_2}{\frac{1-x}{1+x}}&=&-\H{0}{-y_2}-\H{-1}{x}+\H{0}{1-y_2}-\H{\frac{1-y_2}{1+y_2}}x\nonumber\\
\H{0}{\frac{1-x}{1+x}}&=&-\H{1}{x}+\H{-1}x\label{SStrafo1x1x0}\\
\H{1}{\frac{1-x}{1+x}}&=&-\H{-1}{1}-\H{0}{x}+\H{-1}{x}\label{SStrafo1x1x}\\
\H{y_3}{\frac{1-x}{1+x}}&=&\H{0}{y_3}+\H{-1}{x}-\H{0}{y_3-1}-\H{\frac{1-y_3}{1+y_3}}x.\nonumber
\end{eqnarray}
Now let us look at higher weights $w>1.$ We consider $\H{m_1,m_2,\ldots,m_w}{\frac{1-x}{1+x}}$ with $m_i\in \R\setminus(0,1)$ and suppose that we can already apply the transformation for generalized polylogarithms
 of weight $<w.$ If $m_1=1,$ we can remove leading ones and end up with generalized polylogarithms without leading ones and powers of $\H{1}{\frac{1-x}{1+x}}.$ For the powers of $\H{1}{\frac{1-x}{1+x}}$ we
 can use (\ref{SStrafo1x1x}); therefore, only the cases in which the first index $m_1\neq 1$ are to be considered. For $y_1<0$, $ y_1\neq-1$ and $y_2>1$ we get:
\begin{eqnarray}
\H{y_1,m_2,\ldots,m_w}{\frac{1-x}{1+x}}&=&\H{y_1,m_2,\ldots,m_w}1-\int_0^x\frac{1}{1+t}\H{m_2,\ldots,m_w}{\frac{1-t}{1+t}}dt\nonumber\\
					& &+\int_0^x\frac{1}{\frac{1-y_1}{1+y_1}+t}\H{m_2,\ldots,m_w}{\frac{1-t}{1+t}}dt\nonumber\\
\H{-1,m_2,\ldots,m_w}{\frac{1-x}{1+x}}&=&\H{-1,m_2,\ldots,m_w}1-\int_0^x\frac{1}{1+t}\H{m_2,\ldots,m_w}{\frac{1-t}{1+t}}dt\nonumber\\
\H{0,m_2,\ldots,m_w}{\frac{1-x}{1+x}}&=&\H{0,m_2,\ldots,m_w}1-\int_0^x\frac{1}{1-t}\H{m_2,\ldots,m_w}{\frac{1-t}{1+t}}dt\nonumber\\
		 &&-\int_0^x\frac{1}{1+t}\H{m_2,\ldots,m_w}{\frac{1-t}{1+t}}dt\label{Equ:RestrictedTrans}\\
\H{y_2,m_2,\ldots,m_w}{\frac{1-x}{1+x}}&=&\H{y_2,m_2,\ldots,m_w}1+\int_0^x\frac{1}{1+t}\H{m_2,\ldots,m_w}{\frac{1-t}{1+t}}dt\nonumber\\
					& &-\int_0^x\frac{1}{\frac{1-y_2}{1+y_2}+t}
\H{m_2,\ldots,m_w}{\frac{1-t}{1+t}}dt.
\end{eqnarray}
Note that the transformation~\eqref{Equ:RestrictedTrans} can be only applied if $(m_2,\dots,m_w)$ is not the zero vector. However, by definition the arising special case
$\H{0,\ldots,0}{\frac{1-x}{1+x}}$ reduces to $\H{0}{\frac{1-x}{1+x}}$ which
we handled in (\ref{SStrafo1x1x0}).

Since we know the transform for weights $<w$, we can apply it to $\H{m_2,\ldots,m_w}{\frac{1-t}{1+t}}$.  Linearizing the arising products of generalized polylogarithms with the shuffle product and using the definition of
the generalized polylogarithms we obtain the required weight $w$ identity.

\subsection{$k x\rightarrow x$}

\noindent
\begin{lemma}
 If $m_i \in \R, m_l\neq 0, x\in\R^+$ and $k\in\R^+$ such that the generalized
polylogarithm
\end{lemma}
\noindent
{\it
$\H{m_1,\ldots, m_l}{k\cdot x}$ is defined then we have
\begin{eqnarray}
\H{m_1,\ldots, m_l}{k\cdot x}=\H{\frac{m_1}{k},\ldots,\frac{m_l}{k}}{x}.
\end{eqnarray}}
\begin{proof}
 We proceed by induction on $l.$ For $l=1$ we have:
$$
\H{m_1}{k\cdot x}=\int_0^{kx}\frac{1}{\abs{m_1}-\sign{m_1}u}du=\int_0^x\frac{k}{\abs{m_1}-\sign{m_1}ku}du=\H{\frac{m_1}{k}}x.
$$
Suppose the lemma holds for $l.$ For $m_1\neq 0$ we get
\begin{eqnarray*}
\H{m_1,\ldots, m_{l+1}}{k\cdot x}&=&\int_0^{kx}\frac{\H{m_2,\ldots,m_{l+1}}{u}}{\abs{m_1}-\sign{m_1}u}du
	=\int_0^{x}\frac{\H{m_2,\ldots,m_{l+1}}{ku}}{\abs{m_1}-\sign{m_1}ku}kdu\\
&=&\int_0^{x}\frac{\H{\frac{m_2}{k},\ldots,\frac{m_{l+1}}{k}}{u}}{\frac{\abs{m_1}}{k}-\sign{m_1}u}du=\H{\frac{m_1}{k},\ldots,\frac{m_{l+1}}{k}}{x}.
\end{eqnarray*}
For $m_1 = 0$ we get
\begin{eqnarray*}
\H{0,m_2,\ldots, m_{l+1}}{k\cdot x}&=&\int_0^{kx}\frac{\H{m_2,\ldots,m_{l+1}}{u}}{u}du =\int_0^{x}\frac{\H{m_2,\ldots,m_{l+1}}{ku}}{ku}kdu\\
&=&\int_0^{x}\frac{\H{\frac{m_2}{k},\ldots,\frac{m_{l+1}}{k}}{u}}{u}du=\H{\frac{0}{k},\frac{m_2}{k},\ldots,\frac{m_{l+1}}{k}}{x}.
\end{eqnarray*}
\end{proof}

\noindent The proof of the following lemma is straightforward (compare
\cite{Remiddi:1999ew}, where it is stated for harmonic polylogarithms).
\begin{lemma}
 If $m_i \in \R, m_l\neq 0$ and $x\in\R^-$ such that the generalized polylogarithm $\H{m_1,\ldots, m_l}{-x}$ is defined then we have
\begin{eqnarray}
\H{m_1,\ldots, m_l}{-x}=(-1)^{l-k}\H{-m_1,\ldots, -m_l}{x},
\end{eqnarray}
here $k$ denotes the number of $m_i$ which equal zero.
\label{SStransformminusx}
\end{lemma}

\subsection{$\frac{1}{x}\rightarrow x$}
\label{SS1dxx}

\noindent
Here we assume that $0<x<1$.
First, we consider only indices with $m_i\leq0$.
Proceeding recursively on the weight $w$ of the generalized polylogarithm we have for $y<0$ that
\begin{equation*}
\H{y}{\frac{1}{x}}=\H{\frac{1}{y}}{x}-\H{0}x-\H{0}{-y}\quad\text{ and }\quad
\H{0}{\frac{1}{x}}=-\H{0}{x}.
\end{equation*}
Now let us look at higher weights $w>1.$ We consider $\H{m_1,m_2,\ldots,m_w}{\frac{1}{x}}$ with $m_i\leq 0$ and suppose that we can already apply the transformation for generalized polylogarithms
of weight $<w.$ For $m_1<0$ we get (compare \cite{Remiddi:1999ew}):
\begin{eqnarray*}
\H{m_1,m_2,\ldots,m_w}{\frac{1}{x}}&=&\H{m_1,m_2,\ldots,m_w}1+\int_x^1\frac{1}{t^2(-m_1+1/t)}\H{m_2,\ldots,m_w}{\frac{1}{t}}dt\\	 &=&\H{m_1,m_2,\ldots,m_w}1+\int_x^1\frac{1}{t}\H{m_2,\ldots,m_w}{\frac{1}{t}}dt\\
& &-\int_x^1\frac{1}{-\frac{1}{m_1}+t}\H{m_2,\ldots,m_w}{\frac{1}{t}}dt\\
\H{0,m_2,\ldots,m_w}{\frac{1}{x}}&=&\H{0,m_2,\ldots,m_w}1+\int_x^1\frac{1}{t^2(1/t)}\H{m_2,\ldots,m_w}{\frac{1}{t}}dt\\				  &=&\H{0,m_2,\ldots,m_w}1+\int_x^1\frac{1}{t}\H{m_2,\ldots,m_w}{\frac{1}{t}}dt.\\
\end{eqnarray*}
Since we know the transform for weights $<w,$ we can apply it to $\H{m_2,\ldots,m_w}{\frac{1}{t}}$ and finally we obtain the required weight $w$ identity by using the definition of
the generalized polylogarithms.

We can extend this transformation by the following observations. An index $m>0$ in the
index set leads to a branch point at $m$ and a branch cut discontinuity in the complex
plane for $x\in(m,\infty).$~\footnote{For a definition of single-valued
polylogarithms see
\cite{Brown:2004A,Drummond:2012bg,Dixon:2012yy,Chavez:2012kn}.}
This corresponds to the branch point at $x=m$
and the branch cut discontinuity in the complex plane for $x\in(m,\infty)$ of $\log(m-x)=\log(m)-\H{m}x.$ However, the analytic properties of the logarithm are well known and we
can set for $0<x<1$ for instance
\begin{eqnarray}
\H{m}{\frac{1}{x}}&=&\H{\frac{1}{m}}{x}+\H{0}{m}+\H{0}{x}-i\pi, \label{SStrafo1dx11}
\end{eqnarray}
by approaching the argument $\frac{1}{x}$ form the lower half complex plane.
We can now consider an alphabet with at most one letter $m\geq1$ (Note that we
could as well assume $m>0$. However, for simplicity we restrict to $m\geq1$).
For such an alphabet the strategy is as follows: if a generalized polylogarithm has leading $m$'s, we remove them and end up with generalized polylogarithms without leading $m$'s and powers
of $\H{m}{\frac{1}{x}}.$ To deal with the generalized polylogarithms without leading
$m$, we can slightly modify the previous part of this Section and for the powers of $\H{m}{\frac{1}{x}}$ we can
use~(\ref{SStrafo1dx11}).
Following this tactic, we arrive, e.g., at the identity
\begin{align*}
\textnormal{H}_{-2,1}\left(\frac{1}{x}\right)=&-i \pi  \left(-\textnormal{H}_{-\frac{1}{2}}(1)+\textnormal{H}_{-\frac{1}{2}}(x)-\textnormal{H}_0(x)\right)+\textnormal{H}_{-2,1}(1)-\textnormal{H}_{-\frac{1}{2},0}(1)+\textnormal{H}_{-\frac{1}{2},0}(x)\\	 &-\textnormal{H}_{-\frac{1}{2},1}(1)+\textnormal{H}_{-\frac{1}{2},1}(x)-\textnormal{H}_{0,0}(x)+\textnormal{H}_{0,1}(1)-\textnormal{H}_{0,1}(x).
\end{align*}

\section{Power Series Expansion of Multiple Polylogarithms}\label{Sec:PowerSeriesMPL}

\vspace*{1mm}
\noindent
An algorithm to calculate series expansions for harmonic polylogarithms
$\H{m_1,m_{2},\ldots,m_w}{x}$ with $m_i\in\{-1,0,1\}$ has been presented
in~\cite{Remiddi:1999ew}.
Subsequently, we generalize these steps for generalized polylogarithms, i.e.,
$m_i\in\set R$. In general, $\H{m_1,m_{2},\ldots,m_w}{x}$ does not have a Taylor series
expansion if trailing zeroes occur, i.e., the expansion consists of a $\log(x)$-part
and a part free of $\log(x)$. The separation into these two parts can be
easily accomplished
by removing trailing zeroes as described in Section~\ref{Sec:BasicProp}. In other words, one gets a polynomial expression in $H_0(x)=\log(x)$ with coefficients in terms of generalized polylogarithms that have no trailing $0$'s. For each of these coefficients we are now in the position to calculate a Taylor series expansion as follows.

For polylogarithms of weight $1$ with $j\in \R$ and $x\in(0,\abs{j})$ we have
\begin{eqnarray*}
\H{j}x=\left\{
		\begin{array}{ll}
				-\sum_{i=1}^\infty \frac{j^{-i}(-x)^i}{i},&  \textnormal{if }j< 0\\
\\
				\sum_{i=1}^\infty \frac{j^{-i}x^i}{i},& \textnormal{if }j > 0.
		\end{array} \right.
\end{eqnarray*}
Let $\ve m=(m_1,\ldots,m_w)$ be a non-zero vector and $q=\min_{m_i\neq 0}{\abs{m_i}}$. Assuming that for $x \in (0,q)$
$$\H{\ve m}x=\sum_{i=1}^{\infty}\frac{\sigma^ix^i}{i^a}\S{\ve n}{\ve b;i}$$ one can write the relations for $j>0$ and $x \in (0, \min(j,q))$
\begin{eqnarray*}
\H{0,\ve m}x&=&\sum_{i=1}^{\infty}\frac{\sigma^ix^i}{i^{a+1}}\S{\ve n}{\ve b;i}\\
\H{j,\ve m}x&=&\sum_{i=1}^{\infty}\frac{x^i}{ij^i}\S{a,\ve n}{\sigma j,\ve b;i-1}\\
&=&\sum_{i=1}^{\infty}\frac{x^i}{ij^i}\S{a,\ve n}{\sigma j,\ve b;i}-\sum_{i=1}^{\infty}\frac{\sigma^ix^i}{i^{a+1}}\S{\ve n}{\ve b;i}\\
\H{-j,\ve m}x&=&-\sum_{i=1}^{\infty}\frac{x^i}{i(-j)^i}\S{a,\ve n}{-\sigma j,\ve b;i-1}\\
&=&-\sum_{i=1}^{\infty}\frac{x^i}{i(-j)^i}\S{a,\ve n}{-\sigma j,\ve b;i}+\sum_{i=1}^{\infty}\frac{\sigma^ix^i}{i^{a+1}}\S{\ve n}{\ve b;i}.
\end{eqnarray*}
\begin{proof}
We just prove the case that $j>0$. The other cases follow analogously.
 \begin{eqnarray*}
\H{j,\ve m}x&=&\int_0^x{\frac{1}{j-y}\sum_{i=1}^{\infty}\frac{\sigma^iy^i}{i^a}\S{\ve n}{\ve b;i}}dy\\
&=&\int_0^x{\frac{1}{j}\sum_{k=0}^\infty\left(\frac{y}{j}\right)^k\sum_{i=1}^{\infty}\frac{\sigma^iy^i}{i^a}\S{\ve n}{\ve b;i}}dy\\
&=&\frac{1}{j}\int_0^x{\sum_{k=0}^\infty\left(\frac{y}{j}\right)^k\sum_{i=0}^{\infty}\frac{\sigma^{i+1}y^{i+1}}{(i+1)^a}\S{\ve n}{\ve b;i+1}}dy\\
&=&\frac{1}{j}\int_0^x{\sum_{i=0}^\infty\sum_{k=0}^{i}\left(\frac{y}{j}\right)^{i-k}\frac{\sigma^{k+1}y^{k+1}}{(k+1)^a}\S{\ve n}{\ve b;k+1}}dy\\
&=&\int_0^x{\sum_{i=0}^\infty\frac{y^{i+1}}{j^{i+2}}\S{a,\ve n}{\sigma j,\ve b;i+1}}dy\\
&=&\sum_{i=1}^\infty\frac{x^{i+1}}{j^{i+1}(i+1)}\S{a,\ve n}{\sigma j,\ve b;i}\\
&=&\sum_{i=1}^\infty\frac{x^{i}}{ij^{i}}\S{a,\ve n}{\sigma j,\ve b,i-1}.
 \end{eqnarray*}
\end{proof}

Separating the $\log(x)=H_0(x)$ part from the log-free part and applying the formulas from above we get, e.g.,
$$\H{-2,0}{x}=\sum _{i=1}^{\infty } \frac{2^{-i} (-x)^i}{i^2}-\H{0}{x} \sum _{i=1}^{\infty } \frac{2^{-i} (-x)^i}{i}.$$

Summarizing, given $\H{m_1,m_{2},\ldots,m_w}{x}$ with $m_w\neq0$ and $q=\min_{m_i\neq 0}{\abs{m_i}}$, one can calculate a Taylor series expansion which holds for $x\in[0,q)$ and where the coefficients of the Taylor expansion are given explicitly in terms of $S$-sums.

\vspace*{1.5mm}
\noindent
\textbf{Inverse construction.} In addition, since the formulas from above can be reversed, also the other direction is possible: given a sum of the form
$\sum_i^{\infty}(c x)^i\frac{S_{\ve n}(\ve b; i)}{i^k}$
with $k\in \N\cup\{0\}$ and $c\in\R^*,$ one can compute a linear combination of (possibly weighted) generalized polylogarithms $h(x)$ such that for $x\in(0,j)$ with $j>0$ small enough
$\sum_i^{\infty}(c x)^i\frac{S_{\ve n}(\ve b; i)}{i^k}=h(x)$.

E.g., for $x\in(0,\frac{1}{2})$ we have
\begin{equation*}
\sum_{i=1}^{\infty}\frac{2^i x^i S_{1,1}\left(3,\frac{1}{2};i\right)}{i}=
\textnormal{H}_{0,0,\frac{1}{3}}(x)+\textnormal{H}_{0,\frac{1}{6},\frac{1}{3}}(x)+\textnormal{H}_{\frac{1}{2},0,\frac{1}{3}}(x)+\textnormal{H}_{\frac{1}{2},\frac{1}{6},\frac{1}{3}}(x).
\end{equation*}

\vspace*{1.5mm}
\noindent
\textbf{Asymptotic behavior.} Moreover, one can determine the
asymptotic behavior of a generalized polylogarithms $\H{\ve m}x$:
Define $y:=\frac{1}{x}.$ Using Section~\ref{SS1dxx} on $\H{\ve m}{\frac{1}{y}}=\H{\ve m}x$ we can rewrite $\H{\ve m}x$ in terms of generalized polylogarithms at argument $y$ together with some constants.
Now we obtain the power series expansion of the generalized polylogarithms at argument
$y$
around 0 easily using the previous part of this Section. Since sending $x$ to infinity corresponds to sending $y$ to
zero, we get the asymptotic behavior of $\H{\ve m}x.$

E.g., performing these steps yields
 \begin{align*}
  \H{-2,0}x=& \H{-2,0}1+\H{-\frac{1}{2},0}1-\H{-\frac{1}{2},0}{\frac{1}x}+\H{0,0}{\frac{1}x}\\
	   =&\frac{1}{2}\; \H0x^2-\H0x \left(\sum _{i=1}^{\infty } \frac{\left(-\frac{2}{x}\right)^{i}}{i}\right)-\sum _{i=1}^{\infty}\frac{\left(-\frac{2}{x}\right)^{i}}{i^2}+\H{-2,0}1+\H{-\frac{1}{2},0}1.
 \end{align*}

\section{Values of Multiple Polylogarithms Expressed by $S$-Sums at Infinity}
\label{SSinfval}

\vspace*{1mm}
\noindent
Consider the generalized polylogarithm $\H{m_1,m_2,\ldots,m_w}x$ with $q:=\min_{m_i>0}{m_i}$ and $c\in\set R,\ c \geq 0$ such that $\H{m_1,m_2,\ldots,m_w}c$ is finite.
To be more precise, one of the following cases holds:
\begin{enumerate}
\item $c<q$
\item $c=q$ and $m_1\neq q$
\item $(m_1,\ldots,m_w)=(m_1,\ldots,m_k,1,0,\dots,0)$ with $q':=\min_{m_i>0}{m_i}$ and $c<q'$.
\end{enumerate}

If the generalized polylogarithm $\H{\ve{m}}c$ has trailing zeroes, we first extract them.  If $0<c\leq 1,$ we end up in an expression of generalized polylogarithms without
trailing zeroes and powers of $\H{0}{c}$. Otherwise, if $c>1$, we cannot handle polylogarithms of the form $\H{m_1,m_2,\ldots,m_k,1,0,\ldots,0}c$
(see Section~\ref{Sec:BasicProp}) and keep them untouched. Note that by the shuffle product no extra indices are introduced and the first non-zero index is again $m_1$. Thus the arising polylogarithms are again finite at the evaluation point $c$.

We can use the following lemma (iteratively) to rewrite powers of $\H{0}{c}.$
\begin{lemma}
 Let $c>0.$ We have
\begin{eqnarray*}
  0<c< 1:&&\H{0}c=-\S{1}{1-c;\infty}\\
  c=1:&&\H{0}c=0\\
  1<c\leq 2:&&\H{0}c=-\S{1}{1-c;\infty}\\
  2<c:&&\H{0}c=\H{0}2-\H{0}{\frac{c}{2}}.
\end{eqnarray*}
\end{lemma}

To rewrite generalized polylogarithms of the form $\H{m_1,m_2,\ldots,m_w}c$ where $m_w\neq 0$ and $\min_{m_i\neq 0}{\abs{m_i}}>c$, we use the power series expansion
of $\H{m_1,m_2,\ldots,m_w}x$ and send $x\rightarrow c$. E.g., for one of the arising sums we obtain
$$
\sum_{i=1}^{\infty}\frac{\sigma^ix^i}{i^a}\S{\ve n}{\ve b;i} \rightarrow \S{a,\ve n}{c\sigma,\ve b;\infty}.
$$

What remains open are the situations of the type $\H{m_1,m_2,\ldots,m_k,1,0,\ldots,0}c$ with $c>1$ and polylogarithms
$\H{m_1,m_2,\ldots,m_w}c$ with $m_w\neq0$, $\max_{m_i<0}{m_i}> -c$,  $\min_{m_i>0}{m_i}\geq c$ and $m_1\neq c$.

As shortcut, one can rewrite $\H{m_1,m_2,\ldots,m_w}c$ with $(m_1,m_2,\ldots,m_w)=(1,0,\ldots,0)$ and $c>1$ by using
the following lemma iteratively.
\begin{lemma}
 For $c>1$ and  $(m_1,m_2,\ldots,m_w)=(1,0,\ldots,0)$ we have
$$
\H{m_1,m_2,\ldots,m_w}c=(-1)^{w+1}S_w(1;\infty)-\H{0,m_1-1,\ldots,m_w-1}{c-1}.
$$
\end{lemma}

To rewrite generalized polylogarithms of the form $\H{m_1,m_2,\ldots,m_w}c$ where $m_1\in \R \setminus (0, c]$,  $m_w\neq 0,$ $\min_{m_i > 0}{m_i}\geq c$ and
$s:=\max_{m_i<0}{m_i}>-c$ we use Lemma \ref{SStransformplusa1} and afterwards send $x\rightarrow c+s$, i.e., we get
\begin{align*}
\H{m_1,\ldots, m_w}{c}=&\H{m_1,m_2,\ldots, m_w}{-s}+\H{m_1+s}{c+s}\H{m_2,\ldots, m_w}{-s}\\
&+\H{m_1+s,m_2+s}{c+s}\H{m_3,\ldots, m_w}{-s}+\cdots+\\
&+\H{m_1+s,\ldots, m_{k-1}+s}{{c+s}}\H{m_l}{-s}+\H{m_1+s,\ldots, m_w+s}{{c+s}}.
\end{align*}
Note that $-s<c$ and $c+s<c$.
To rewrite generalized polylogarithms of the form $\H{m_1,m_2,\ldots,m_k,1,0,\ldots,0}c$ where $c>1$ we can proceed similar to the previous case: we use
Lemma \ref{SStransformplusa3} to rewrite $\H{m_1,m_2,\ldots,m_w}{x+1}$ and afterwards send $x\rightarrow c-1$. Observe that both cases reduce the problem to a simpler situation: We can either handle all arising polylogarithms by the recipes given above. If not, observe that the remaining polylogarithms are of the form $\H{n_1,\ldots, n_{l}}{u}$ with $q':=\min_{n_i\neq0}|n_i|$ and $u<c$. Note that the evaluation point gets smaller. In addition, the index vector $(n_1,\dots,n_{l})$ is given by a subset of $(m_1,\dots,m_w)$ up to the possible addition of $s<0$ where $s$ is the maximum of the arising negative indices of $(m_1,\dots,m_w)$. Thus applying this reduction iteratively (to the polylogarithms that cannot be handled by the cases from above) decreases step by step the positive evaluation point and in the same time decreases the negative indices and thus enlarges the absolute value of $s$ superlinearly. Consequently, from a certain point on $q'$ gets larger than $q$ and the difference between these values increases in each iteration. This guarantees that eventually $u<q'$, i.e., the arising polylogarithms can be transformed to $S$-sums by using the series expansion representation and the evaluation at $x=u$ as described above.

Summarizing, we can rewrite all finite generalized polylogarithms
$\H{m_1,m_2,\ldots,m_w}c$ in terms of $S$-sums at infinity, like, e.g.,
\begin{align*}
 \H{4,1,0}{3}=&-\S{-3}{\infty} + 2 \S{3}{\infty} + \S{-2}{\infty} \S{1}{\tfrac{1}{2};\infty} - \S{2}{\infty} \S{1}{\tfrac{2}{3};\infty}\\
	      &- \S{3}{-\tfrac{1}{2};\infty} + \S{-1}{\infty} \left(\S{-2}{\infty} - \S{1, 1} {\tfrac{1}{2}, -2;\infty}\right) \\
	      & - \S{1, 2} {\tfrac{1}{4}, 4;\infty}+ \S{1, 2} {\tfrac{1}{3}, -3;\infty} + \S{1, 2}{\tfrac{1}{2}, -1;\infty} \\
	      & + \S{2, 1}{-1, \tfrac{1}{2};\infty} - \S{2, 1}{\tfrac{1}{4}, 4;\infty}- \S{1, 1, 1}{\tfrac{1}{2}, -2, \tfrac{1}{2};\infty}.
\end{align*}
\noindent Note that this process can be reversed: given a finite $S$-sum at infinity we
can rewrite it, e.g., in terms of
generalized polylogarithms at one:
\begin{align*}
 \SS{2,3,1}{-\tfrac{1}{2},\tfrac{1}{3},2}{\infty}=&-\H{0, -2, 0, 0, -6, -3}1 + \H{0, -2, 0, 0, 0, -3}1 + \H{0, 0, 0, 0, -6, -3}1\\
						& - \H{0, 0, 0, 0, 0, -3}1.
\end{align*}
For a complete algorithm for the reverse direction we refer to~\cite[Algorithm~1]{Ablinger:12}.

We note that if we consider the generalized polylogarithm $\H{m_1,m_2,\ldots,m_w}x$ and $c\in\set R,\ c < 0$ such that $\H{m_1,m_2,\ldots,m_w}c$ is finite then we can use Lemma~\ref{SStransformminusx} to transform $\H{m_1,m_2,\ldots,m_w}x$ to an expression in terms of generalized polylogarithms at argument $-c.$
Hence we can again use the strategy presented above to rewrite this generalized
polylogarithm in terms of $S$-sums at infinity.

\section{Analytic Continuation of $S$-sums}
\label{Sec:AnalyticCont}

\vspace*{1mm}
\noindent

It has been remarked in~\cite{Yndurain:1983} that Carlson's
theorem~\cite{Carlson:14,Titchmarsh:1939}\footnote{The original theorem needs to
be extended to be even applicable in case of the harmonic sums.}
can be used to perform the analytic continuation of harmonic sums:  the evaluation of
the harmonic sums at integer points determines uniquely (among a certain class of
analytic functions with exponential growth) the analytic continuation in the complex
plane. We will elaborate this property in further details for harmonic sums and more
generally for $S$-sums.

A function $f:M\to\set C$ for a region $M\subset\set C$ is of Carlson-type (C-type) if
\begin{itemize}
\item it is analytic in the right half plane,
\item it is of exponential type in the half plane, i.e., there exist $c,C\in\set R^+$ such that
for all $z\in\set C$ with $\Im(z)\geq0$ we have that
$$|f(z)| \leq C e^{c|z|},$$
\item and there exist $c,C\in\R^+$ with $c<\pi$ such that
$$\abs{f(i y)}\leq C e^{c \abs{y}}.$$
\end{itemize}

For such functions the following theorem holds; for a detailed proof (together with even more general statements) we refer to~\cite[Sec.~9.2]{Boas:54}.

\begin{thm}[Carlson]
Let $h(z)$ be a function of C-type. If $h(n)=0$ for all $n\in\set N$ then $h=0$.
\end{thm}

Throughout the remaining article we rely on the following consequence: Let $f(z)$ and $g(z)$ be functions of C-type. If $f(n)=g(n)$ for all $n\in\set N$, then $f=g$ in the domain where both functions are defined.

The following simple closure properties will be used later. The proofs are immediate.

\begin{lemma}\label{Lemma:Closure}
Let $f(z)$ and $g(z)$ be functions of C-type and let $a\in\set R^{+}$. Then
$f(z)+g(z)$, $a^z=e^{z\ln(a)}$ and $a^z f(z)$ are of C-type.
\end{lemma}

\begin{prop}\label{Prop:PosSAnalyt}
For $m_i\in\N,$ $b_i\in\R^+$ and $z\in \C$ consider the integral $f(z)=I_{m_1,m_2,\ldots,m_k}(b_1,b_2,\ldots,b_k;z)$ from Theorem~\ref{SSintrep}.
Then $f(z)$ is of C-type.
\end{prop}

\begin{proof}
We note that $f(z)$ is analytic for $\Re{(z)}\geq 0$. Next we show that $f(z)$ is of exponential type. For $1\geq\abs{z}\geq 0$ with $\Re{(z)}\geq 0$ we can bound $|f(z)|$ by a constant.
For $\abs{z}>1$  and $d:=\max(b_1,b_1b_2,\ldots,b_1b_2\cdots b_k)\geq0$ we have
\begin{equation}\label{Equ:BoundF}
\begin{split}
&\abs{f(z)}=\\
&\hspace{0.5cm} \Biggl|\int_0^{b_1\cdots b_k}{\frac{dx_{k}^{m_k}}{x_{k}^{m_k}} \cdots \int_0^{x_{k}^2}{\frac{dx_{k}^1}{x_{k}^1-b_1\cdots b_{k-1}}}} \hspace{0.2cm}\cdots\hspace{0.2cm}  \int^{x_{2}^1}_{b_1}{\frac{dx_{1}^{m_1}}{x_{1}^{m_1}} \cdots \int_0^{x_{1}^2}{\frac{{x}^z-1}{x-1}dx}}\Biggr|\\
&\hspace{0.5cm} \leq\int_0^{b_1\cdots b_k}{\frac{dx_{k}^{m_k}}{\abs{x_{k}^{m_k}}} \cdots \int_0^{x_{k}^2}{\frac{dx_{k}^1}{\abs{x_{k}^1-b_1\cdots b_{k-1}}}}} \hspace{0.2cm}\cdots\hspace{0.2cm}  \int^{x_{2}^1}_{b_1}{\frac{dx_{1}^{m_1}}{\abs{x_{1}^{m_1}}} \cdots \int_0^{x_{1}^2}{\abs{\frac{{x}^z-1}{x-1}}dx}}.
\end{split}
\end{equation}
Now we show that there is a $c\in\R^+$ such that $\abs{\frac{x^z-1}{x-1}}\leq c (d+1)^{\abs{z}}$ for $0\leq x \leq d$ and all $z=a+b i$ for $a\in \R^+, b\in\R$.
If $d<1$, we have
\begin{eqnarray*}
 \abs{\frac{x^z-1}{x-1}}\leq \frac{x^a+1}{1-d} \leq  \frac{1}{1-d} (d+1)^{\abs{z}}.
\end{eqnarray*}
Finally, let  $d\geq 1.$ Since $\lim_{x\rightarrow 1}\abs{\frac{x^z-1}{x-1}}=\abs{z}\leq (d+1)^{\abs{z}}$, there are $\ep>0$ and $c>0$ such that $\abs{\frac{x^z-1}{x-1}}\leq c (d+1)^{\abs{z}}$ for $1-\ep\leq x \leq 1+\ep.$
It remains to show that $\abs{\frac{x^z-1}{x-1}}\leq c (d+1)^{\abs{z}}$ for $x\in(0,1-\ep)\cup(1+\ep,d):$
\begin{eqnarray*}
 \abs{\frac{x^z-1}{x-1}}\leq \frac{x^a+1}{\ep} \leq  \frac{1}{\ep} (d+1)^{\abs{z}}.
\end{eqnarray*}
In summary, we can choose $c\in\set R^{+}$ (independent of the choice of $z$) with $\abs{\frac{x^z-1}{x-1}}\leq c (d+1)^{\abs{z}}$ for $0\leq x \leq d$ and using~\eqref{Equ:BoundF} it follows that
\begin{eqnarray*}
f(z)&\leq& c\cdot e^{\log(d+1)\abs{z}} \underbrace{\int_0^{b_1\cdots b_k}{\abs{\frac{dx_{k}^{m_k}}{x_{k}^{m_k}}} \cdots \int_0^{x_{k}^2}{\abs{\frac{dx_{k}^1}{x_{k}^1-b_1\cdots b_{k-1}}}}} \hspace{0.2cm}\cdots\hspace{0.2cm}  \int^{x_{2}^1}_{b_1}{\abs{\frac{dx_{1}^{m_1}}{x_{1}^{m_1}}} \cdots \int_0^{x_{1}^2}{1dx}}}_{\in\R^+, \textnormal{ constant}}\\
&\leq& \bar{c}\cdot e^{\log(d+1)\abs{z}}
\end{eqnarray*}
for some $\bar{c}\in\R^+$ (again independent of the choice of $z$). To conclude that $f(z)$ is of C-type, we have to show that we can choose a $\bar{c}\in\R^+$ such that for all $b\in\R$ we have
\begin{eqnarray*}
&&\abs{f(b i,b_1,\ldots,b_k)}=\\
&&\hspace{0.5cm} \Biggl|\int_0^{b_1\cdots b_k}{\frac{dx_{k}^{m_k}}{x_{k}^{m_k}} \cdots \int_0^{x_{k}^2}{\frac{dx_{k}^1}{x_{k}^1-b_1\cdots b_{k-1}}}} \hspace{0.2cm}\cdots\hspace{0.2cm}  \int^{x_{2}^1}_{b_1}{\frac{dx_{1}^{m_1}}{x_{1}^{m_1}} \cdots \int_0^{x_{1}^2}{\frac{{x}^{b i}-1}{x-1}dx}}\Biggr|\leq \bar{c}\cdot e^{\abs{b}}.
\end{eqnarray*}
We show that there is a $c\in\R^+$ such that $\abs{\frac{x^{bi}-1}{x-1}}\leq c {\cdot e^{\abs{b}}}$ for $0\leq x \leq d$ and all $b\in\set R^{+}$: If $d<1$ we have
$\abs{\frac{x^{ib}-1}{x-1}}\leq \frac{2}{1-d} \leq  \frac{2}{1-d} \cdot e^{\abs{b}}$.
Finally let $d\geq 1.$ Since $\lim_{x\rightarrow 1}\abs{\frac{x^{bi}-1}{x-1}}=\abs{b}\leq e^{\abs{b}}$ there are $\ep>0$ and $c>0$ such that $\abs{\frac{x^{bi}-1}{x-1}}\leq c e^{\abs{b}}$ for $1-\ep\leq x \leq 1+\ep.$
It remains to show that $\abs{\frac{x^{bi}-1}{x-1}}\leq c e^{\abs{b}}$ for $x\in(0,1-\ep)\cup(1+\ep,d):$
$\abs{\frac{x^{bi}-1}{x-1}}\leq \frac{2}{\ep} \leq  \frac{2}{\ep} e^{\abs{b}}$.
As for the proof part concerning the exponential growth the integral can be bounded as desired.
\end{proof}

As a consequence the analytic continuation of $S$-sums $S_{m_1,m_2,\ldots,m_k}(b_1,b_2,\ldots,b_k;n)$ with $m_i\in\set N$ and $b_i\in\set R^{+}$ within the class of C-type functions is uniquely determined by the integral $I_{m_1,m_2,\ldots,m_k}(b_1,b_2,\ldots,b_k;z)$.

If not all $b_i\in\R^+$, we can split the integral at $0$ whenever the integration interval contains $0$. For example let us consider the integral representation of $\S{1,2,2}{-1,2,-\frac{1}3;n}$:
\begin{eqnarray*}
 \S{1,2,2}{-1,2,-\frac{1}3;n}=\int _0^{\frac{2}{3}}\int _0^{x_5}\int _{-2}^{x_4}\int _0^{x_3}\int _{-1}^{x_2}\frac{-1+x_1^n}{x_5 \left(2+x_4\right) x_3 \left(1+x_2\right) \left(-1+x_1\right)}dx_1dx_2dx_3dx_4dx_5.
\end{eqnarray*}
Splitting at zero yields
\begin{eqnarray*}
  && \left(\int _0^{\frac{2}{3}}\int _0^{x_5}\frac{1}{x_5 \left(2+x_4\right)}dx_4dx_5\right) \int _{-2}^0\int _{x_3}^0\int _{-1}^{x_2}-\frac{-1+(-1)^n(-x_1)^n}{x_3 \left(1+x_2\right) \left(-1+x_1\right)}dx_1dx_2dx_3\\
  &&+\left(\int _0^{\frac{2}{3}}\int _0^{x_5}\int_0^{x_4}\int _0^{x_3}\frac{1}{x_5 \left(2+x_4\right) x_3 \left(1+x_2\right)}dx_2dx_3dx_4dx_5\right) \int_{-1}^0 \frac{-1+(-1)^n(-x_1)^n}{-1+x_1} \, dx_1\\
  &&+\int _0^{\frac{2}{3}}\int _0^{x_5}\int _0^{x_4}\int _0^{x_3}\int _0^{x_2}\frac{-1+x_1^n}{x_5\left(2+x_4\right) x_3 \left(1+x_2\right) \left(-1+x_1\right)}dx_1dx_2dx_3dx_4dx_5.\\
\end{eqnarray*}
An analytic continuation of $(-1)^n$ does not exist (in particular, it is not of C-type), hence we cannot continue $\S{1,2,2}{-1,2,-\frac{1}3;n}$ analytically. However we can continue the functions
\begin{eqnarray*}
 f_e(n)&:=&\S{1,2,2}{-1,2,-\frac{1}3;2n}\\
  &=&\left(\int _0^{\frac{2}{3}}\int _0^{x_5}\frac{1}{x_5 \left(2+x_4\right)}dx_4dx_5\right) \int _{-2}^0\int _{x_3}^0\int _{-1}^{x_2}-\frac{-1+(-x_1)^{2n}}{x_3 \left(1+x_2\right) \left(-1+x_1\right)}dx_1dx_2dx_3\\
  &&+\left(\int _0^{\frac{2}{3}}\int _0^{x_5}\int_0^{x_4}\int _0^{x_3}\frac{1}{x_5 \left(2+x_4\right) x_3 \left(1+x_2\right)}dx_2dx_3dx_4dx_5\right) \int_{-1}^0 \frac{-1+(-x_1)^{2n}}{-1+x_1} \, dx_1\\
  &&+\int _0^{\frac{2}{3}}\int _0^{x_5}\int _0^{x_4}\int _0^{x_3}\int _0^{x_2}\frac{-1+x_1^{2n}}{x_5\left(2+x_4\right) x_3 \left(1+x_2\right) \left(-1+x_1\right)}dx_1dx_2dx_3dx_4dx_5\\
\end{eqnarray*}
and
\begin{eqnarray*}
f_o(n)&:=&\S{1,2,2}{-1,2,-\frac{1}3;2n+1}\\
  &= & \left(\int _0^{\frac{2}{3}}\int _0^{x_5}\frac{1}{x_5 \left(2+x_4\right)}dx_4dx_5\right) \int _{-2}^0\int _{x_3}^0\int _{-1}^{x_2}-\frac{-1-(-x_1)^{2n+1}}{x_3 \left(1+x_2\right) \left(-1+x_1\right)}dx_1dx_2dx_3\\
  &&+\left(\int _0^{\frac{2}{3}}\int _0^{x_5}\int_0^{x_4}\int _0^{x_3}\frac{1}{x_5 \left(2+x_4\right) x_3 \left(1+x_2\right)}dx_2dx_3dx_4dx_5\right) \int_{-1}^0 \frac{-1-(-x_1)^{2n+1}}{-1+x_1} \, dx_1\\
  &&+\int _0^{\frac{2}{3}}\int _0^{x_5}\int _0^{x_4}\int _0^{x_3}\int _0^{x_2}\frac{-1+x_1^{2n+1}}{x_5\left(2+x_4\right) x_3 \left(1+x_2\right) \left(-1+x_1\right)}dx_1dx_2dx_3dx_4dx_5\\
\end{eqnarray*}
analytically since $(-1)^n=1$ for $n$ even and $(-1)^n=-1$ for $n$ odd.
For a graphical illustration we refer to Figure~\ref{Fig:SSum}. More generally, there is the following result.

\begin{figure}
\centering
\includegraphics[width=10cm]{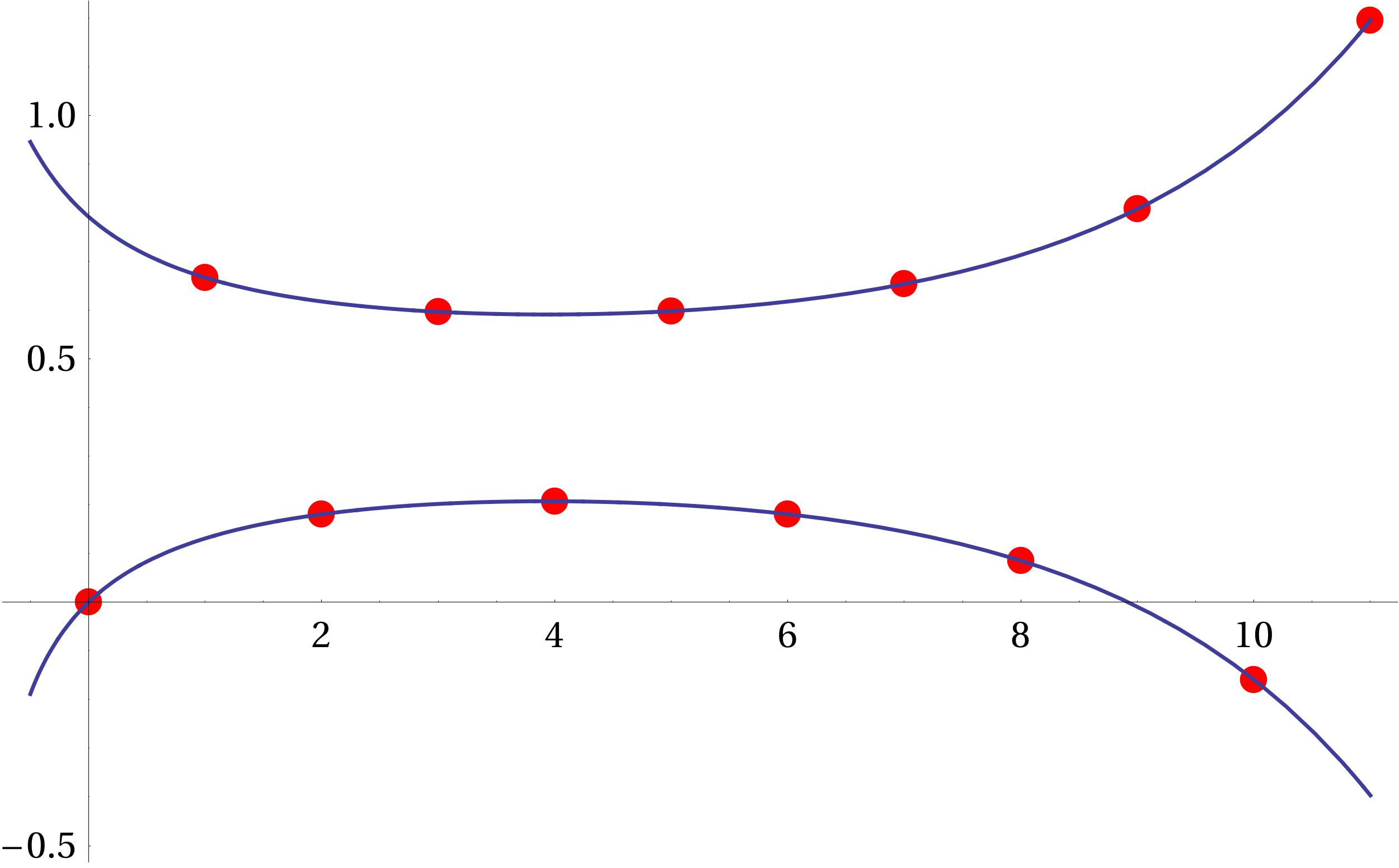}
\caption{\label{Fig:SSum}$\S{1,2,2}{-1,2,-\frac{1}3;n}$ for $n\in \{0,1,\ldots,11\}$, and $f_e\left(\frac{n}2\right)$,  $f_o\left(\frac{n-1}2\right)$ for $n\in (-\frac{1}2,11)$.}
\end{figure}

\begin{prop}\label{Prop:SCType}
For $m_i\in\N,$ $b_i\in\R^*$ and $z\in \C$ consider the integral $f(z)=I_{m_1,m_2,\ldots,m_k}(b_1,b_2,\ldots,b_k;z)$ from Theorem~\ref{SSintrep}.
Then $f(2z)$ and $f(2z+1)$ are of C-type.
\end{prop}

\noindent The proof is analogous to Proposition~\ref{Prop:PosSAnalyt}.
Consequently the analytic continuation of  even and odd $S$-sums $S_{m_1,m_2,\ldots,m_k}(b_1,b_2,\ldots,b_k;2n)$ and $S_{m_1,m_2,\ldots,m_k}(b_1,b_2,\ldots,b_k;2n+1)$ with $m_i\in\set N$ and $b_i\in\set R^*$ within the class of C-type functions are uniquely determined with the integral representations $I_{m_1,m_2,\ldots,m_k}(b_1,b_2,\ldots,b_k;2z)$
and $I_{m_1,m_2,\ldots,m_k}(b_1,b_2,\ldots,b_k;2z+1)$, respectively.

We remark that in physics applications due to the definiteness of the crossing
relations in case of the respective process using the light--cone expansion
\cite{Wilson:1969zs,Zimmermann:1970,Brandt:1972nw,Frishman:1971qn,Blumlein:1996vs}
$n$ is either even {\it or} odd and the analytic continuation is carried out
from the respective subset.

\section{The Mellin Transform and Inverse Mellin Transform Between Multiple
Polylogarithms and $S$-sums}
\label{SSIntegralRep}
\label{SSMellin}

\vspace*{1mm}
\noindent
For a harmonic polylogarithm $h(x)$, i.e., a generalized polylogarithm with the index
set $\{-1,0,1\}$, the Mellin transform~\eqref{eq:MEL1}
can be expressed in terms harmonic sums and constants (written in terms of harmonic
polylogarithms at $1$ or equivalently in terms of generalized zeta values).
Extending the Mellin transform to
\begin{align*}
\Mp{h(x)}n&:=\M{h(x)}n=\int_{0}^1x^nh(x)dx\\
\Mp{\frac{h(x)}{1+x}}n&:=\M{\frac{h(x)}{1+x}}n=\int_{0}^1 x^n \frac{h(x)}{1+x}dx\\
\Mp{\frac{h(x)}{1-x}}n&:=\int_{0}^1\frac{(x^n-1)h(x)}{1-x}dx
\end{align*}
gives a 1--1 correspondence between the harmonic polylogarithms and the harmonic sums via the extended Mellin transform and its inverse transform.

In this Section we generalize this construction and look at the Mellin-transform of generalized polylogarithms with
indices in $\R\setminus(0,1)$.
It will turn out that these transforms can be expressed using a
subclass of the $S$-sums (containing the harmonic sums) together with constants (written in terms of generalized
polylogarithms evaluated at certain points or equivalently in terms of infinite $S$-sums; see Section~\ref{SSinfval}). In particular, using the inverse Mellin transform, the $S$-sums can be expressed in terms of the Mellin transform of generalized polylogarithms.

In order to establish this connection, we have to deal with functions like $f(x)=1/(a-x)$ with $a \in (0,1)$.  However, already for this simple case, the Mellin transform is not defined since the integral $\int_0^1\frac{x^n}{a-x}$ does not converge. Hence we modify the definition of the
Mellin transform to include these cases, like $1/(a-x),$ as follows.

Let $h(x)$ be a generalized polylogarithm with indices in $\R\setminus(0,1)$ or $h(x)=1$ for $x\in [0,1]$; let $a \in (0,\infty)$, $a_1\in (1,\infty)$, $a_2 \in (0,1]$. Then we define the
extended and modified Mellin-transform as follows:
\begin{eqnarray}
\Mp{h(x)}{n}:&=&\M{h(x)}{n}=\int_0^1{x^nh(x)dx},\nonumber\\
\Mp{\frac{h(x)}{a+x}}n:&=&\M{\frac{h(x)}{a+x}}{n}=\int_0^1{\frac{x^nh(x)}{a+x}dx},\nonumber\\
\Mp{\frac{h(x)}{a_1-x}}n:&=&\M{\frac{h(x)}{a_1-x}}{n}=\int_0^1{\frac{x^nh(x)}{a_1-x}dx},\nonumber\\
\Mp{\frac{h(x)}{a_2-x}}n:&=&\int_0^{\frac{1}{a_2}}{\frac{(x^n-1)h(a_2\; x)}{1-x}}dx=\int_0^1{\frac{((\frac{x}{a_2})^n-1)h(x)}{a_2-x}dx}.\label{SSmeldef4}
\end{eqnarray}

\noindent
In (\ref{SSmeldef4}) both extensions
$\int_0^{1/a_2} \frac{x^n-1}{1-x} h(a_2 x) dx$ and
$\int_0^1{ \frac{[({x}/{a_2})^n-1]}{a_2-x} h(x)}dx$ are equivalent, which can be
seen
using a simple substitution. Since from an algorithmic point of view it is easier to handle the second integral, we prefer this representation.
From now on we will call the extended and modified Mellin transform $M^+$ just Mellin transform and we will write $M$ instead of $M^+.$

For later considerations the following property (see Section~\ref{Sec:AnalyticCont}) will be helpful.

\begin{lemma}\label{Lemma:MHPLCType}
Let $m_i\in\R\setminus(0,1)$, $a\in\R^*$ and $m\in\{0,1\}$ and define
$f(z)=\M{\frac{H_{m_1,\dots,m_k}(x)}{(a\pm x)^m}}{z}$. Then $f(z)$, $f(2z)$ and $f(2z+1)$ are of C-type.
\end{lemma}

\subsection{Calculating the Mellin Transform in terms of $S$-sums}

Subsequently, we solve the following problem:\\
\textbf{Given} ${\ve m}=(m_1,\dots,m_k)$ with $m_i\in\R\setminus(0,1)$, $a\in\R^*$ and $m\in\{0,1\}$.\\
\textbf{Find} an expression $F(n)$ given as a linear combination of $S$-sums such that for all $n\in\N$ we have
\begin{equation}\label{Equ:MellRep}
\M{\frac{H_{\ve m}(x)}{(a\pm x)^m}}{n}=F(n).
\end{equation}

Note that this construction represents the Mellin transform of weighted generalized polylogarithms only for positive integer points. However, the uniquely determined analytic continuation (among all the C-type functions) for even and odd $S$-sums describes the Mellin transform in the complex plane; further details are given below of~\eqref{Equ:SSumLinCombForMHPL}.

In order to solve the specified problem,
we enhance the ideas of~\cite{Remiddi:1999ew} for
analytic aspects from the class of harmonic polylogarithms to the class of generalized polylogarithms.

Namely, due to the following Lemma we are able to reduce the calculation of the Mellin transform of generalized polylogarithms with indices in
$\R\setminus(0,1)$ weighted by $1/(a + x)$ or $1/(a-x)$ to the calculation of the Mellin transform of generalized polylogarithms with indices
in $\R\setminus(0,1)$ which are not weighted, \ie to the calculation of expressions of the form $\M{\H{\ve m}{x}}n.$

\begin{lemma}\label{Lemma:Mellin:ToNonWeighted}
For $n\in \N,$ $\ve m~=~(m_1,\overline{\ve{m}})~=~(m_1,m_2,\ldots,m_k)$ with $m_i\in \R \setminus (0,1)$ and
$a \in (0,\infty)$, $a_1\in (1,\infty)$, $a_2 \in (0,1)$, we have
\begin{eqnarray*}
\M{\frac{1}{a+x}}n &=&\left\{
	\begin{array}{ll}
		(-a)^n\left(\S{1}{-\frac{1}{a};n}+\H{-a}{1}\right),& \textnormal{if } 0< a< 1\\
		(-a)^n\left(\S{1}{-\frac{1}{a};n}-\S{1}{-\frac{1}{a};\infty}\right),& \textnormal{if } 1\geq a\\
		 \end{array} \right.\\
\M{\frac{1}{a-x}}n &=&\left\{
	\begin{array}{ll}
		-\S{1}{\frac{1}{a};n},& \textnormal{if } 0< a\leq 1\\
		a^n\left(-\S{1}{\frac{1}{a};n}+\S{1}{\frac{1}{a};\infty}\right),& \textnormal{if } 1<a
		 \end{array} \right.\\
\M{\frac{\H{\ve m}{x}}{a+x}}n &=&-n\;\M{\H{-a,\ve m}{x}}{n-1}+\H{-a,\ve m}{1},\\
\M{\frac{\H{\ve m}{x}}{1-x}}n &=&-n\;\M{\H{1,\ve m}{x}}{n-1},\\
\M{\frac{\H{\ve m}{x}}{a_1-x}}n &=&-n\;\M{\H{a_1,\ve m}{x}}{n-1}+\H{a_1,\ve m}1,\\
\M{\frac{\H{m_1,\overline{\ve{m}}}{x}}{a_2-x}}n &=&\left\{
		\begin{array}{ll}			 -\H{m_1,\overline{\ve{m}}}1\S{1}{\frac{1}{a_2};n}-\sum_{i=1}^{n}{\frac{\left(\frac{1}{a_2}\right)^i\Mma{\frac{\Hma{\overline{\ve{m}}}{x}}{\abs{m_1}+x}}i}{i}},& \textnormal{if } m_1< 0\\	 -\H{0,\overline{\ve{m}}}1\S{1}{\frac{1}{a_2};n}+\sum_{i=1}^{n}{\frac{\left(\frac{1}{a_2}\right)^i\Mma{\frac{\Hma{\overline{\ve{m}}}{x}}{x}}i}{i}},&  \textnormal{if } m_1=0  \\	 \sum_{i=1}^{n}{\frac{\left(\frac{1}{a_2}\right)^i\Mma{\frac{\Hma{\overline{\ve{m}}}{x}}{1-x}}i}{i}},& \textnormal{if } m_1 = 1, \\	 -\H{m_1,\overline{\ve{m}}}1\S{1}{\frac{1}{a_2};n}+\sum_{i=1}^{n}{\frac{\left(\frac{1}{a_2}\right)^i\Mma{\frac{\Hma{\overline{\ve{m}}}{x}}{m_1-x}}i}{i}},& \textnormal{if } m_1 > 1,
		 \end{array} \right.
\end{eqnarray*}
where the arising constants on the right hand side are finite.
\label{SSmelweighted}
\end{lemma}
\begin{proof}
For $a>0$ we get
\begin{eqnarray*}
\M{\frac{1}{a+x}}n
	 &=&\int_0^1{\frac{x^n}{a+x}dx}=\int_0^1{\frac{x^n-(-a)^n}{a+x}dx}+\int_0^1{\frac{(-a)^n}{a+x}dx}\\
	 &=&\int_0^1{(-a)^{n-1}\sum_{i=0}^{n-1}{\frac{x^i}{(-a)^i}dx}}+(-a)^n\H{-a}{1}\\
	&=&(-a)^{n-1}\sum_{i=0}^{n-1}{\frac{1}{(-a)^i(i+1)}}+(-a)^n\H{-a}{1}\\
	&=&(-a)^n\S{1}{-\frac{1}{a};n}+(-a)^n\H{-a}{1}.
\end{eqnarray*}
If $a \geq 1$ this is equal to $(-a)^n\S{1}{-\frac{1}{a};n}-\S{1}{-\frac{1}{a};\infty}.$\\
For $0<a\leq1$ we get
\begin{eqnarray*}
\M{\frac{1}{a-x}}n
	 &=&\frac{1}{a^n}\int_0^1{\frac{x^n-a^n}{a-x}dx}=-\frac{1}{a^n}\int_0^1{a^{n-1}\sum_{i=0}^{n-1}\left(\frac{x}{a}\right)^idx}\\
	 &=&-\frac{1}{a}\sum_{i=0}^{n-1}\frac{1}{a^i}\int_0^1{x^idx}=-\frac{1}{a}\sum_{i=0}^{n-1}\frac{1}{a^i(1+i)}\\
	&=&-\S{1}{\frac{1}{a};n}.
\end{eqnarray*}
For $a>1$ we get
\begin{eqnarray*}
\M{\frac{1}{a-x}}n
	 &=&\int_0^1{\frac{x^n}{a-x}dx}=\int_0^1{\frac{x^n-a^n}{a-x}dx}+a^n\int_0^1{\frac{1}{a-x}dx}\\
	 &=&-a^n\S{1}{\frac{1}{a};n}+a^n\H{a}{1}=-a^n\S{1}{\frac{1}{a};n}+a^n\S{1}{\frac{1}{a};\infty}.
\end{eqnarray*}
For $a>0$ we get
\begin{eqnarray*}
\int_0^1 x^n \H{-a,\ve m}{x}dx
		&=&\frac{\H{-a,\ve m}1}{n+1}-\frac{1}{n+1}\int_0^1 \frac{x^{n+1}}{(a+x)}\H{\ve m}{x}dx\\
		&=&\frac{\H{-a,\ve m}1}{n+1}-\frac{1}{n+1}\M{\frac{\H{\ve m}{x}}{a+x}}{n+1}.
\end{eqnarray*}
Hence we get
\begin{eqnarray*}
\M{\frac{\H{\ve m}{x}}{a+x}}{n+1} &=&-(n+1)\M{\H{-a,\ve m}{x}}{n}+\H{-a,\ve m}{1}.
\end{eqnarray*}
Similarly we get
\begin{eqnarray*}
\int_0^1 x^n \H{1,\ve m}{x}dx
		&=&\frac{1}{n+1}\lim_{\epsilon \rightarrow 1}\left(\epsilon^{n+1}\H{1,\ve m}{\epsilon}-\int_0^{\epsilon}
			{\frac{x^{n+1}-1}{1-x}\H{\ve m}x dx}+\H{1,\ve m}{\epsilon} \right)\\
		&=&\frac{1}{n+1}\int_0^{1}{\frac{x^{n+1}-1}{1-x}\H{\ve m}x dx}\\
		&=&\frac{1}{n+1}\M{\frac{\H{1,\ve m}{x}}{1-x}}{n+1}.
\end{eqnarray*}
And hence
\begin{eqnarray}
\M{\frac{\H{\ve m}{x}}{1-x}}{n+1} &=&-(n+1)\M{\H{1,\ve m}{x}}{n}.
\end{eqnarray}
For  $a_1\in (1,\infty)$ we get
\begin{eqnarray*}
\int_0^1 x^n \H{a_1,\ve m}{x}dx
		&=&\frac{\H{a_1,\ve m}1}{n+1}-\frac{1}{n+1}\int_0^1 \frac{x^{n+1}}{(a_1-x)}\H{\ve m}{x}dx\\
		&=&\frac{\H{a_1,\ve m}1}{n+1}-\frac{1}{n+1}\M{\frac{\H{\ve m}{x}}{a_1-x}}{n+1}.
\end{eqnarray*}
Hence we get
\begin{eqnarray*}
\M{\frac{\H{\ve m}{x}}{a_1-x}}{n+1} &=&-(n+1)\M{\H{a_1,\ve m}{x}}{n}+\H{a_1,\ve m}{1}.
\end{eqnarray*}
For $a_2 \in (0,1)$ and $m_1<0$ we get
\begin{eqnarray*}
\M{\frac{\H{m_1,\ve m}{x}}{a_2-x}}n &=&\frac{1}{a_2^n}\int_0^1{\frac{(x^n-a_2^n)\H{m_1,\ve m}{x}}{a_2-x}dx}\\
	&=&\left.-\H{m_1,\ve m}{x}\sum_{i=1}^n{\frac{\left(\frac{x}{a_2}\right)^i}{i}}\right|_0^1+\int_0^1{\frac{\H{\ve m}x}{\abs{m_1}+x}\sum_{i=1}^n{\frac{\left(\frac{x}{a_2}\right)^i}{i}}dx}\\
	&=&-\H{m_1,\ve m}1\S{1}{\frac{1}{a_2};n}+\sum_{i=1}^n{\frac{\left(\frac{1}{a_2}\right)^i}{i}}\int_0^1{\frac{x^i\H{\ve m}x}{\abs{m_1}+x}dx}\\
	&=&-\H{m_1,\ve m}1\S{1}{\frac{1}{a_2};n}+\sum_{i=1}^n{\frac{\left(\frac{1}{a_2}\right)^i}{i}}\M{\frac{\H{\ve m}x}{\abs{m_1}+x}}{i},
\end{eqnarray*}
\begin{eqnarray*}
\M{\frac{\H{0,\ve m}{x}}{a_2-x}}n &=&\frac{1}{a_2^n}\int_0^1{\frac{(x^n-a_2^n)\H{0,\ve m}{x}}{a_2-x}dx}\\
	&=&\left.-\H{0,\ve m}{x}\sum_{i=1}^n{\frac{\left(\frac{x}{a_2}\right)^i}{i}}\right|_0^1+\int_0^1{\frac{\H{\ve m}x}{x}\sum_{i=1}^n{\frac{\left(\frac{x}{a_2}\right)^i}{i}}dx}\\
	&=&-\H{0,\ve m}1\S{1}{\frac{1}{a_2};n}+\sum_{i=1}^n{\frac{\left(\frac{1}{a_2}\right)^i}{i}}\int_0^1{\frac{x^i\H{\ve m}x}{x}dx}\\
	&=&-\H{0,\ve m}1\S{1}{\frac{1}{a_2};n}+\sum_{i=1}^n{\frac{\left(\frac{1}{a_2}\right)^i}{i}}\M{\frac{\H{\ve m}x}{x}}{i},
\end{eqnarray*}
and
\begin{eqnarray*}
\M{\frac{\H{1,\ve m}{x}}{a_2-x}}n &=&\frac{1}{a_2^n}\lim_{\epsilon \rightarrow 1} \int_0^{\epsilon}{\frac{(x^n-a_2^n)\H{m_1,\ve m}{x}}{a_2-x}dx}\\
	&=&\lim_{\epsilon \rightarrow 1}\left(\left.-\H{1,\ve m}{x}\sum_{i=1}^n{\frac{\left(\frac{x}{a_2}\right)^i}{i}}\right|_0^\epsilon+\int_0^\epsilon{\frac{\H{\ve m}x}{1-x}\sum_{i=1}^n{\frac{\left(\frac{x}{a_2}\right)^i}{i}}dx}\right)\\
	&=&\lim_{\epsilon \rightarrow 1}\left(-\H{1,\ve m}{\epsilon}\S{1}{\frac{\epsilon}{a_2};n}+\sum_{i=1}^n{\frac{\left(\frac{1}{a_2}\right)^i}{i}}\int_0^\epsilon{\frac{x^i\H{\ve m}x}{1-x}dx}\right)\\
	&=&\lim_{\epsilon \rightarrow 1}\Biggl(-\H{1,\ve m}{\epsilon}\S{1}{\frac{\epsilon}{a_2};n}+\S{1}{\frac{1}{a_2};n}\H{1,\ve m}{\epsilon}\\&&+\sum_{i=1}^n{\frac{\left(\frac{1}{a_2}\right)^i}{i}}\int_0^\epsilon{\frac{(x^i-1)\H{\ve m}x}{1-x}dx}\Biggr)\\
	&=&\sum_{i=1}^n{\frac{\left(\frac{1}{a_2}\right)^i}{i}}\M{\frac{\H{\ve m}x}{1-x}}{i}.
\end{eqnarray*}
For $a_2 \in (0,1)$ and $m_1>0$ we get
\begin{eqnarray*}
\M{\frac{\H{m_1,\ve m}{x}}{a_2-x}}n &=&\frac{1}{a_2^n}\int_0^1{\frac{(x^n-a_2^n)\H{m_1,\ve m}{x}}{a_2-x}dx}\\
	&=&\left.-\H{m_1,\ve m}{x}\sum_{i=1}^n{\frac{\left(\frac{x}{a_2}\right)^i}{i}}\right|_0^1+\int_0^1{\frac{\H{\ve m}x}{m_1-x}\sum_{i=1}^n{\frac{\left(\frac{x}{a_2}\right)^i}{i}}dx}\\
	&=&-\H{m_1,\ve m}1\S{1}{\frac{1}{a_2};n}+\sum_{i=1}^n{\frac{\left(\frac{1}{a_2}\right)^i}{i}}\int_0^1{\frac{x^i\H{\ve m}x}{m_1-x}dx}\\
	&=&-\H{m_1,\ve m}1\S{1}{\frac{1}{a_2};n}+\sum_{i=1}^n{\frac{\left(\frac{1}{a_2}\right)^i}{i}}\M{\frac{\H{\ve m}x}{m_1-x}}{i}.
\end{eqnarray*}
\end{proof}

To calculate a non-weighted Mellin transform $\M{\H{\ve m}{x}}n$
we proceed by recursion. First let us state the base cases, \ie the Mellin transforms of generalized polylogarithms with depth 1.
\begin{lemma}\label{Lemma:MellinBaseCase} For $n\in \N$, $a \in \R \setminus (0,1)$  we have
\begin{eqnarray*}
\M{\H{a}{x}}n=\left\{
		 \begin{array}{ll}
			\frac{-1}{(n+1)^2}\left(1+a^{n+1}(n+1)\S{1}{\frac{1}{a};n}\right.\\
				\hspace{1.2cm}\left.-(a^{n+1}-1)(n+1)\S{1}{\frac{1}{a};\infty}  \right),& \textnormal{if } a \leq -1
\vspace{0.3cm}\\
			\frac{-1}{(n+1)^2}\left(1+a^{n+1}(n+1)\S{1}{\frac{1}{a};n}\right.\\
				\hspace{1.2cm}\left.+(a^{n+1}-1)(n+1)\H{a}1\right),&  \textnormal{if } -1 < a < 0
\vspace{0.3cm}\\
			\frac{-1}{(n+1)^2},&  \textnormal{if } a=0
\vspace{0.3cm}\\
			\frac{1}{(n+1)^2}\left(1+(n+1)\S{1}{n}\right),&  \textnormal{if } a=1
\vspace{0.3cm}\\
			\frac{1}{(n+1)^2}\left(1+a^{n+1}(n+1)\S{1}{\frac{1}{a};n}\right.\\
				\hspace{1.2cm}\left.-(a^{n+1}-1)(n+1)\S{1}{\frac{1}{a};\infty}  \right),& \textnormal{if } a > 1
		 \end{array} \right.
\end{eqnarray*}
where the arising constants are finite.
\label{SSweight1mel}
\end{lemma}
\begin{proof}
First let $a\leq-1.$ By integration by parts we get:
\begin{eqnarray*}
\M{\H{a}{x}}n&=&\int_0^1{x^n\H{a}x dx}=\left.\frac{x^{n+1}}{n+1}\H{a}{x}\right|_0^1-\int_0^1{\frac{x^{n+1}}{n+1}\frac{1}{\abs{a}+x}dx}\\
&=&\frac{\H{a}1}{n+1}-\frac{1}{n+1}\left(\int_0^1{\frac{x^{n+1}-a^{n+1}}{\abs{a}+x}dx}+a^{n+1}\int_0^1{\frac{1}{\abs{a}+x}dx}\right)\\
&=&\frac{\H{a}1}{n+1}(1-a^{n+1})-\frac{1}{n+1}\int_0^1{a^n\sum_{i=0}^n\frac{x^{i}}{a^i}dx}\\
&=&\frac{\H{a}1}{n+1}(1-a^{n+1})-\frac{a^n}{n+1}\sum_{i=0}^n\frac{1}{a^i(i+1)}\\
&=&\frac{-\S{1}{\frac{1}{a};\infty}}{n+1}(1-a^{n+1})-\frac{a^{n+1}}{n+1}\S{1}{\frac{1}{a};n+1}\\
&=&\frac{-1}{n+1}\left(\frac{1}{n+1}+a^{n+1}\S{1}{\frac{1}{a};n}-(a^{n+1}-1)\S{1}{\frac{1}{a};\infty}\right).
\end{eqnarray*}
For $-1<a<0$ we obtain:
\begin{eqnarray*}
\M{\H{a}{x}}n&=&\int_0^1{x^n\H{a}x dx}=\left.\frac{x^{n+1}}{n+1}\H{a}{x}\right|_0^1-\int_0^1{\frac{x^{n+1}}{n+1}\frac{1}{\abs{a}+x}dx}\\
&=&\frac{\H{a}1}{n+1}-\frac{1}{n+1}\left(\int_0^1{\frac{x^{n+1}-a^{n+1}}{\abs{a}+x}dx}+a^{n+1}\int_0^1{\frac{1}{\abs{a}+x}dx}\right)\\
&=&\frac{\H{a}1}{n+1}(1-a^{n+1})-\frac{1}{n+1}\int_0^1{a^n\sum_{i=0}^n\frac{x^{i}}{a^i}dx}\\
&=&\frac{\H{a}1}{n+1}(1-a^{n+1})-\frac{a^n}{n+1}\sum_{i=0}^n\frac{1}{a^i(i+1)}\\
&=&\frac{\H{a}1}{n+1}(1-a^{n+1})-\frac{a^{n+1}}{n+1}\S{1}{\frac{1}{a};n+1}\\
&=&\frac{-1}{(n+1)^2}\left(1+a^{n+1}(n+1)\S{1}{\frac{1}{a};n}+(a^{n+1}-1)(n+1)\H{a}1\right).
\end{eqnarray*}
For $a=0$ it follows that
\begin{eqnarray}
\M{\H{0}{x}}n&=&\int_0^1{x^n\H{0}x dx}=\left.\frac{x^{n+1}}{n+1}\H{0}{x}\right|_0^1-\int_0^1{\frac{x^{n+1}}{n+1}\frac{1}{x}dx}\nonumber\\
&=&-\int_0^1{\frac{x^n}{n+1}dx}=\left.-\frac{x^{n+1}}{(n+1)^2}\right|_0^1=-\frac{1}{(n+1)^2}.\nonumber
\end{eqnarray}
For $a=1$ we get:
\begin{eqnarray}
\int_0^1 x^n \H{1}{x}dx
		&=&\lim_{\epsilon \rightarrow 1} \int_0^{\epsilon} x^n \H{1}{x}dx \nonumber\\
		&=&\lim_{\epsilon \rightarrow 1}\left( \left. \frac{x^{n+1}}{n+1}\H{1}{x}\right|_0^{\epsilon}-\int_0^{\epsilon}
			\frac{x^{n+1}}{(n+1)(1-x)}dx \right) \nonumber\\
		&=&\frac{1}{n+1}\lim_{\epsilon \rightarrow 1}\left(\epsilon^{n+1}\H{1}{\epsilon}-\int_0^{\epsilon}
			{\frac{x^{n+1}-1}{1-x}x dx}-\H{1}{\epsilon} \right) \nonumber\\
		&=&\frac{1}{n+1}\left(\lim_{\epsilon \rightarrow 1}(\epsilon^{n+1}-1)\H{1}{\epsilon}+\lim_{\epsilon \rightarrow 1}
			{\int_0^{\epsilon}\sum_{i=0}^n{x^i}dx}\right)\nonumber\\
		&=&\frac{1}{n+1}\left(0+\int_0^{1}\sum_{i=0}^n{x^i}dx\right)\nonumber\\
		&=&\frac{1}{n+1} \sum_{i=0}^n{\frac{1}{i+1}}=\frac{1}{(n+1)^2}\left(1+(n+1)\S{1}{n}\right). \nonumber
\end{eqnarray}
Finally for $a > 1$ we conclude that
\begin{eqnarray*}
\M{\H{a}{x}}n&=&\int_0^1{x^n\H{a}x dx}=\left.\frac{x^{n+1}}{n+1}\H{a}{x}\right|_0^1-\int_0^1{\frac{x^{n+1}}{n+1}\frac{1}{a-x}dx}\\
&=&\frac{\H{a}1}{n+1}-\frac{1}{n+1}\left(\int_0^1{\frac{x^{n+1}-a^{n+1}}{a-x}dx}+a^{n+1}\int_0^1{\frac{1}{a-x}dx}\right)\\
&=&\frac{\H{a}1}{n+1}(1-a^{n+1})+\frac{1}{n+1}\int_0^1{a^n\sum_{i=0}^n\frac{x^{i}}{a^i}dx}\\
&=&\frac{\H{a}1}{n+1}(1-a^{n+1})+\frac{a^n}{n+1}\sum_{i=0}^n\frac{1}{a^i(i+1)}\\
&=&\frac{\S{1}{\frac{1}{a};\infty}}{n+1}(1-a^{n+1})+\frac{a^{n+1}}{n+1}\S{1}{\frac{1}{a};n+1}\\
&=&\frac{1}{n+1}\left(\frac{1}{n+1}+a^{n+1}\S{1}{\frac{1}{a};n}-(a^{n+1}-1)\S{1}{\frac{1}{a};\infty}  \right).
\end{eqnarray*}
\end{proof}

The higher depth results for $\M{\H{\ve m}{x}}n$ can now be obtained by recursion:
\begin{lemma}\label{Lemma:MellinRecursion} For $n\in \N$, $a \in \R\setminus(0,1)$ and $\ve m \in (\R\setminus(0,1))^k$,
\begin{eqnarray*}
\M{\H{a,\ve m}{x}}n=\left\{
		\begin{array}{ll}
			\frac{(1-a^{n+1})\H{a,\ve m}1}{n+1}-\frac{a^{n}}{n+1}\sum_{i=0}^{n}{\frac{\M{\H{\ve m}{x}}i}{a^i}},& \textnormal{if } a < 0\\
			\frac{\H{0,\ve m}1}{n+1}-\frac{1}{n+1} \M{\H{\ve m}{x}}n,&  \textnormal{if } a=0  \\
			\frac{1}{n+1} \sum_{i=0}^n{\M{\H{\ve m}{x}}n},&  \textnormal{if } a=1  \\
			\frac{(1-a^{n+1})\H{a,\ve m}1}{n+1}+\frac{a^{n}}{n+1}\sum_{i=0}^{n}{\frac{\M{\H{\ve m}{x}}i}{a^i}},& \textnormal{if } a > 1
		 \end{array} \right.
\end{eqnarray*}
where the arising constants are finite.
\label{SSmelnotweighted}
\end{lemma}
\begin{proof}
We get the following results by integration by parts. For $a<0$ we get:
\begin{eqnarray*}
&&\M{\H{a,\ve m}{x}}n
=\int_0^1{x^n\H{a,\ve m}x dx}=\left.\frac{x^{n+1}}{n+1}\H{a,\ve m}{x}\right|_0^1-\int_0^1{\frac{x^{n+1}\H{\ve m}x}{n+1}\frac{1}{\abs{a}+x}dx}\\
&&\hspace{2cm}=\frac{\H{a,\ve m}1}{n+1}-\frac{1}{n+1}\left(\int_0^1{\frac{x^{n+1}-a^{n+1}}{\abs{a}+x}\H{\ve m}x dx}+a^{n+1}\int_0^1{\frac{\H{\ve m}x}{\abs{a}+x}dx}\right)\\
&&\hspace{2cm}=\frac{\H{a,\ve m}1}{n+1}(1-a^{n+1})-\frac{1}{n+1}\int_0^1{a^n\sum_{i=0}^n\frac{x^{i}\H{\ve m}x}{a^i}dx}\\
&&\hspace{2cm}=\frac{\H{a,\ve m}1}{n+1}(1-a^{n+1})-\frac{a^n}{n+1}\sum_{i=0}^n\frac{1}{a^i}\M{\H{\ve m}x}i.
\end{eqnarray*}
For $a=0$ we get:
\begin{eqnarray*}
\int_0^1 x^n \H{0,\ve m}{x}dx &=& \left. \frac{x^{n+1}}{n+1}\H{0,\ve m}{x}\right|_0^1-\int_0^1 \frac{x^n}{n+1}\H{\ve m}{x}dx\\
				&=&\frac{\H{0,\ve m}1}{n+1}-\frac{1}{n+1} \M{\H{\ve m}{x}}n.
\end{eqnarray*}
For $a=1$ we it follows that
\begin{eqnarray}
\int_0^1 x^n \H{1,\ve m}{x}dx
		&=&\lim_{\epsilon \rightarrow 1} \int_0^{\epsilon} x^n \H{1,\ve m}{x}dx \nonumber\\
		&=&\lim_{\epsilon \rightarrow 1}\left( \left. \frac{x^{n+1}}{n+1}\H{1,\ve m}{x}\right|_0^{\epsilon}-\int_0^{\epsilon}
			\frac{x^{n+1}}{(n+1)(1-x)}\H{\ve m}{x}dx \right) \nonumber\\
		&=&\frac{1}{n+1}\lim_{\epsilon \rightarrow 1}\left(\epsilon^{n+1}\H{1,\ve m}{\epsilon}-\int_0^{\epsilon}
			{\frac{x^{n+1}-1}{1-x}\H{\ve m}x dx}-\H{1,\ve m}{\epsilon} \right) \nonumber\\
		&=&\frac{1}{n+1}\left(\lim_{\epsilon \rightarrow 1}(\epsilon^{n+1}-1)\H{1,\ve m}{\epsilon}+\lim_{\epsilon \rightarrow 1}
			{\int_0^{\epsilon}\sum_{i=0}^n{x^i\H{\ve m}x}dx}\right)\nonumber\\
		&=&\frac{1}{n+1}\left(0+\int_0^{1}\sum_{i=0}^n{x^i\H{\ve m}x}dx\right)\nonumber\\
		&=&\frac{1}{n+1} \sum_{i=0}^n{\M{\H{\ve m}{x}}i}. \nonumber
\end{eqnarray}
For $a>1$ we conclude that
\begin{eqnarray*}
&&\M{\H{a,\ve m}{x}}n
=\int_0^1{x^n\H{a,\ve m}x dx}=\left.\frac{x^{n+1}}{n+1}\H{a,\ve m}{x}\right|_0^1-\int_0^1{\frac{x^{n+1}\H{\ve m}x}{n+1}\frac{1}{a-x}dx}\\
&&\hspace{2cm}=\frac{\H{a,\ve m}1}{n+1}-\frac{1}{n+1}\left(\int_0^1{\frac{x^{n+1}-a^{n+1}}{a-x}\H{\ve m}x dx}+a^{n+1}\int_0^1{\frac{\H{\ve m}x}{a-x}dx}\right)\\
&&\hspace{2cm}=\frac{\H{a,\ve m}1}{n+1}(1-a^{n+1})+\frac{1}{n+1}\int_0^1{a^n\sum_{i=0}^n\frac{x^{i}\H{\ve m}x}{a^i}dx}\\
&&\hspace{2cm}=\frac{\H{a,\ve m}1}{n+1}(1-a^{n+1})+\frac{a^n}{n+1}\sum_{i=0}^n\frac{1}{a^i}\M{\H{\ve m}x}i.
\end{eqnarray*}
\end{proof}

Summarizing, using Lemma~\ref{SSmelweighted} together with Lemma~\ref{SSweight1mel} and Lemma~\ref{SSmelnotweighted} we are able to calculate the Mellin transform of
generalized polylogarithms with indices in $\R \setminus (0,1).$ In addition, the polylogarithms can be weighted by $1/(a + x)$ or $1/(a-x)$ for $a\in\R.$
These Mellin transforms can be expressed using $S$-sums.

As an example we compute the Mellin transform of  $$\frac{\textnormal{H}_{0,-2}(x)}{x+2}.$$
Using Lemma~\ref{Lemma:Mellin:ToNonWeighted} we can write
\begin{eqnarray*}
 \M{\frac{\H{0,-2}x}{x+2}}{n}=-n\cdot\M{\H{-2,0,-2}x}{n-1}+\H{-2,0,-2}1.
\end{eqnarray*}
Now we apply Lemma~\ref{Lemma:MellinRecursion} which leads to
\begin{eqnarray*}
\M{\H{-2,0,-2}x}{n-1}=\frac{1-(-2)^{n}}{n}\H{-2,0,-2}1-\frac{(-2)^{n-1}}{n}\sum_{i=0}^{n-1}\frac{\M{\H{0,-2}x}{i}}{(-2)^i}.
\end{eqnarray*}
Again using Lemma~\ref{Lemma:MellinRecursion} we find
\begin{eqnarray*}
\M{\H{0,-2}x}{i}=\frac{1}{i+1}\H{0,-2}1-\frac{1}{i+1}\sum_{i=0}^n\M{\H{-2}x}{i}.
\end{eqnarray*}
Finally applying Lemma~\ref{Lemma:MellinBaseCase} we get
\begin{eqnarray*}
\M{\H{-2}x}{i}=-\frac{\left(1+(-2)^{i+1}(i+1)\S{1}{-\tfrac{1}{2};i}-((-2)^{i+1}-1)(i+1)\S{1}{-\tfrac{1}{2};\infty}\right)}{(i+1)^2}.
\end{eqnarray*}
Putting this together yields
\begin{align*}
 \M{\frac{\H{0,-2}x}{x+2}}{n}=&(-2)^n \textnormal{H}_{-2,0,-2}(1)\\
 &\hspace*{-3cm}+(-2)^{n-1} \sum _{i=0}^{n-1}\tfrac{(-2)^{-i}}{i+1} \Big({H}_{0,-2}(1)+\tfrac{1+(-2)^{1+i} (1+i) \textnormal{S}_1\left(-\frac{1}{2};i\right)-\left(-1+(-2)^{1+i}\right) (1+i) \textnormal{S}_1\left(-\frac{1}{2};\infty \right)}{(1+i)^2}\Big).
\end{align*}
To this end, transforming this expression to $S$-sums we arrive at
\begin{equation}\label{Exp:Mellin}
\begin{split}
 \M{\frac{\H{0,-2}x}{x+2}}{n}=&-(-1)^n(2)^n \Big[\textnormal{H}_{-2,0,-2}(1)+\textnormal{H}_{0,-2}(1) \textnormal{S}_1\left(-\tfrac{1}{2};n\right)-\textnormal{S}_2(n) \textnormal{S}_1\left(-\tfrac{1}{2};\infty \right)\\
			      &+\textnormal{S}_1\left(-\tfrac{1}{2};\infty \right) \textnormal{S}_2\left(-\tfrac{1}{2};n\right)+\textnormal{S}_{2,1}\left(1,-\tfrac{1}{2};n\right)\Big].
\end{split}
\end{equation}

\medskip

Looking at the presented algorithm the representation $F(n)$ in~\eqref{Equ:MellRep} will be of the form
\begin{equation}\label{Equ:SSumLinCombForMHPL}
F(n)=(-1)^n\Big[r_0a_0^nf_0(n)+\dots+r_da_d^nf_d(n)\Big]+\Big[r_{d+1}a_{d+1}^nf_{d+1}(n)+\dots+r_ea_e^nf_e(n)\Big]
\end{equation}
with $d\leq e$, $r_i\in\set R$, $a_i\in\set R^{+}$ and where each $f_i(n)$ stands for
a $S$-sum.

Consider the C-type function $f(z):=\M{\frac{H_{\ve m}(x)}{(a\pm x)^m}}{z}$; see Lemma~\ref{Lemma:MHPLCType}.
We remark that the right hand side of~\eqref{Equ:MellRep} for even and odd $n$ has
also an analytic continuation by replacing the $S$-sums with the integral
representation of Theorem~\ref{SSintrep}. In particular the function is of C-type due to Proposition~\ref{Prop:SCType} and~Lemma~\ref{Lemma:Closure}. Because of Carlson's Theorem it follows that the constructed function agrees with $f(z)$. In other words, the representation of the Mellin transform with $S$-sums for integer values gives the representation in the complex plane.

\subsection{Calculating the Inverse Mellin Transform in terms of generalized polylogarithms}
\label{SSInvMellin}

\noindent
Subsequently, we consider a subset of the generalized harmonic polylogarithms, namely the set
\begin{eqnarray*}
\bar{H}:=\biggl\{\H{m_1,m_2,\ldots,m_k}x\biggl| m_i \in \R, \abs{\frac{1}{m_{j_k}}}\leq 1 \textnormal{ and } \abs{\frac{m_{j_i}}{m_{j_{i-1}}}}\leq 1 \textnormal{ for } 2\leq i\leq k \textnormal{ where} \\
\{m_{j_1},m_{j_2},\ldots,m_{j_k}\} \textnormal{ is the set of nonzero indices, and } j_u<j_v \textnormal{ if } u<v \biggr\}
\end{eqnarray*}
and call the elements $\bar{H}$-polylogarithms; note that $m_i\in\set R\setminus(-1,1)$ with $m_i\neq0$. For this subclass,
it turns out that the Mellin transform of properly weighted generalized polylogarithms can be expressed using the $S$-sums from the set
\begin{equation}\label{Equ:SBar}
\begin{split}
&\bar{S}:=\biggl\{\S{a_1,a_2,\ldots,a_k}{b_1,b_2,\ldots,b_k;n}\biggl|a_i\in \Z^* \textnormal{ for } 1\leq i\leq k ;b_1\in \R^*;\\
  &\hspace{5.8cm} b_i \in [-1,1]\setminus \{0\}\textnormal{ for } 2\leq i\leq k\biggr\};
  \end{split}
\end{equation}
the elements of this set are also called $\bar{S}$-sums.

Restricting to this class, we solve the following problem:\\
\textbf{Given} an $\bar{S}$-sum $\S{\ve m}{\ve b; n}$.\\
\textbf{Find} a linear combination $F(n)$ of Mellin transforms of weighted generalized polylogarithms from $\bar{H}$ such that for all $n\in\set N$  we have
\begin{equation}\label{Equ:SToMRep}
\S{\ve m}{\ve b; n}=F(n).
\end{equation}
In other words,
there is a certain bijection between the set of properly weighted $\bar{H}-$polylogarithms and $\bar{S}$-sums via the Mellin-transform and its inverse Mellin transform.

In order to present the algorithm for the inverse Mellin transform, we define the
following order on $S$-sums.
Let $\S{\ve m_1}{\ve b_1; n}$ and $\S{\ve m_2}{\ve b_2;n}$ be $S$-sums with weights
$w_1$, $w_2$ and depths $d_1$ and $d_2,$ respectively. Then
$$
		  	\begin{array}{ll}
						\S{\ve m_1}{\ve b_1; n} \prec \S{\ve m_2}{\ve b_2;n}, \ \textnormal{if} \ w_1<w_2, & \textnormal{or } \ (w_1=w_2 \ \textnormal{and} \ d_1<d_2).
			\end{array}
$$
We say that a $S$-sum $s_1$ is more complicated than a $S$-sum $s_2$ if $s_2 \prec
s_1$.
For a set of $S$-sums we call a $S$-sum {\upshape most complicated} if it is a largest
element with respect to $\prec$.

Then the inverse Mellin transform from $\bar{S}$ to $\bar{H}$ relies on the following main properties:
\begin{lemma}\label{Lemma:MostComplicated}

\vspace*{1mm}\noindent
\begin{enumerate}
\item  In the Mellin transform of a generalized polylogarithm $\H{\ve m}x \in \bar{H}$ with $\ve m=(m_1,m_2,\ldots,m_k)$
weighted by $1/(c-x)$ or
$1/(c+x)$ ($c\in\R^*$) there is only one {\upshape most complicated}
$S$-sum $\S{\ve a}{\ve b; n}\in \bar{S}$ with $\ve a=(a_1,\ldots,a_l)$ and $\ve
b=(b_1,\ldots,b_l)$
,\ie
\begin{equation}\label{SSsinglmostcomp1}
\M{\frac{\H{\ve m}{x}}{c\pm x}}{n}=p\cdot\S{\ve a}{\ve b; n}+t
\end{equation}
where $p\in\R^*$ and $t$ consists of a linear combination of $\bar{S}$-sums over $\set R$ that are smaller than $\S{\ve a}{\ve b; n}$. Note that if $\H{\ve m}x$ is of weight $w,$ then $\S{\ve a}{\ve b; n}$ is of weight $w+1.$
In particular, if $m_k\neq1$ then $a_l\neq 1$ or $b_l\neq 1$; these properties will be used later in Section~\ref{Sec:AsymptoticExp}.
\item Given $\S{\ve a}{\ve b; n}\in \bar{S}$ with $\ve a=(a_1,\ldots,a_l)$ and $\ve b=(b_1,\ldots,b_l)$, there is an algorithm (see~\cite[Algorithms~3]{Ablinger:12}) that computes an
$\H{\ve m}{x}\in\bar{H}$ with $\ve m=(m_1,m_2,\ldots,m_k)$, $d\in\{-1,1\}$, and $c\in\set R^*$ such that $\S{\ve a}{\ve b; n}$ is the most complicated $S$-sum that occurs in $\M{\frac{\H{\ve m}{x}}{c+d\,x}}{n}$. Note that if $\S{\ve a}{\ve b; n}$ is of weight $w,$ then $\H{\ve m}{x}$ is of weight $w-1.$ In particular, if $a_l\neq1$ or $b_l\neq1$ then $m_k\neq1$. Moreover, if $b_i\in[-1,1]$ with $b_i\neq0$ for all $i$ then $c>1$; again these properties will be used later in Section~\ref{Sec:AsymptoticExp}.
\end{enumerate}
\end{lemma}
Note that parts of these properties have been exploited already in~\cite{Remiddi:1999ew}
to
construct the inverse Mellin transform between harmonic sums and harmonic polylogarithms.
The rigorous proofs for this special case given in~\cite{Ablinger:2010kw} carry over to
this more general case and are omitted here.

The computation of the inverse Mellin transform of a linear combination of $\bar{S}$-sums
now is straightforward (compare~\cite{Ablinger:2010kw,Remiddi:1999ew}):
\begin{itemize}
	\item Locate the most complicated $\bar{S}$-sum.
	\item Exploit property~(2) from above and compute an
$\bar{H}$-polylogarithm weighted by $1/(c-x)$ or $1/(c+x)$ such that the $\bar{S}$-sum is the {\itshape most complicated} $\bar{S}$-sum in the Mellin transform of this weighted
generalized polylogarithm.
	\item Add it and subtract it with an appropriate factor (see the next step).
	\item Perform the Mellin transform to the subtracted version. This will cancel
        the original $\bar{S}$-sum (by choosing the appropriate factor).
	\item Repeat the above steps until there are no more $\bar{S}$-sums.
	\item Let $c$ be the remaining constant term; replace it by
        $\textnormal{M}^{-1}(c)$, or equivalently, multiply $c$ by $\delta(1-x)$
        where $\delta(x)$ stands for the Dirac-$\delta$-distribution \cite{YOSIDA}.
\end{itemize}

In the following we will illustrate the computation of the inverse Mellin transformation by means of an example. We consider the
sum $\textnormal{S}_{2,1}\left(1,-\tfrac{1}{2};n\right)$. It is the most complicated $\bar{S}$-sum in the Mellin transformation of
$\frac{\textnormal{H}_{0,-2}(x)}{x+2}$ which can be computed by Algorithm~3 in~\cite{Ablinger:12}. So in a first step we get
$$\textnormal{S}_{2,1}\left(1,-\tfrac{1}{2};n\right)=\textnormal{S}_{2,1}\left(1,-\tfrac{1}{2};n\right)-\left(-\tfrac{1}{2}\right)^n\M{\tfrac{\textnormal{H}_{0,-2}(x)}{x+2}}{n} +\left(-\tfrac{1}{2}\right)^n\M{\tfrac{\textnormal{H}_{0,-2}(x)}{x+2}}{n}.$$
Calculating the first of the two Mellin transforms (see~\eqref{Exp:Mellin}) we get
\begin{align*}
\textnormal{S}_{2,1}\left(1,-\tfrac{1}{2};n\right)=&\left(-\tfrac{1}{2}\right)^n\M{\frac{\textnormal{H}_{0,-2}(x)}{x+2}}{n}+\textnormal{S}_2(n) \textnormal{S}_1\left(-\tfrac{1}{2};\infty \right)-\textnormal{S}_2\left(-\tfrac{1}{2};n\right) \textnormal{S}_1\left(-\tfrac{1}{2};\infty
   \right)\\
 &+\textnormal{S}_1\left(-\tfrac{1}{2};n\right) \textnormal{S}_2\left(-\tfrac{1}{2};\infty \right)+\textnormal{S}_3\left(-\tfrac{1}{2};\infty \right)-\textnormal{S}_{1,2}\left(-\tfrac{1}{2},1;\infty \right).
\end{align*}
In the right hand side of this expression we find a linear combination of the three $\bar{S}$-sums $\textnormal{S}_2\left(-\tfrac{1}{2};n\right),\textnormal{S}_2(n)$ and $\textnormal{S}_1\left(-\tfrac{1}{2};n\right)$, which we want to express using the
Mellin transformation of generalized polylogarithms. The most complicated $\bar{S}$-sums from these three is $\textnormal{S}_2\left(-\tfrac{1}{2};n\right)$. In particular, it is the most complicated $\bar{S}$-sum in the Mellin transformation of
$\frac{\textnormal{H}_0(x)}{x+2}.$
Hence following the above method we rewrite $\textnormal{S}_{2,1}\left(1,-\tfrac{1}{2};n\right)$ to the expression
\begin{align*} &-\left(\textnormal{S}_2\left(-\tfrac{1}{2};n\right)-\left(-\tfrac{1}{2}\right)^n\M{\frac{\textnormal{H}_0(x)}{x+2}}{n}+\left(-\tfrac{1}{2}\right)^n\M{\frac{\textnormal{H}_0(x)}{x+2}}{n}\right) \textnormal{S}_1\left(-\tfrac{1}{2};\infty
   \right)\\ &+\left(-\tfrac{1}{2}\right)^n\M{\frac{\textnormal{H}_{0,-2}(x)}{x+2}}{n}+\textnormal{S}_2(n) \textnormal{S}_1\left(-\tfrac{1}{2};\infty \right)+\textnormal{S}_1\left(-\tfrac{1}{2};n\right) \textnormal{S}_2\left(-\tfrac{1}{2};\infty \right)\\
  &+\textnormal{S}_3\left(-\tfrac{1}{2};\infty \right)-\textnormal{S}_{1,2}\left(-\tfrac{1}{2},1;\infty \right).
\end{align*}
Performing the first of these occurring Mellin transformations  yields
\begin{eqnarray*}
&&\left(-\tfrac{1}{2}\right)^n \M{\frac{\textnormal{H}_{0,-2}(x)}{x+2}}{n}+\left(-\tfrac{1}{2}\right)^n \textnormal{S}_1\left(-\tfrac{1}{2};\infty \right) \M{\frac{\textnormal{H}_0(x)}{x+2}}{n}+\textnormal{S}_2(n) \textnormal{S}_1\left(-\tfrac{1}{2};\infty \right)\\
&&-\textnormal{S}_{1,2}\left(-\tfrac{1}{2},1;\infty\right)+\textnormal{S}_1\left(-\tfrac{1}{2};n\right)
   \textnormal{S}_2\left(-\tfrac{1}{2};\infty \right)-\textnormal{S}_1\left(-\tfrac{1}{2};\infty \right) \textnormal{S}_2\left(-\tfrac{1}{2};\infty \right)+\textnormal{S}_3\left(-\tfrac{1}{2};\infty \right).
\end{eqnarray*}
Continuing this process for the remaining $\bar{S}$-sums we arrive at
\begin{equation}\label{Equ:InverseMellinExp}
\begin{split}
\textnormal{S}_{2,1}&\left(1,-\tfrac{1}{2};n\right)= \textnormal{S}_1\left(-\tfrac{1}{2};\infty \right)\M{\frac{\textnormal{H}_0(x)}{1-x}}{n}+\textnormal{S}_3\left(-\tfrac{1}{2};\infty \right)-
   \textnormal{S}_{1,2}\left(-\tfrac{1}{2},1;\infty \right)\\
&+(-1)^n\left(\tfrac{1}{2}\right)^n\Big[\textnormal{S}_1\left(-\tfrac{1}{2};\infty \right)
\M{\frac{\textnormal{H}_0(x)}{x+2}}{n}+ \textnormal{S}_2\left(-\tfrac{1}{2};\infty \right)
   \M{\frac{1}{x+2}}{n}+
      \M{\frac{\textnormal{H}_{0,-2}(x)}{x+2}}{n}\Big].
\end{split}
\end{equation}
We remark that there is a second way to compute the inverse Mellin transform of an $\bar{S}$-sum by using the integral representation of Theorem~\ref{SSintrep}: Repeated
suitable integration by parts of the given integral representation  provides the inverse Mellin transform in terms of integrals that can be expressed in terms of generalized polylogarithms.

\medskip

Looking at the presented algorithm the representation $F(n)$ in~\eqref{Equ:SToMRep} will be of the form
\begin{equation}\label{Equ:SSumToMellin}
F(n)=(-1)^n\Big[r_0a_0^nf_0(n)+\dots+r_da_d^nf_d(n)\Big]+\Big[r_{d+1}a_{d+1}^nf_{d+1}(n)+\dots+r_ea_e^nf_e(n)\Big]
\end{equation}
with $d\leq e$, $r_i\in\set R$, $a_i\in\set R^{+}$ and where each $f_i(n)$ stands for a Mellin transform of weighted harmonic polylogarithm from $\bar{H}$.

We remark that the right hand side of~\eqref{Equ:SToMRep} for even and odd $n$
gives rise to an analytic continuation. In particular the function is of C-type
due to Lemma~\ref{Lemma:MHPLCType} and~Lemma~\ref{Lemma:Closure}. Because of the
generalized Carlson Theorem it follows that the obtained function agrees with the
integral representation of the even, resp., odd $S$-sum given in Theorem~\ref{SSintrep}.

In many QCD applications\label{RemarksOnQCDAndXSpace} one starts with a complicated function $f(x)$ in $x$-space
where one knows that it can be expressed in terms of a linear combination of weighted
generalized polylogarithms (or related classes of them) over real numbers. It is
possible to choose either $n$ as an even or odd positive integer, depending on the
physics problem. In some cases the corresponding other choice may even yield a
vanishing result due to a symmetry inherent in the process. Then one
tactic is to perform the simplification on $\M{f(x)}{n}$ (e.g., with the symbolic summation packages~\cite{Ablinger:2010pb,Blumlein:2010zv,Blumlein:2012hg,Ablinger:2012ph}) and to express it in terms
of $S$-sums. Note that by Lemma~\ref{Lemma:MHPLCType} $\M{f(x)}{n}$ is of C-type.
Finally, applying the inverse Mellin transform to each of these $S$-sums and
combining these calculations, one obtains an expression $F(n)$ in the
format~\eqref{Equ:SToMRep}. Usually, the alternating sign and the powers $a_i^n$ vanish.
Since also the output function $F(n)$ can be analytically continued, the
continuation is of C-type and $\M{f(x)}{n}=F(n)$ for almost all integer $n$, their
analytic continuations agree by Carlson's Theorem. Now drop the Mellin operator
from $F(n)$ which gives a linear combination $h(x)$ of generalized polylogarithms such that $\M{h(x)}{n}=F(n)$. This finally yields that $f(x)=h(x)$.

\section{Differentiation of $S$-Sums}
\label{SSdifferentiation}

\vspace*{1mm}
\noindent As worked out in Section~\ref{Sec:AnalyticCont} the even and odd S-sums can be uniquely analytically continued within the class of C-type functions.
This allows one to consider
differentiation with respect to $n$. As it turns out the obtained expressions are
again integral
representations of $S$-sum expressions. This will make it possible to define the
differentiation of
$S$-sums in an analytically consistent way.

We present two strategies to compute the differentiation of $S$-sums. The first
approach
follows the ideas from~\cite{Blumlein:2009ta,Blumlein:2009fz} for harmonic sums: compute
the Mellin
transform, differentiate the arising integrals (which can be analytically continued in $n$), and apply afterwards the inverse Mellin transform to bring these expressions back to $S$-sums. This method is restricted to $\bar{S}$-sums (for a definition see~\eqref{Equ:SBar}) using our current technologies. Finally, we will present a second approach using directly the integral representation of $S$-sums that will work in general.

In both approaches one can start with an even $S$-sum $\S{\ve a}{\ve b;2n}$ or odd $S$-sum $\S{\ve a}{\ve b;2n+1}$ (where the uniqueness property of the analytic
continuation is given). For convenience, one can also work with the usual $S$-sum and obtains an output of the form
$$\S{\ve a}{\ve b;n}=(-1)^n F(n)+G(n),$$
which encodes both cases: Setting $n\to2n$, resp.\ $n\to2n+1$,
yields the differentiation (uniquely determined among the C-type functions)
for the even resp.\ odd case.

\subsection{Differentiation by means of the Mellin transform and its inverse}

\noindent
Using the inverse Mellin transform we express an $\bar{S}$-sum to the form~\eqref{Equ:SSumToMellin} where the $f_i(n)$ stands for an expression of the form $\M{\frac{\H{\ve m}x}{k\pm x}}n$
with $\H{\ve m}x \in \bar{H}$. Note that from now on the possibly occurring $(-1)^n$ in~\eqref{Equ:SSumToMellin} remains untouched and we deal only with the remaining objects which can be analytically continued. Next we differentiate these summands. For ``$+$'' and $k\in\R^+$ this leads to:
\begin{eqnarray*}
&&\frac{d}{d n}\M{\frac{\H{\ve m}x}{k+x}}n=\frac{d}{d n}\int_0^1\frac{x^n\H{\ve m}x}{k+x}dx=\int_0^1\frac{x^n\H{0}x\H{\ve m}x}{k+x}dx\\
&&\hspace{1cm}=\int_0^1\frac{x^n(\H{0,m_1,\ldots,m_l}x+\H{m_1,0,m_2\ldots,m_l}x+\cdots+\H{m_1,\ldots,m_l,0}x)}{k+x}dx\\
&&\hspace{1cm}=\M{\frac{\H{0,m_1,\ldots,m_l}x}{k+x}}n+\cdots+ \M{\frac{\H{m_1,\ldots,m_l,0}x}{k+x}}n.
\end{eqnarray*}
For ``$-$'' and $k\in (1,\infty)$ this leads to:
\begin{eqnarray*}
&&\frac{d}{d n}\M{\frac{\H{\ve m}x}{k-x}}n=\frac{d}{d n}\int_0^1\frac{x^n\H{\ve m}x}{k-x}dx=\int_0^1\frac{x^n\H{0}x\H{\ve m}x}{k-x}dx\\
&&\hspace{1cm}=\int_0^1\frac{x^n\H{0,m_1,\ldots,m_l}x+\H{m_1,0,m_2\ldots,m_l}x+\cdots+\H{m_1,\ldots,m_l,0}x}{k-x}dx\\
&&\hspace{1cm}=\M{\frac{\H{0,m_1,\ldots,m_l}x}{k- x}}n+\cdots+ \M{\frac{\H{m_1,\ldots,m_l,0}x}{k-x}}n.
\end{eqnarray*}
Finally, for ``$-$'' and $k\in(0,1]$ this leads to:
\begin{eqnarray}
\frac{d}{d n}\M{\frac{\H{\ve m}x}{k-x}}n&=&\frac{d}{d n}\int_0^1\frac{\left(\frac{x^n}{k^n}-1\right)\H{\ve m}x}{k-x}dx=
\int_0^1\frac{\frac{x^n}{k^n}\left(\H{0}x-\H{0}k\right)\H{\ve m}x}{k-x}dx\nonumber\\
&=&\int_0^1\frac{\left(\frac{x^n}{k^n}-1\right)\H{0}x\H{\ve m}x}{k-x}dx-\H{0}k\int_0^1\frac{\left(\frac{x^n}{k^n}-1\right)\H{\ve m}x}{k-x}dx\nonumber\\
&&+\int_0^1\frac{\H{0}{\frac{x}{k}}\H{\ve m}x}{k-x}dx.
\label{SSdiffconst1}
\end{eqnarray}
In the following lemmas we will see that $\int_0^1\frac{\Hma{0}{\frac{x}{k}}\H{\ve m}x}{k-x}dx$ is finite and in other words, the last integral in (\ref{SSdiffconst1}) is well defined. Moreover,
we will see how we can evaluate this constant in terms of generalized polylogarithms.
Hence in all three cases we eventually arrive at Mellin transforms of generalized polylogarithms. To this end, we perform the Mellin transform of these expressions and arrive at a final expression in terms of $\bar{S}$-sums together with infinite $\bar{S}$-sums.

\begin{lemma}
For a generalized polylogarithm $\H{\ve m}x$ with $\ve m=(m_1,m_2,\ldots,m_l),$ $m_i~\in~\R \setminus (0,1]$ and $i\in \N$ the integral
$$
\int_0^1\frac{\H{0}x^i\H{\ve m}x}{1-x}dx
$$
exists and can be expressed in terms of generalized polylogarithms.
\end{lemma}
\begin{proof}
Using integration by parts we get
\begin{eqnarray*}
\int_0^1\frac{\H{0}x^i\H{\ve m}x}{1-x}dx&:=&\lim_{b \rightarrow 1^-}\lim_{a \rightarrow 0^+}{\int_a^{b}\frac{\H{0}x^i\H{\ve m}x}{1-x}dx}\\
	&=&\lim_{b \rightarrow 1^-}\lim_{a \rightarrow 0^+}\left(\left.\H{0}x^i\H{1,\ve m}x\right|_a^{b}-
		i{\int_a^{b}\frac{\H{0}x^{i-1}\H{1,\ve m}x}{x}dx}\right)\\
	&=&-i{\int_0^{1}\frac{\H{0}x^{i-1}\H{1,\ve m}x}{x}dx}.
\end{eqnarray*}
After expanding the product $\H{0}x^{i-1}\H{1,\ve m}x$ and applying the integral we end up in generalized polylogarithms at one with leading zero. Hence the integral is finite.
\end{proof}

\begin{lemma}
For $k\in(0,1)$ and $\ve a=(a_1,a_2,\ldots,a_l),$ $m_i \in\R \setminus (0,1]$ we have that
\begin{eqnarray*}
\int_0^1\frac{\H{0}{\frac{x}{k}}\H{\ve a}x}{k-x}dx
\end{eqnarray*}
is finite and can be expressed in terms of generalized polylogarithms.
\end{lemma}
\begin{proof}
If $\H{\ve a}x$ has trailing zeroes, \ie $a_l=0,$ we first extract them and get a linear combination of expressions of the form $$\int_0^1\frac{\H{0}{\frac{x}{k}}\H{0}x^w\H{\ve m}x}{k-x}dx$$
with $w\in \N$. Now let us look at such an expression:
\begin{eqnarray*}
\int_0^1\frac{\H{0}{\frac{x}{k}}\H{0}x^w\H{\ve m}x}{k-x}dx&=&\int_0^{\frac{1}{k}}\frac{\H{0}{x}\H{0}{kx}^w\H{\ve m}{kx}}{k-kx}kdx\\
&=&\int_0^{\frac{1}{k}}\frac{\H{0}{x}(\H{0}{x}+\H{0}{k})^w\H{\frac{m_1}{k},\ldots,\frac{m_l}{k}}{x}}{1-x}dx\\
&=&\int_0^{\frac{1}{k}}\frac{\H{\frac{m_1}{k},\ldots,\frac{m_l}{k}}{x}\sum_{i=0}^w{\binom{w}{i}\H{0}{x}^{i+1}\H{0}{k}^{w-i}}} {1-x}dx\\
&=&\sum_{i=0}^w{\binom{w}{i}\H{0}{k}^{w-i}\underbrace{\int_0^{\frac{1}{k}}\frac{\H{\frac{m_1}{k},\ldots,\frac{m_l}{k}}{x}\H{0}{x}^{i+1}}{1-x}dx}_A}
\end{eqnarray*}
\begin{eqnarray*}
A&=&\int_0^{1}\frac{\H{\frac{m_1}{k},\ldots,\frac{m_l}{k}}{x}\H{0}{x}^{i+1}}{1-x}dx+\int_1^{\frac{1}{k}}\frac{\H{\frac{m_1}{k},\ldots,\frac{m_l}{k}}{x}\H{0}{x}^{i+1}}{1-x}dx\\
&=&\underbrace{\int_0^1\frac{\H{\frac{m_1}{k},\ldots,\frac{m_l}{k}}{x}\H{0}{x}^{i+1}}{x}dx}_B-\underbrace{\int_0^{\frac{1}{k}-1}\frac{\H{\frac{m_1}{k},\ldots,\frac{m_l}{k}}{x+1}\H{0}{x+1}^{i+1}}{x}dx}_C.
\end{eqnarray*}
The integral $B$ is finite and expressible in terms of generalized polylogarithms due to the previous lemma.
After expanding the product $\H{\frac{m_1}{k},\ldots,\frac{m_l}{k}}{x+1}\H{0}{x+1}^{i+1}$, applying Lemma \ref{SSeinplustrafo} and applying the integral we end up in generalized polylogarithms at one and at $\frac{1}k-1$ which are all finite. Hence C is finite.
Summarizing $$\int_0^1\frac{\H{0}{\frac{x}{k}}\H{\ve a}x}{k-x}dx$$ is finite, since we can write it as a sum of finite integrals.
\end{proof}

We demonstrate this strategy by differentiating $\S{2}{2;n}$. After calculating
\begin{equation*}
\S{2}{2;n}=\M{\frac{\H{0}x}{\frac{1}{2}-x}}n=\int_0^1{\frac{(2^n x^n-1)\H{0}x}{\frac{1}{2}-x}dx}
\end{equation*}
we differentiate with respect to $n$ and obtain
\begin{eqnarray*}
&&\frac{d}{d n}\int_0^1{\frac{(2^n x^n-1)\H{0}x}{\frac{1}{2}-x}dx}=\int_0^1\frac{2^nx^n\left(\H{0}x+\H{0}2\right)\H{0}x}{\frac{1}{2}-x}dx\\
&&\hspace{1cm}=2 \int_0^1\frac{\left(2^nx^n-1\right)\H{0,0}x}{\frac{1}{2}-x}dx+\H{0}2\int_0^1\frac{\left(2^nx^n-1\right)\H{0}x}{\frac{1}{2}-x}dx\\
&&\hspace{1.5cm}+\int_0^1\frac{\H{0}{2\,x}\H{0}x}{\frac{1}{2}-x}dx\\
&&\hspace{1cm}=2 \M{\frac{\H{0,0}x}{\frac{1}{2}-x}}n+\H{0}2\M{\frac{\H{0}x}{\frac{1}{2}-x}}n+\int_0^1\frac{\H{0}{2\, x}\H{0}x}{\frac{1}{2}-x}dx\\
&&\hspace{1cm}=-2\S{3}{2;n}+\H{0}2\S{2}{2;n}+\H{0,0,-1}1+2\H{0,0,1}1+\H{0,1,-1}1.
\end{eqnarray*}


\subsection{Differentiation based on the nested integral representation}

\normalsize
\noindent
Using the integral representation from Theorem~\ref{SSintrep} the differentiation can be extended from $\bar{S}$-sums to $S$-sums.
We demonstrate the underlying ideas by the differentiation of $\SS{1,2,1}{\tfrac{1}{2},2,\tfrac{1}{3}}n$.
We start with the iterated integral  representation
\small
\begin{align*}
\SS{1,2,1}{\tfrac{1}{2},2,\tfrac{1}{3}}n=&\int_0^{\tfrac{1}{3}}{\tfrac{1}{x-1}\int_{1}^{x}{\tfrac{1}{y}\int_0^{y}{\frac{1}{z-\frac{1}{2}}\int_{\tfrac{1}{2}}^{z}{\frac{w^{n}-1}{w-1}}dwdz}dy}dx}\\
=&\int_0^{\frac{1}{3}}{\frac{1}{x-1}\int^{1}_{x}{-\frac{1}{y}\int_0^{y}{\frac{1}{\frac{1}{2}-z}\int_{\tfrac{1}{2}}^{z}{\frac{w^{n}-1}{w-1}}dwdz}dy}dx}
\end{align*}
\normalsize
given by Theorem \ref{SSintrep}.
Differentiation with respect to $n$ yields
\small
\begin{eqnarray*}
\int_0^{\frac{1}{3}}{\frac{1}{x-1}\int^{1}_{x}{-\frac{1}{y}\int_0^{y}{\frac{1}{\frac{1}{2}-z}\int_{\tfrac{1}{2}}^{z}{\frac{w^{n}\H0w}{w-1}}dwdz}dy}dx}=:A.
\end{eqnarray*}
\normalsize
We want to rewrite $A$ in terms of $S$-sums at $n$ and finite $S$-sums at $\infty$ (or
finite generalized polylogarithms). Therefore we first rewrite $A$ in the following form:
\small
\begin{eqnarray*}
A&=&\underbrace{\int_0^{\frac{1}{3}}{\frac{1}{x-1}\int^{1}_{x}{-\frac{1}{y}\int_0^{y}{\frac{1}{\frac{1}{2}-z}\int_{\tfrac{1}{2}}^{z}{\frac{w^{n}-1}{w-1}\H0w}dwdz}dy}dx}}_{B:=}+\\
&&\underbrace{\int_0^{\frac{1}{3}}{\frac{1}{x-1}\int^{1}_{x}{-\frac{1}{y}\int_0^{y}{\frac{1}{\frac{1}{2}-z}\int_{\tfrac{1}{2}}^{z}{\frac{\H0w}{w-1}}dwdz}dy}dx}}_{C:=}.
\end{eqnarray*}
\normalsize
Let us first consider $B$: we split the integral at the zeroes of the denominators and get
\small
\begin{eqnarray*}
B&=&\int_0^{\frac{1}{3}}{\frac{1}{x-1}\int^{\frac{1}{2}}_{x}{-\frac{1}{y}\int_0^{y}{\frac{1}{z-\frac{1}{2}}\int^{\frac{1}{2}}_{z}{-\frac{w^{n}-1}{w-1}\H0w}dwdz}dy}dx}+\\
&&\int_0^{\frac{1}{3}}{\frac{1}{x-1}\int^{1}_{\tfrac{1}{2}}{-\frac{1}{y}\int_0^{\frac{1}{2}}{\frac{1}{z-\frac{1}{2}}\int^{\frac{1}{2}}_{z}{-\frac{w^{n}-1}{w-1}\H0w}dwdz}dy}dx}+\\
&&\int_0^{\frac{1}{3}}{\frac{1}{x-1}\int^{1}_{\tfrac{1}{2}}{-\frac{1}{y}\int_{\tfrac{1}{2}}^{y}{\frac{1}{z-\frac{1}{2}}\int_{\tfrac{1}{2}}^{z}{\frac{w^{n}-1}{w-1}\H0w}dwdz}dy}dx}\\
&=&B_1+B_2+B_3.
\end{eqnarray*}
\normalsize
Starting from the inner integral and integrating integral by integral leads to
\small
\begin{eqnarray*}
B_1&=&\int_0^{\frac{1}{3}}\frac{1}{x-1}\int^{\frac{1}{2}}_{x}-\frac{1}{y}\int_0^{y}\frac{1}{z-\frac{1}{2}}\biggl(\H{0}{\tfrac{1}{2}}\SS{1}{\tfrac{1}{2}}n-\H{0}z\SS{1}{z}n\biggr. \\
   & & \biggl.+\SS{2}{\tfrac{1}{2}}n-\SS{2}{z}n\biggr)dzdydx\\
   &=&\cdots =\\
   &=&-\textnormal{H}_3(1) \textnormal{H}_{0,1,0}(1) \textnormal{S}_1\left(\tfrac{1}{2};n\right)+\textnormal{H}_0(2) \textnormal{H}_{3,0,\tfrac{3}{2}}(1) \textnormal{S}_1\left(\tfrac{1}{2};n\right)-\textnormal{H}_0(3) \textnormal{H}_{3,0,\tfrac{3}{2}}(1)
   \textnormal{S}_1\left(\tfrac{1}{2};n\right)\\
    &&-\textnormal{H}_{0,3,0,\tfrac{3}{2}}(1) \textnormal{S}_1\left(\tfrac{1}{2};n\right)-2 \textnormal{H}_{3,0,0,\tfrac{3}{2}}(1)
   \textnormal{S}_1\left(\tfrac{1}{2};n\right)-\textnormal{H}_0(2) \textnormal{H}_3(1) \textnormal{S}_{1,2}\hspace{-0.2em}\left(\tfrac{1}{2},1;n\right)\\
    &&+\textnormal{H}_0(3) \textnormal{H}_3(1) \textnormal{S}_{1,2}\hspace{-0.2em}\left(\tfrac{1}{2},2;n\right)+\textnormal{H}_{0,3}(1)
   \textnormal{S}_{1,2}\hspace{-0.2em}\left(\tfrac{1}{2},2;n\right)-2 \textnormal{H}_3(1) \textnormal{S}_{1,3}\hspace{-0.2em}\left(\tfrac{1}{2},1;n\right)\\
    &&+2 \textnormal{H}_3(1) \textnormal{S}_{1,3}\hspace{-0.2em}\left(\tfrac{1}{2},2;n\right)-\textnormal{H}_3(1)
   \textnormal{S}_{2,2}\hspace{-0.2em}\left(\tfrac{1}{2},1;n\right)+\textnormal{H}_3(1) \textnormal{S}_{2,2}\hspace{-0.2em}\left(\tfrac{1}{2},2;n\right)\\
    &&-\textnormal{H}_0(3)
   \textnormal{S}_{1,2,1}\hspace{-0.2em}\left(\tfrac{1}{2},2,\tfrac{1}{3};n\right)-\textnormal{S}_{1,2,2}\hspace{-0.2em}\left(\tfrac{1}{2},2,\tfrac{1}{3};n\right)-2
   \textnormal{S}_{1,3,1}\hspace{-0.2em}\left(\tfrac{1}{2},2,\tfrac{1}{3};n\right)\\
    &&-\textnormal{S}_{2,2,1}\hspace{-0.2em}\left(\tfrac{1}{2},2,\tfrac{1}{3};n\right).
\end{eqnarray*}
\normalsize
Applying the same strategy to $B_2$ and $B_3$ leads to
\small
\begin{eqnarray*}
B&=&-\textnormal{H}_3(1) \textnormal{H}_{0,1,0}(2) \textnormal{S}_1\left(\tfrac{1}{2};n\right)+\textnormal{H}_0(2) \textnormal{H}_{3,0,\tfrac{3}{2}}(1) \textnormal{S}_1\left(\tfrac{1}{2};n\right)
  -\textnormal{H}_0(3) \textnormal{H}_{3,0,\tfrac{3}{2}}(1)\textnormal{S}_1\left(\tfrac{1}{2};n\right)\\
&&-\textnormal{H}_{0,3,0,\tfrac{3}{2}}(1) \textnormal{S}_1\left(\tfrac{1}{2};n\right)-2 \textnormal{H}_{3,0,0,\tfrac{3}{2}}(1)
   \textnormal{S}_1\left(\tfrac{1}{2};n\right)+\textnormal{H}_0(3) \textnormal{H}_3(1) \textnormal{S}_{1,2}\hspace{-0.2em}\left(\tfrac{1}{2},2;n\right)\\
&&+\textnormal{H}_{0,3}(1) \textnormal{S}_{1,2}\hspace{-0.2em}\left(\tfrac{1}{2},2;n\right)-\textnormal{H}_0(3)
   \textnormal{S}_{1,2,1}\hspace{-0.2em}\left(\tfrac{1}{2},2,\tfrac{1}{3};n\right)-\textnormal{S}_{1,2,2}\hspace{-0.2em}\left(\tfrac{1}{2},2,\tfrac{1}{3};n\right)\\
&&-2\textnormal{S}_{1,3,1}\hspace{-0.2em}\left(\tfrac{1}{2},2,\tfrac{1}{3};n\right)-\textnormal{S}_{2,2,1}\hspace{-0.2em}\left(\tfrac{1}{2},2,\tfrac{1}{3};n\right).
\end{eqnarray*}
\normalsize
Let us now consider $C$: again we split the integral at the zeroes of the denominators and get:
\small
\begin{eqnarray*}
C&=&\int_0^{\frac{1}{3}}{\frac{1}{x-1}\int^{\frac{1}{2}}_{x}{-\frac{1}{y}\int_0^{y}{\frac{1}{z-\frac{1}{2}}\int^{\frac{1}{2}}_{z}{-\frac{\H0w}{w-1}}dwdz}dy}dx}+\\
&&\int_0^{\frac{1}{3}}{\frac{1}{x-1}\int^{1}_{\tfrac{1}{2}}{-\frac{1}{y}\int_0^{\frac{1}{2}}{\frac{1}{z-\frac{1}{2}}\int^{\frac{1}{2}}_{z}{-\frac{\H0w}{w-1}}dwdz}dy}dx}+\\
&&\int_0^{\frac{1}{3}}{\frac{1}{x-1}\int^{1}_{\tfrac{1}{2}}{-\frac{1}{y}\int_{\tfrac{1}{2}}^{y}{\frac{1}{z-\frac{1}{2}}\int_{\tfrac{1}{2}}^{z}{\frac{\H0w}{w-1}}dwdz}dy}dx}\\
&=&C_1+C_2+C_3.
\end{eqnarray*}
\normalsize
Starting from the inner integral and integrating integral by integral leads to
\small
\begin{eqnarray*}
C_1&=&\int_0^{\frac{1}{3}}{\frac{1}{x-1}\int^{\frac{1}{2}}_{x}{-\frac{1}{y}\int_0^{y}{\frac{\H{1,0}{\frac{1}{2}}-\H{1,0}{z}}{z-\frac{1}{2}}dz}dy}dx}\\ &=&\int_0^{\frac{1}{3}}{\frac{1}{x-1}\int^{\frac{1}{2}}_{x}{-\frac{-\H{\frac{1}{2}}{y}\H{1,0}{\frac{1}{2}}+\H{\frac{1}{2},1,0}{y}} {y}}dy}dx\\ &=&\int_0^{\frac{1}{3}}{\frac{\H{1,0}{\tfrac{1}{2}}\left(\H{0,\tfrac{1}{2}}{\tfrac{1}{2}}-\H{0,\tfrac{1}{2}}{x}\right)
      -\H{0,\tfrac{1}{2},1,0}{\tfrac{1}{2}}+\H{0,\tfrac{1}{2},1,0}{x}}{x-1}dx}\\
   &=&\textnormal{H}_0(2) \textnormal{H}_2(1) \textnormal{H}_3(1) \textnormal{H}_{0,1}(1)+\textnormal{H}_3(1) \textnormal{H}_{0,2}(1) \textnormal{H}_{0,1}(1)-\textnormal{H}_0(2) \textnormal{H}_3(1) \textnormal{H}_{0,1,2}(1)\\
    &&-\textnormal{H}_0(2) \textnormal{H}_2(1) \textnormal{H}_{3,0,\tfrac{3}{2}}(1)-\textnormal{H}_{0,2}(1)
   \textnormal{H}_{3,0,\tfrac{3}{2}}(1)-2 \textnormal{H}_3(1) \textnormal{H}_{0,0,1,2}(1)-\textnormal{H}_3(1) \textnormal{H}_{0,1,0,2}(1)\\
&&+\textnormal{H}_0(3) \textnormal{H}_{3,0,\tfrac{3}{2},3}(1)+\textnormal{H}_{0,3,0,\tfrac{3}{2},3}(1)+2
\textnormal{H}_{3,0,0,\tfrac{3}{2},3}(1)+\textnormal{H}_{3,0,\tfrac{3}{2},0,3}(1).
\end{eqnarray*}
\normalsize
Applying the same strategy to $C_2$ and $C_3$ leads to
\small
\begin{eqnarray*}
C&=&\textnormal{H}_3(1) \textnormal{H}_0(2){}^2 \textnormal{H}_{0,-1}(1)+\textnormal{H}_2(1) \textnormal{H}_3(1) \textnormal{H}_0(2) \textnormal{H}_{0,1}(1)+\textnormal{H}_3(1) \textnormal{H}_0(2) \textnormal{H}_{-1,0,1}(1)\\
&&+\textnormal{H}_3(1) \textnormal{H}_0(2) \textnormal{H}_{0,-1,1}(1)-\textnormal{H}_3(1) \textnormal{H}_0(2)\textnormal{H}_{0,1,2}(1)-\textnormal{H}_2(1) \textnormal{H}_0(2) \textnormal{H}_{3,0,\tfrac{3}{2}}(1)\\
&&+\textnormal{H}_3(1) \textnormal{H}_{0,1}(1) \textnormal{H}_{0,2}(1)-\textnormal{H}_{0,2}(1) \textnormal{H}_{3,0,\tfrac{3}{2}}(1)-\textnormal{H}_3(1) \textnormal{H}_{-1,0,1,-1}(1)-2\textnormal{H}_3(1) \textnormal{H}_{0,0,1,2}(1)\\
&&-\textnormal{H}_3(1) \textnormal{H}_{0,1,0,2}(1)+\textnormal{H}_0(3) \textnormal{H}_{3,0,\tfrac{3}{2},3}(1)+\textnormal{H}_{0,3,0,\tfrac{3}{2},3}(1)+2\textnormal{H}_{3,0,0,\tfrac{3}{2},3}(1)+\textnormal{H}_{3,0,\tfrac{3}{2},0,3}(1).
\end{eqnarray*}
\normalsize
Adding $B$ and $C$ leads to the final result:
\small
\begin{equation}\label{Equ:DiffWithLargeConst}
\begin{split}
&\frac{\partial \SS{1,2,1}{\tfrac{1}{2},2,\tfrac{1}{3}}{n}} {\partial n}=\\
&\hspace{0.5cm}\textnormal{H}_0(2) \textnormal{H}_{3,0,\tfrac{3}{2}}\hspace{-0.2em}(1) \textnormal{S}_1\left(\tfrac{1}{2};n\right)-\textnormal{H}_3(1) \textnormal{H}_{0,1,0}\hspace{-0.2em}(2) \textnormal{S}_1\left(\tfrac{1}{2};n\right)-\textnormal{H}_0(3) \textnormal{H}_{3,0,\tfrac{3}{2}}\hspace{-0.2em}(1)
   \textnormal{S}_1\left(\tfrac{1}{2};n\right)\\
&\hspace{0.5cm}-\textnormal{H}_{0,3,0,\tfrac{3}{2}}\hspace{-0.2em}(1) \textnormal{S}_1\left(\tfrac{1}{2};n\right)-2 \textnormal{H}_{3,0,0,\tfrac{3}{2}}\hspace{-0.2em}(1)
   \textnormal{S}_1\left(\tfrac{1}{2};n\right)+\textnormal{H}_0(3) \textnormal{H}_3(1) \textnormal{S}_{1,2}\hspace{-0.2em}\left(\tfrac{1}{2},2;n\right)\\
&\hspace{0.5cm}+\textnormal{H}_{0,3}\hspace{-0.2em}(1) \textnormal{S}_{1,2}\hspace{-0.2em}\left(\tfrac{1}{2},2;n\right)-\textnormal{H}_0(3)   \textnormal{S}_{1,2,1}\hspace{-0.2em}\left(\tfrac{1}{2},2,\tfrac{1}{3};n\right)-\textnormal{S}_{1,2,2}\hspace{-0.2em}\left(\tfrac{1}{2},2,\tfrac{1}{3};n\right)\\
&\hspace{0.5cm}-2 \textnormal{S}_{1,3,1}\hspace{-0.2em}\left(\tfrac{1}{2},2,\tfrac{1}{3};n\right)-\textnormal{S}_{2,2,1}\hspace{-0.2em}\left(\tfrac{1}{2},2,\tfrac{1}{3};n\right)+\textnormal{H}_3(1) \textnormal{H}_0(2){}^2 \textnormal{H}_{0,-1}\hspace{-0.2em}(1)\\
&\hspace{0.5cm}+\textnormal{H}_2(1)\textnormal{H}_3(1) \textnormal{H}_0(2) \textnormal{H}_{0,1}\hspace{-0.2em}(1)+\textnormal{H}_3(1) \textnormal{H}_0(2) \textnormal{H}_{-1,0,1}\hspace{-0.2em}(1)+\textnormal{H}_3(1) \textnormal{H}_0(2) \textnormal{H}_{0,-1,1}\hspace{-0.2em}(1)\\
&\hspace{0.5cm}-\textnormal{H}_3(1) \textnormal{H}_0(2) \textnormal{H}_{0,1,2}\hspace{-0.2em}(1)-\textnormal{H}_2(1) \textnormal{H}_0(2)\textnormal{H}_{3,0,\tfrac{3}{2}}\hspace{-0.2em}(1)+\textnormal{H}_3(1) \textnormal{H}_{0,1}\hspace{-0.2em}(1) \textnormal{H}_{0,2}\hspace{-0.2em}(1)-\textnormal{H}_{0,2}\hspace{-0.2em}(1) \textnormal{H}_{3,0,\tfrac{3}{2}}\hspace{-0.2em}(1)\\
&\hspace{0.5cm}-\textnormal{H}_3(1) \textnormal{H}_{-1,0,1,-1}\hspace{-0.2em}(1)-2 \textnormal{H}_3(1) \textnormal{H}_{0,0,1,2}\hspace{-0.2em}(1)-\textnormal{H}_3(1)
   \textnormal{H}_{0,1,0,2}\hspace{-0.2em}(1)\\
&\hspace{0.5cm}+\textnormal{H}_0(3) \textnormal{H}_{3,0,\tfrac{3}{2},3}\hspace{-0.2em}(1)+\textnormal{H}_{0,3,0,\tfrac{3}{2},3}\hspace{-0.2em}(1)+2 \textnormal{H}_{3,0,0,\tfrac{3}{2},3}\hspace{-0.2em}(1)+\textnormal{H}_{3,0,\tfrac{3}{2},0,3}\hspace{-0.2em}(1).
\end{split}
\end{equation}
\normalsize

\noindent As worked out in~\cite{Ablinger:12} the proposed strategy works in general for $S$-sums.

\section{Relations between $S$-Sums}\label{Sec:Relation}

\vspace*{1mm}
\noindent
$S$-sums obey stuffle relations due to their quasi-shuffle algebra, differentiation
relations, and generalized argument relations.  Subsequently, the available ideas for
harmonic sums~\cite{Blumlein:2003gb,Blumlein:2009ta,Blumlein:2009fz} and cyclotomic
harmonic
sums~\cite{Ablinger:2011te} are extended and explored for $S$-sums.

\subsection{Quasi-Shuffle or Stuffle Relations}\label{SSalgrel}

\noindent
Based on~\cite{Hoffman} it has been worked out in \cite{Moch:2001zr} that $Z$-sums
\begin{equation}\label{Equ:ZSum}
\textnormal{Z}_{a_1,\ldots ,a_k}(x_1,\ldots ,x_k;n)= \sum_{n\geq i_1 > i_2 > \cdots > i_k \geq 1} \frac{x_1^{i_1}}{i_1^{a_1}}\cdots
	\frac{x_k^{i_k}}{i_k^{a_k}}
\end{equation}
with their quasi-shuffle product form a quasi-shuffle algebra. Since $S$-sums and $Z$-sums can be transformed into each other, a quasi-shuffle algebra can be carried over to $S$-sums.

We remark that the quasi-shuffle algebra can be derived directly (without using the
notion of $Z$-sums) as follows. Consider an alphabet $A$, where pairs $(m,x)$ with
$m\in\N$ and $x\in\R^*$ form the letters, \ie we identify a $S$-sum $$
\S{m_1,m_2,\ldots,m_k}{x_1,x_2,\ldots,x_k;n}
$$
with the word $(m_1,x_1)(m_2,x_2)\ldots(m_k,x_k)$.
We define the degree of a letter $(m,x) \in A$ as $\abs{(m,x)}=m$ and we order the letters for  $m_1,m_2\in\N$ and $x_1,x_2\in\R^*$ by
\begin{eqnarray*}
\begin{array}{llll}
	(m_1,x_1)	&\prec (m_2,x_2) 	&\textnormal{if }&m_1<m_2\\
	(m_1,x_1)	&\prec (m_1,x_2) 	&\textnormal{if }&\abs{x_1}<\abs{x_2}\\
	(m_1,-x_1)	&\prec (m_1,x_1) 	&\textnormal{if }&x_1>0.
\end{array}
\end{eqnarray*}
We extend this order lexicographically to words. Using this order, it can be shown
analogously as in \cite{Ablinger:2010kw} (compare \cite{Hoffman}) that the $S$-sums
form
a quasi shuffle algebra which is the free polynomial algebra on the
\textit{Lyndon} words with alphabet $A$, see also \cite{RADFORD}.

As a consequence the number of algebraic independent sums in this algebra can be counted by counting the number of \textit{Lyndon} words.
If we consider for example an alphabet with $n$ letters and we look for the number of basis sums with depth $d,$
we can use the first Witt formula~\cite{Witt1937,Witt1956,Reutenauer1993}:
The number of \textit{Lyndon} words of length $d$ over an alphabet of length $n$ is given by
$$N_A(d) = \frac{1}{d}\sum_{q|d}{\mu\left(\frac{d}{q}\right)n^q}$$
where \begin{equation}
	\mu(n)=\left\{
		  	\begin{array}{ll}
						1\  & \textnormal{if } n = 1  \\
						0\  & \textnormal{if } n \textnormal{ is divided by the square of a prime}  \\
						(-1)^s\  & \textnormal{if } n \textnormal{ is the product of } s \textnormal{ different primes}
					\end{array} \right. \nonumber\\
\label{abmue}
\end{equation}
is the M\"obius function \cite{MOEBIUS:1832,Hardy:1978}. In the following we call
the
relations from the quasi-shuffle
algebra also stuffle relations; cf.~\cite{Blumlein:2009cf}.

In the subsequent example we look at $S$-sums of depth 2 on alphabets of 4 letters.
Hence we obtain
$ \frac{1}{2}\sum_{q|2}{\mu\left(\frac{2}{q}\right)4^q}=6$
basis sums.
We can use an analogous method of the method presented in
\cite{Blumlein:2003gb,Ablinger:2010kw}
for harmonic sums to find the basis $S$-sums together with the relations for the
dependent $S$-sums.

For example, consider the letters $(1,\frac{1}{2}),(1,-\frac{1}{2}),(3,\frac{1}{2}),(3,-\frac{1}{2}).$ At depth $d=2$ we obtain $16$ sums with these letters.
Using the relations
\begin{eqnarray*}
\textnormal{S}_{1,3}\left(\tfrac{1}{2},\tfrac{1}{2};n\right)&=& -\textnormal{S}_{3,1}\left(\tfrac{1}{2},\tfrac{1}{2};n\right)+\textnormal{S}_1\left(\tfrac{1}{2};n\right)\textnormal{S}_3\left(\tfrac{1}{2};n\right)+\textnormal{S}_4\left(\tfrac{1}{4};n\right)\\
\textnormal{S}_{1,3}\left(-\tfrac{1}{2},-\tfrac{1}{2};n\right)&=&-\textnormal{S}_{3,1}\left(-\tfrac{1}{2},-\tfrac{1}{2};n\right)+\textnormal{S}_1\left(-\tfrac{1}{2};n\right)\textnormal{S}_3\left(-\tfrac{1}{2};n\right)+\textnormal{S}_4\left(\tfrac{1}{4};n\right)\\
\textnormal{S}_{1,3}\left(\tfrac{1}{2},-\tfrac{1}{2};n\right)&=&-\textnormal{S}_{3,1}\left(-\tfrac{1}{2},\tfrac{1}{2};n\right)+\textnormal{S}_1\left(\tfrac{1}{2};n\right)\textnormal{S}_3\left(-\tfrac{1}{2};n\right)+\textnormal{S}_4\left(-\tfrac{1}{4};n\right)\\
\textnormal{S}_{1,3}\left(-\tfrac{1}{2},\tfrac{1}{2};n\right)&=&-\textnormal{S}_{3,1}\left(\tfrac{1}{2},-\tfrac{1}{2};n\right)+\textnormal{S}_1\left(-\tfrac{1}{2};n\right)\textnormal{S}_3\left(\tfrac{1}{2};n\right)+\textnormal{S}_4\left(-\tfrac{1}{4};n\right)\\
\textnormal{S}_{1,1}\left(-\tfrac{1}{2},-\tfrac{1}{2};n\right)&=& \tfrac{1}{2} \textnormal{S}_1\left(-\tfrac{1}{2};n\right)^2+\tfrac{1}{2}\textnormal{S}_2\left(\tfrac{1}{4};n\right)\\
\textnormal{S}_{1,1}\left(\tfrac{1}{2},\tfrac{1}{2};n\right)&=& \tfrac{1}{2} \textnormal{S}_1\left(\tfrac{1}{2};n\right)^2+\tfrac{1}{2}\textnormal{S}_2\left(\tfrac{1}{4};n\right)\\
\textnormal{S}_{1,1}\left(\tfrac{1}{2},-\tfrac{1}{2};n\right)&=&-\textnormal{S}_{1,1}\left(-\tfrac{1}{2},\tfrac{1}{2};n\right)+\textnormal{S}_1\left(-\tfrac{1}{2};n\right)\textnormal{S}_1\left(\tfrac{1}{2};n\right)+\textnormal{S}_2\left(-\tfrac{1}{4};n\right)\\
\textnormal{S}_{3,3}\left(-\tfrac{1}{2},-\tfrac{1}{2};n\right)&=& \frac{1}{2}\textnormal{S}_3\left(-\tfrac{1}{2};n\right)^2+\tfrac{1}{2} \textnormal{S}_6\left(\tfrac{1}{4};n\right)\\
\textnormal{S}_{3,3}\left(\tfrac{1}{2},\tfrac{1}{2};n\right)&=& \tfrac{1}{2}\textnormal{S}_3\left(\tfrac{1}{2};n\right)^2+\tfrac{1}{2} \textnormal{S}_6\left(\tfrac{1}{4};n\right)\\
\textnormal{S}_{3,3}\left(\tfrac{1}{2},-\tfrac{1}{2};n\right)&=&-\textnormal{S}_{3,3}\left(-\tfrac{1}{2},\tfrac{1}{2};n\right)+\textnormal{S}_3\left(-\tfrac{1}{2};n\right)\textnormal{S}_3\left(\tfrac{1}{2};n\right)+\textnormal{S}_6\left(-\tfrac{1}{4};n\right)
\end{eqnarray*}
we find the $6$ basis sums
\begin{eqnarray*}
&&\textnormal{S}_{3,1}\left(\tfrac{1}{2},\tfrac{1}{2};n\right),\textnormal{S}_{3,1}\left(-\tfrac{1}{2},\tfrac{1}{2};n\right),\textnormal{S}_{3,1}\left(\tfrac{1}{2},-\tfrac{1}{2};n\right),\textnormal{S}_{3,1}\left(-\tfrac{1}{2},-\tfrac{1}{2};n\right),\\
&&\textnormal{S}_{1,1}\left(-\tfrac{1}{2},\tfrac{1}{2};n\right),\textnormal{S}_{3,3}\left(-\tfrac{1}{2},\tfrac{1}{2};n\right),
\end{eqnarray*}
in which all the other 10 sums of depth 2 can be expressed.

\subsection{Differential Relations}
\label{SSdiffrel}

\noindent
In Section \ref{SSdifferentiation} we described the differentiation of $S$-sums with
respect to the upper summation limit.  Similarly to the application of differentiation to
harmonic sums, cf.~\cite{Blumlein:2009ta,Blumlein:2009fz}, and cyclotomic sums,
cf.~\cite{Ablinger:2011te}, the differentiation leads to new relations in the class of
$S$-sums. For instance we find
$$
\frac{d}{d n}\S{2}{2;n}=-\S{3}{2;n}+\H{0}2\S{2}{2;n}+\H{0,0,-1}1+2\H{0,0,1}1+\H{0,1,-1}1.
$$
Subsequently, we collect the derivatives with respect to $n:$
\begin{eqnarray}
\S{a_1,\ldots,a_k}{b_1,\ldots,b_k;n}^{(D)}
= \left\{\frac{\partial^N}{\partial n^N}\S{a_1,\ldots,a_k}{b_1,\ldots,b_k;n};
N \in \N\right\}.
\end{eqnarray}
Continuing the example from above we get the additional relations:
\begin{align*}
\textnormal{S}_{3,1}\hspace{-0.2em}\left(\tfrac{1}{2},\tfrac{1}{2};n\right)=& \frac{1}{12} \frac{\partial}{\partial n}\textnormal{S}_3\hspace{-0.2em}\left(\tfrac{1}{4};n\right)-\frac{1}{2}
   \frac{\partial}{\partial n}\textnormal{S}_{2,1}\hspace{-0.2em}\left(\tfrac{1}{2},\tfrac{1}{2};n\right)-\frac{1}{2} \textnormal{H}_{1,0}\hspace{-0.2em}\left(\tfrac{1}{2}\right) \textnormal{S}_2\hspace{-0.2em}\left(\tfrac{1}{2};n\right)\\&+\textnormal{H}_0\hspace{-0.2em}\left(\tfrac{1}{2}\right) \textnormal{S}_{2,1}\hspace{-0.2em}\left(\tfrac{1}{2},\tfrac{1}{2};n\right)-\frac{1}{2}
   \textnormal{H}_{\tfrac{1}{2}}\hspace{-0.2em}\left(\tfrac{1}{4}\right) \textnormal{H}_{0,1,0}\hspace{-0.2em}\left(\tfrac{1}{2}\right)+\frac{1}{12} \textnormal{H}_{0,0,1,0}\hspace{-0.2em}\left(\tfrac{1}{4}\right)\\&+\frac{1}{2} \textnormal{H}_{\tfrac{1}{2},0,1,0}\hspace{-0.2em}\left(\tfrac{1}{4}\right)-\frac{1}{12} \textnormal{H}_0\hspace{-0.2em}\left(\tfrac{1}{4}\right) \textnormal{S}_3\hspace{-0.2em}\left(\tfrac{1}{4};n\right)-\frac{1}{4}
   \textnormal{S}_2\hspace{-0.2em}\left(\tfrac{1}{2};n\right)^2\\
\textnormal{S}_{3,1}\hspace{-0.2em}\left(-\tfrac{1}{2},-\tfrac{1}{2};n\right)=& \frac{1}{12}
   \frac{\partial}{\partial n}\textnormal{S}_3\hspace{-0.2em}\left(\tfrac{1}{4};n\right)-\frac{1}{2}
   \frac{\partial}{\partial n}\textnormal{S}_{2,1}\hspace{-0.2em}\left(-\tfrac{1}{2},-\tfrac{1}{2};n\right)+\frac{1}{2} \textnormal{H}_{-\tfrac{1}{2},0}\hspace{-0.2em}\left(\tfrac{1}{4}\right)
   \textnormal{S}_2\hspace{-0.2em}\left(-\tfrac{1}{2};n\right)\\&+\frac{1}{2} \textnormal{H}_0\hspace{-0.2em}\left(\tfrac{1}{4}\right) \textnormal{S}_{2,1}\hspace{-0.2em}\left(-\tfrac{1}{2},-\tfrac{1}{2};n\right)-\frac{1}{2}
   \textnormal{H}_{-\tfrac{1}{2}}\hspace{-0.2em}\left(\tfrac{1}{4}\right) \textnormal{H}_{0,-1,0}\hspace{-0.2em}\left(\tfrac{1}{2}\right)\\&-\frac{1}{2} \textnormal{H}_{-\tfrac{1}{2},0,1,0}\hspace{-0.2em}\left(\tfrac{1}{4}\right)+\frac{1}{12} \textnormal{H}_{0,0,1,0}\hspace{-0.2em}\left(\tfrac{1}{4}\right)-\frac{1}{2}
   \textnormal{H}_{-\tfrac{1}{2}}\hspace{-0.2em}\left(\tfrac{1}{4}\right) \textnormal{H}_0\hspace{-0.2em}\left(\tfrac{1}{2}\right) \textnormal{S}_2\hspace{-0.2em}\left(-\tfrac{1}{2};n\right)\\&-\frac{1}{12} \textnormal{H}_0\hspace{-0.2em}\left(\tfrac{1}{4}\right)
   \textnormal{S}_3\hspace{-0.2em}\left(\tfrac{1}{4};n\right)-\frac{1}{4} \textnormal{S}_2\hspace{-0.2em}\left(-\tfrac{1}{2};n\right)^2\\
\textnormal{S}_{3,1}\hspace{-0.2em}\left(-\tfrac{1}{2},\tfrac{1}{2};n\right)=& \frac{1}{6}
   \frac{\partial}{\partial n}\textnormal{S}_3\hspace{-0.2em}\left(-\tfrac{1}{4};n\right)-\frac{1}{2}
   \frac{\partial}{\partial n}\textnormal{S}_{2,1}\hspace{-0.2em}\left(-\tfrac{1}{2},\tfrac{1}{2};n\right)-\frac{1}{2}
   \frac{\partial}{\partial n}\textnormal{S}_{2,1}\hspace{-0.2em}\left(\tfrac{1}{2},-\tfrac{1}{2};n\right)\\&-\frac{1}{2} \textnormal{H}_{1,0}\hspace{-0.2em}\left(\tfrac{1}{2}\right)
   \textnormal{S}_2\hspace{-0.2em}\left(-\tfrac{1}{2};n\right)+\frac{1}{2} \textnormal{H}_{-\tfrac{1}{2},0}\hspace{-0.2em}\left(\tfrac{1}{4}\right) \textnormal{S}_2\hspace{-0.2em}\left(\tfrac{1}{2};n\right)\\&+\textnormal{H}_0\hspace{-0.2em}\left(\tfrac{1}{2}\right) \textnormal{S}_{2,1}\hspace{-0.2em}\left(-\tfrac{1}{2},\tfrac{1}{2};n\right)+\frac{1}{2}
   \textnormal{H}_0\hspace{-0.2em}\left(\tfrac{1}{4}\right) \textnormal{S}_{2,1}\hspace{-0.2em}\left(\tfrac{1}{2},-\tfrac{1}{2};n\right)\\&-\textnormal{S}_{3,1}\hspace{-0.2em}\left(\tfrac{1}{2},-\tfrac{1}{2};n\right)+\frac{1}{2}
   \textnormal{H}_{\tfrac{1}{2}}\hspace{-0.2em}\left(\tfrac{1}{4}\right) \textnormal{H}_{0,-1,0}\hspace{-0.2em}\left(\tfrac{1}{2}\right)+\frac{1}{2} \textnormal{H}_{\tfrac{3}{4}}\hspace{-0.2em}\left(\tfrac{1}{4}\right)
  \textnormal{H}_{0,1,0}\hspace{-0.2em}\left(\tfrac{1}{2}\right)\\&+\frac{1}{2} \textnormal{H}_{-\tfrac{1}{2},0,-1,0}\hspace{-0.2em}\left(\tfrac{1}{4}\right)-\frac{1}{6}
   \textnormal{H}_{0,0,-1,0}\hspace{-0.2em}\left(\tfrac{1}{4}\right)-\frac{1}{2} \textnormal{H}_0\hspace{-0.2em}\left(\tfrac{1}{2}\right)
   \textnormal{H}_{\tfrac{3}{4}}\hspace{-0.2em}\left(\tfrac{1}{4}\right) \textnormal{S}_2\hspace{-0.2em}\left(\tfrac{1}{2};n\right)\\&-\frac{1}{2} \textnormal{H}_{\tfrac{1}{2},0,-1,0}\hspace{-0.2em}\left(\tfrac{1}{4}\right)-\frac{1}{6} \textnormal{H}_0\hspace{-0.2em}\left(\tfrac{1}{4}\right)
   \textnormal{S}_3\hspace{-0.2em}\left(-\tfrac{1}{4};n\right)-\frac{1}{2} \textnormal{S}_2\hspace{-0.2em}\left(\tfrac{-1}{2};n\right)
   \textnormal{S}_2\hspace{-0.2em}\left(\tfrac{1}{2};n\right).
\end{align*}
Using all the relations we can reduce the number of basis sums at depth $d=2$ to $3$ by introducing differentiation. In our case the basis sums are:
\begin{eqnarray}\label{Equ:BasisSums}
\textnormal{S}_{3,1}\hspace{-0.2em}\left(\tfrac{1}{2},-\tfrac{1}{2};n\right),
\textnormal{S}_{1,1}\hspace{-0.2em}\left(-\tfrac{1}{2},\tfrac{1}{2};n\right),
\textnormal{S}_{3,3}\hspace{-0.2em}\left(-\tfrac{1}{2},\tfrac{1}{2};n\right).
\end{eqnarray}
Note that we introduced the letters $(2,\pm \frac{1}{2})$ and $(3,\pm \frac{1}{4}),$ however these letters appear just in sums of depth 2 with weight $3$ and are only used to express sums of weight $4.$

\subsection{Duplication Relations}
\label{SSduplrel}

\noindent
As for harmonic sums and cyclotomic sums,
cf.~\cite{Vermaseren:1998uu,Blumlein:2009ta,
Blumlein:2009fz,Blumlein:2009cf,Ablinger:2011te},
we have a duplication
relation:
\begin{thm}[Duplication Relation]
 For $a_i \in \N$, $b_i \in \R^+$ and $n\in \N$ we have
$$
\sum{\S{a_m,\ldots,a_1}{\pm b_m,\ldots,\pm b_1;2\;n}}=\frac{1}{2^{\sum_{i=1}^m{a_i}-m}}\S{a_m,\ldots,a_1}{b_m^2,\ldots,b_1^2;n}
$$
where we sum on the left hand side over the $2^m$ possible combinations concerning $\pm$.
\end{thm}
\begin{proof}
We proceed by induction on $m.$ For $m=1$ we get
\begin{eqnarray*}
 \S{a}{-b,2n}	 &=&\sum_{i=1}^{2n}{\frac{(-b)^i}{i^a}}=\sum_{i=1}^{n}{\left(\frac{b^{2i}}{(2i)^a}-\frac{b^{2i-1}}{(2i-1)^a}\right)}
		=\frac{1}{2^a}\S{a}{b^2;n}-\sum_{i=1}^n{\frac{b^{2i-1}}{(2i-1)^a}}\\
&=&\frac{1}{2^a}\S{a}{b^2;n}-\sum_{i=1}^n{\left(\frac{b^{2i}}{(2i)^a}+\frac{b^{2i-1}}{(2i-1)^a}-\frac{b^{2i}}{(2i)^a}\right)}\\
	&=&\frac{1}{2^a}\S{a}{b^2;n}-\S{a}{b;2n}+\frac{1}{2^a}\S{a}{b^2;n}\\
	&=&\frac{1}{2^{a-1}}\S{a}{b^2,n}-\S{a}{b;2n}.
\end{eqnarray*}
Supposing that the theorem holds for $m\geq1$, we obtain
\begin{eqnarray*}
&&\sum{\S{a_{m+1}, \ldots,a_{1}}{\pm b_{m+1},\ldots, \pm b_1;2n}}=\\
	&&\hspace{2cm}= \sum_{i=1}^n\biggl(\frac{(\pm b_{m+1}^{2i})}{(2i)^{a_{m+1}}}\sum{\S{a_m,\ldots,a_1}{\pm b_m,\ldots,\pm b_1;2\;i}}\biggr.\\
		&&\hspace{2.5cm}+\biggl.\frac{\pm b_{m+1}^{2i-1}}{^(2i-1)^{a_{m+1}}}\sum{\S{a_m,\ldots,a_1}{\pm b_m,\ldots,\pm b_1;2\;i+1}}\biggr)\\
	 &&\hspace{2cm}=2\sum_{i=1}^n\frac{(b^2)^i}{(2i)^{a_{m+1}}}\sum{\S{a_m,\ldots,a_1}{\pm b_m,\ldots,\pm b_1;2\;i}}\\
	 &&\hspace{2cm}=2\sum_{i=1}^n\frac{(b^2)^i}{(2i)^{a_{m+1}}}\frac{1}{2^{\sum_{i=1}^m a_i-m}}{\S{a_m,\ldots,a_1}{b_m^2,\ldots,b_1^2;i}}\\
	&&\hspace{2cm}=\frac{1}{2^{\sum_{i=1}^{m+1} a_i-(m+1)}}\S{a_{m+1},a_m,\ldots,a_1}{\pm b_{m+1},\pm b_m,\ldots,\pm b_1;n}.
\end{eqnarray*}
\end{proof}

Considering again the example from above, the duplication relations are
\begin{align*}
\textnormal{S}_{3,3}\hspace{-0.2em}\left(\tfrac{1}{2},\tfrac{1}{2};2 n\right)=& \frac{1}{16}\textnormal{S}_{3,3}\hspace{-0.2em}\left(\tfrac{1}{4},\tfrac{1}{4};n\right)-\textnormal{S}_{3,3}\hspace{-0.2em}\left(\tfrac{-1}{2},\tfrac{-1}{2};2n\right)-\textnormal{S}_{3,3}\hspace{-0.2em}\left(\tfrac{-1}{2},\tfrac{1}{2};2 n\right)-\textnormal{S}_{3,3}\hspace{-0.2em}\left(\tfrac{1}{2},\tfrac{-1}{2};2 n\right)\\
\textnormal{S}_{1,1}\hspace{-0.2em}\left(\tfrac{1}{2},\tfrac{1}{2};2n\right)=& \textnormal{S}_{1,1}\left(\tfrac{1}{4},\tfrac{1}{4};n\right)-\textnormal{S}_{1,1}\hspace{-0.2em}\left(\tfrac{-1}{2},\tfrac{-1}{2};2 n\right)-\textnormal{S}_{1,1}\hspace{-0.2em}\left(\tfrac{-1}{2},\tfrac{1}{2};2n\right)-\textnormal{S}_{1,1}\left(\tfrac{1}{2},\tfrac{-1}{2};2n\right)\\
\textnormal{S}_{3,1}\hspace{-0.2em}\left(\tfrac{1}{2},\tfrac{1}{2};2 n\right)=& \tfrac{1}{4}\textnormal{S}_{3,1}\hspace{-0.2em}\left(\tfrac{1}{4},\tfrac{1}{4};n\right)-\textnormal{S}_{3,1}\hspace{-0.2em}\left(\tfrac{-1}{2},\tfrac{-1}{2};2n\right)-\textnormal{S}_{3,1}\hspace{-0.2em}\left(\tfrac{-1}{2},\tfrac{1}{2};2 n\right)-\textnormal{S}_{3,1}\hspace{-0.2em}\left(\tfrac{1}{2},\tfrac{-1}{2};2 n\right)\\
\textnormal{S}_{1,3}\hspace{-0.2em}\left(\tfrac{1}{2},\tfrac{1}{2};2n\right)=& \tfrac{1}{4}\textnormal{S}_{1,3}\hspace{-0.2em}\left(\tfrac{1}{4},\tfrac{1}{4};n\right)-\textnormal{S}_{1,3}\hspace{-0.2em}\left(\tfrac{-1}{2},\tfrac{-1}{2};2 n\right)-\textnormal{S}_{1,3}\hspace{-0.2em}\left(\tfrac{-1}{2},\tfrac{1}{2};2 n\right)-\textnormal{S}_{1,3}\hspace{-0.2em}\left(\tfrac{1}{2},\tfrac{-1}{2};2 n\right).
\end{align*}
Note that from duplication we do not get a further reduction in this particular example. We would have to introduce new sums of the same depth and weight to express the basis sums~\eqref{Equ:BasisSums}.

\subsection{Examples for Specific Index Sets}

In this Subsection we want to present the numbers of basis sums for specific index sets
at
special depths or weights. In the Tables \ref{SSdepth2table}, \ref{SSdepth3table} and
\ref{SSdepth4table} we summarize the number of algebraic basis sums at the possible index
sets at depths $2,3,4$ respectively. To illustrate how these tables are to be understood we
take a closer look at the depth $d=3$ example with the index sets $\{a_1,a_1,a_2\}$ and
$\{x_1,x_2,x_3)\}$; these sets stand for  multi-sets and each element is allowed to be
taken once to build all the possible $S$-sums. There we obtain the $18$ sums:

\small
\begin{eqnarray*}
&& \S{a_1, a_1, a_2}{x_1, x_2, x_3;n}, \S{a_1, a_1, a_2}{x_1, x_3, x_2;n}, \S{a_1, a_1, a_2}{x_2, x_1, x_3;n}, \\
&& \S{a_1, a_1, a_2}{x_2, x_3, x_1;n}, \S{a_1, a_1, a_2}{x_3, x_1, x_2;n}, \S{a_1, a_1, a_2}{x_3, x_2, x_1;n}, \\
&& \S{a_1, a_2, a_1}{x_1, x_2, x_3;n}, \S{a_1, a_2, a_1}{x_1, x_3, x_2;n}, \S{a_1, a_2, a_1}{x_2, x_1, x_3;n}, \\
&& \S{a_1, a_2, a_1}{x_2, x_3, x_1;n}, \S{a_1, a_2, a_1}{x_3, x_1, x_2;n}, \S{a_1, a_2, a_1}{x_3, x_2, x_1;n}, \\
&& \S{a_2, a_1, a_1}{x_1, x_2, x_3;n}, \S{a_2, a_1, a_1}{x_1, x_3, x_2;n}, \S{a_2, a_1, a_1}{x_2, x_1, x_3;n}, \\
&& \S{a_2, a_1, a_1}{x_2, x_3, x_1;n}, \S{a_2, a_1, a_1}{x_3, x_1, x_2;n}, \S{a_2, a_1, a_1}{x_3, x_2, x_1;n}.
\end{eqnarray*}

\normalsize
\noindent
Using all applicable algebraic relations (provided by the quasi-shuffle algebra) like, e.g.,
\small
\begin{eqnarray*}
&&\S{a_1, a_1, a_2}{x_1, x_2, x_3;n} = \S{a_2}{x_3;n} \S{a_1, a_1}{x_1, x_2;n} + \S{a_1, a_1 + a_2}{x_1, x_2 x_3;n}\\
				&&\hspace{1cm} - \S{a_1}{x_1;n} \S{a_2, a_1}{x_3, x_2;n} - \S{a_2, 2 a_1}{x_3, x_1 x_2;n} + \S{a_2, a_1, a_1}{x_3, x_2, x_1;n}
\end{eqnarray*}
\normalsize
we can express all these $18$ sums using the following $6$ basis sums (with the price
of introducing additional $S$-sums of lower depth subject to the relations given by the
quasi shuffle algebra):
\small
\begin{eqnarray*}
&& \S{a_2, a_1, a_1}{x_1, x_2, x_3;n}, \S{a_2, a_1, a_1}{x_1, x_3, x_2;n}, \S{a_2, a_1, a_1}{x_2, x_1, x_3;n},\\
&& \S{a_2, a_1, a_1}{x_2, x_3, x_1;n}, \S{a_2, a_1, a_1}{x_3, x_1, x_2;n}, \S{a_2, a_1, a_1}{x_3, x_2, x_1;n}.
\end{eqnarray*}
\normalsize

\begin{table}\centering
\begin{tabular}{|c | r | r | r|}
\hline	
Index set& sums &basis sums & dependent sums\\
\hline	
  $\{a_1,a_1\},\{x_1,x_1\}$ &    1 &   0 &   1 \\
  $\{a_1,a_1\},\{x_1,x_2\}$ &    2 &   1 &   1 \\
\hline
  $\{a_1,a_2\},\{x_1,x_1\}$ &    2 &   1 &   1 \\
  $\{a_1,a_2\},\{x_1,x_2\}$ &    4 &   2 &   2 \\
\hline
\end{tabular}
\caption{\label{SSdepth2table}Number of basis sums for different index sets at depth 2.}
\end{table}

\begin{table}\centering
\begin{tabular}{|c | r | r | r|}
\hline	
Index set& sums &basis sums & dependent sums\\
\hline	
  $\{a_1,a_1,a_1\},\{x_1,x_1,x_1\}$ &    1 &   0 &   1  \\
  $\{a_1,a_1,a_1\},\{x_1,x_1,x_2\}$ &    3 &   1 &   2  \\
  $\{a_1,a_1,a_1\},\{x_1,x_2,x_3\}$ &    6 &   2 &   4  \\
\hline	
  $\{a_1,a_1,a_2\},\{x_1,x_1,x_1\}$ &    3 &   1 &   2  \\
  $\{a_1,a_1,a_2\},\{x_1,x_1,x_2\}$ &    9 &   3 &   6  \\
  $\{a_1,a_1,a_2\},\{x_1,x_2,x_3\}$ &   18 &   6 &  12  \\
\hline	
  $\{a_1,a_2,a_3\},\{x_1,x_1,x_1\}$ &    6 &   2 &   4  \\
  $\{a_1,a_2,a_3\},\{x_1,x_1,x_2\}$ &   18 &   6 &  12  \\
  $\{a_1,a_2,a_3\},\{x_1,x_2,x_3\}$ &   36 &  12 &  24  \\
\hline
\end{tabular}
\caption{\label{SSdepth3table}Number of basis sums for different index sets at depth 3.}
\end{table}

\begin{table}[t] \centering
\begin{tabular}{|c | r | r | r|}
\hline	
Index set& sums &basis sums & dependent sums\\
\hline	
  $\{a_1,a_1,a_1,a_1\},\{x_1,x_1,x_1,x_1\}$ &    1 &   0 &   1  \\
  $\{a_1,a_1,a_1,a_1\},\{x_1,x_1,x_1,x_2\}$ &    4 &   1 &   3  \\
  $\{a_1,a_1,a_1,a_1\},\{x_1,x_1,x_2,x_2\}$ &    6 &   1 &   5  \\
  $\{a_1,a_1,a_1,a_1\},\{x_1,x_1,x_2,x_3\}$ &   12 &   3 &   9  \\
  $\{a_1,a_1,a_1,a_1\},\{x_1,x_2,x_3,x_4\}$ &   24 &   6 &  18  \\
\hline
  $\{a_1,a_1,a_1,a_2\},\{x_1,x_1,x_1,x_1\}$ &    4 &   1 &   3  \\
  $\{a_1,a_1,a_1,a_2\},\{x_1,x_1,x_1,x_2\}$ &   16 &   4 &  12  \\
  $\{a_1,a_1,a_1,a_2\},\{x_1,x_1,x_2,x_2\}$ &   24 &   6 &  18  \\
  $\{a_1,a_1,a_1,a_2\},\{x_1,x_1,x_2,x_3\}$ &   48 &  12 &  36  \\
  $\{a_1,a_1,a_1,a_2\},\{x_1,x_2,x_3,x_4\}$ &   96 &  24 &  72  \\
\hline	
  $\{a_1,a_1,a_2,a_2\},\{x_1,x_1,x_1,x_1\}$ &    6 &   1 &   5  \\
  $\{a_1,a_1,a_2,a_2\},\{x_1,x_1,x_1,x_2\}$ &   24 &   6 &  18  \\
  $\{a_1,a_1,a_2,a_2\},\{x_1,x_1,x_2,x_2\}$ &   36 &   8 &  28  \\
  $\{a_1,a_1,a_2,a_2\},\{x_1,x_1,x_2,x_3\}$ &   72 &  18 &  54  \\
  $\{a_1,a_1,a_2,a_2\},\{x_1,x_2,x_3,x_4\}$ &  144 &  36 & 108  \\
\hline	
  $\{a_1,a_1,a_2,a_3\},\{x_1,x_1,x_1,x_1\}$ &   12 &   3 &   9  \\
  $\{a_1,a_1,a_2,a_3\},\{x_1,x_1,x_1,x_2\}$ &   48 &  12 &  36  \\
  $\{a_1,a_1,a_2,a_3\},\{x_1,x_1,x_2,x_2\}$ &   72 &  18 &  54  \\
  $\{a_1,a_1,a_2,a_3\},\{x_1,x_1,x_2,x_3\}$ &  144 &  36 & 108  \\
  $\{a_1,a_1,a_2,a_3\},\{x_1,x_2,x_3,x_4\}$ &  288 &  72 & 216  \\
\hline	
  $\{a_1,a_2,a_3,a_4\},\{x_1,x_1,x_1,x_1\}$ &   24 &   6 &  18  \\
  $\{a_1,a_2,a_3,a_4\},\{x_1,x_1,x_1,x_2\}$ &   96 &  24 &  72  \\
  $\{a_1,a_2,a_3,a_4\},\{x_1,x_1,x_2,x_2\}$ &  144 &  36 & 108  \\
  $\{a_1,a_2,a_3,a_4\},\{x_1,x_1,x_2,x_3\}$ &  288 &  72 & 216  \\
  $\{a_1,a_2,a_3,a_4\},\{x_1,x_2,x_3,x_4\}$ &  576 & 144 & 432  \\
\hline	
\end{tabular}
\caption{\label{SSdepth4table}Number of basis sums for different index sets at depth 4.}
\end{table}

In Table \ref{SSweighttable1} we summarize the number of algebraic basis sums at specified weights for arbitrary indices in the $x_i$, while in Table \ref{SSweighttable2} we summarize
the number of algebraic and differential bases sums for $x_i\in \{1,-1,1/2,-1/2,2,-2\}$ and where each of the indices $\{1/2,-1/2,2,-2\}$ is allowed to appear just once in each sum.
To illustrate how these tables are to be understood, we take a closer look at two examples.
At weight $w=3$ with index set $\{x_1,x_2,x_3\},$ we consider the $19$ sums:
\small
\begin{eqnarray*}
&& \S{3}{x_1 x_2 x_3;n}, \S{1, 2}{x_1, x_2 x_3;n}, \S{1, 2}{x_2, x_1 x_3;n}, \S{1, 2}{x_1 x_2, x_3;n}, \S{1, 2}{x_3, x_1 x_2;n},\\
&& \S{1, 2}{x_1 x_3, x_2;n}, \S{1, 2}{x_2 x_3, x_1;n}, \S{2, 1}{x_1, x_2 x_3;n}, \S{2, 1}{x_2, x_1 x_3;n}, \S{2, 1}{x_1 x_2, x_3;n},\\
&& \S{2, 1}{x_3, x_1 x_2;n}, \S{2, 1}{x_1 x_3, x_2;n}, \S{2, 1}{x_2 x_3, x_1;n}, \S{1, 1, 1}{x_1, x_2, x_3;n},  \S{1, 1, 1}{x_1, x_3, x_2;n},\\
&& \S{1, 1, 1}{x_2, x_1, x_3;n}, \S{1, 1, 1}{x_2, x_3, x_1;n}, \S{1, 1, 1}{x_3, x_1, x_2;n},  \S{1, 1, 1}{x_3, x_2, x_1;n}.
\end{eqnarray*}
\normalsize
Using all applicable algebraic relations (provided by the quasi-shuffle algebra) like, e.g.,
\small
\begin{eqnarray*}
\S{1, 2}{x_2 x_3, x_1;n} &=& \S{1}{x_2 x_3;n} \S{2}{x_1;n} + \S{3}{x_1 x_2 x_3;n} - \S{2, 1}{x_1, x_2 x_3;n}, \\
\S{1, 1, 1}{x_1, x_2, x_3;n} &=& -\S{1}{x_3;n} \S{1, 1}{x_2, x_1;n} + \S{1}{x_1;n} \S{1, 1}{x_2, x_3;n}\\
					&& + \S{2, 1}{x_1 x_2, x_3;n}- \S{2, 1}{x_2 x_3, x_1;n} + \S{1, 1, 1}{x_3, x_2, x_1;n}
\end{eqnarray*}
\normalsize
we can express all these $19$ sums using the following $9$ basis sums:
\small
\begin{eqnarray*}
&&\S{3}{x_1 x_2 x_3;n}, \S{2, 1}{x_1, x_2 x_3;n}, \S{2, 1}{x_2, x_1 x_3;n}, \S{2, 1}{x_1 x_2, x_3;n}, \S{2, 1}{x_3, x_1 x_2;n},\\
&&\S{2, 1}{x_1 x_3, x_2;n}, \S{2, 1}{x_2 x_3, x_1;n}, \S{1, 1, 1}{x_3, x_1, x_2;n}, \S{1, 1, 1}{x_3, x_2, x_1;n}.
\end{eqnarray*}
\normalsize

\begin{table}[t]\centering
\begin{tabular}{|c | c | r | r | r|}
\hline	
weight&index set& sums &basis sums & dependent sums\\
\hline	
 2& $\{x_1,x_1\}$ &    2 &   1 &   1 \\
  & $\{x_1,x_2\}$ &    3 &   3 &   1 \\
\hline
 3& $\{x_1,x_1,x_1\}$ &    6 &   3 &   3 \\
  & $\{x_1,x_1,x_2\}$ &   12 &   6 &   6 \\
  & $\{x_1,x_2,x_3\}$ &   19 &   9 &  10 \\
\hline
 4& $\{x_1,x_1,x_1,x_1\}$ &   20 &   8 &  12 \\
  & $\{x_1,x_1,x_1,x_2\}$ &   50 &  20 &  30 \\
  & $\{x_1,x_1,x_2,x_2\}$ &   64 &  24 &  40 \\
  & $\{x_1,x_1,x_2,x_3\}$ &  106 &  40 &  66 \\
  & $\{x_1,x_2,x_3,x_4\}$ &  175 &  64 & 111 \\
\hline
 5& $\{x_1,x_1,x_1,x_1,x_1\}$ &   70 &  25 &   45 \\
  & $\{x_1,x_1,x_1,x_1,x_2\}$ &  210 &  70 &  140 \\
  & $\{x_1,x_1,x_1,x_2,x_2\}$ &  325 & 105 &  220 \\
  & $\{x_1,x_1,x_1,x_2,x_3\}$ &  555 & 175 &  380 \\
  & $\{x_1,x_1,x_2,x_2,x_3\}$ &  725 & 225 &  500 \\
  & $\{x_1,x_1,x_2,x_3,x_4\}$ & 1235 & 375 &  860 \\
  & $\{x_1,x_2,x_3,x_4,x_5\}$ & 2101 & 625 & 1476 \\
\hline
 6& $\{x_1,x_1,x_1,x_1,x_1,x_1\}$ &   252 &   75 &   177\\
  & $\{x_1,x_1,x_1,x_1,x_1,x_2\}$ &   882 &  252 &   630\\
  & $\{x_1,x_1,x_1,x_1,x_2,x_2\}$ &  1596 &  438 &  1158\\
  & $\{x_1,x_1,x_1,x_2,x_2,x_2\}$ &  1911 &  522 &  1389\\
  & $\{x_1,x_1,x_1,x_1,x_2,x_3\}$ &  2786 &  756 &  2030\\
  & $\{x_1,x_1,x_1,x_2,x_2,x_3\}$ &  4431 & 1176 &  3255\\
  & $\{x_1,x_1,x_2,x_2,x_3,x_3\}$ &  5886 & 1539 &  4347\\
  & $\{x_1,x_1,x_1,x_2,x_3,x_4\}$ &  7721 & 2016 &  5705\\
  & $\{x_1,x_1,x_2,x_2,x_3,x_4\}$ & 10251 & 2646 &  7605\\
  & $\{x_1,x_1,x_2,x_3,x_4,x_5\}$ & 17841 & 4536 & 13305 \\
  & $\{x_1,x_2,x_3,x_4,x_5,x_6\}$ & 31031 & 7776 & 23255\\
\hline
\end{tabular}
\caption{\label{SSweighttable1}Number of basis sums for different index sets up to weight 6.}
\end{table}

\noindent At weight $w=2$ where each of the indices from the index stet $\{1/2,-1/2,2,-2\}$ is allowed to appear just once, we consider the $38$ sums:
\small
\begin{eqnarray*}
&&\S{-2}{n},\S{2}{n},\S{-1,-1}{n},\S{-1,1}{n},\S{1,-1}{n},\S{1,1}{n},\S{2}{-2;n},\S{2}{-\tfrac{1}{2};n},\\
&&\S{2}{\tfrac{1}{2};n},\S{2}{2;n},\S{1,1}{-2,-1;n},\S{1,1}{-2,-\tfrac{1}{2};n},\S{1,1}{-2,\tfrac{1}{2};n},\S{1,1}{-2,1;n},\\
&&\S{1,1}{-2,2;n},\S{1,1}{-1,-2;n},\S{1,1}{-1,-\tfrac{1}{2};n},\S{1,1}{-1,\tfrac{1}{2};n},\S{1,1}{-1,2;n},\\
&&\S{1,1}{-\tfrac{1}{2},-2;n},\S{1,1}{-\tfrac{1}{2},-1;n},\S{1,1}{-\tfrac{1}{2},\tfrac{1}{2};n},\S{1,1}{-\tfrac{1}{2},1;n},\S{1,1}{-\tfrac{1}{2},2;n},\\
&&\S{1,1}{\tfrac{1}{2},-2;n},\S{1,1}{\tfrac{1}{2},-1;n},\S{1,1}{\tfrac{1}{2},-\tfrac{1}{2};n},\S{1,1}{\tfrac{1}{2},1;n},\S{1,1}{\tfrac{1}{2},2;n},\\
&&\S{1,1}{1,-2;n},\S{1,1}{1,-\tfrac{1}{2};n},\S{1,1}{1,\tfrac{1}{2};n},\S{1,1}{1,2;n},\S{1,1}{2,-2;n},\S{1,1}{2,-1;n},\\
&&\S{1,1}{2,-\tfrac{1}{2};n},\S{1,1}{2,\tfrac{1}{2};n},\S{1,1}{2,1;n}.
\end{eqnarray*}
\normalsize
Using all algebraic and differential relations like, e.g.,
\small
\begin{align*}
\S{1,1}{-\tfrac{1}{2},2;n}=&-\frac{\partial}{\partial n}\S{-1}{n}+\H{-1,0}{1}+\S{1}{-\tfrac{1}{2};n}\;\S{1}{2;n}-\S{1,1}{2,-\tfrac{1}{2};n}\\
\S{1,1}{-\tfrac{1}{2},-2;n}=&-\frac{\partial}{\partial n}\S{1}{n}-\H{1,0}{1}+\S{1}{-2;n}\;\S{1}{-\tfrac{1}{2};n}-\S{1,1}{-2,-\tfrac{1}{2};n}
\end{align*}
\normalsize
we can express all these $38$ sums using the following $17$ basis sums:
\small
\begin{eqnarray*}
&&\S{1,-1}{n},\S{1,1}{-2,-\tfrac{1}{2};n},\S{1,1}{-2,\tfrac{1}{2};n},\S{1,1}{-2,2;n},\S{1,1}{-1,-2;n},\\
&&\S{1,1}{-1,-\tfrac{1}{2};n},\S{1,1}{-1,\tfrac{1}{2};n},\S{1,1}{-1,2;n},\S{1,1}{-\tfrac{1}{2},\tfrac{1}{2};n},\S{1,1}{\tfrac{1}{2},-\tfrac{1}{2};n},\\
&&\S{1,1}{1,-2;n},\S{1,1}{1,-\tfrac{1}{2};n},\S{1,1}{1,\tfrac{1}{2};n},\S{1,1}{1,2;n},\S{1,1}{2,-2;n},\\
&&\S{1,1}{2,-\tfrac{1}{2};n},\S{1,1}{2,\tfrac{1}{2};n}.
\end{eqnarray*}
\normalsize

We remark that the found relations are induced by the quasi-shuffle algebra. Using the
summation package {\tt Sigma}~\cite{Schneider:2007} and the underlying difference field
theory~\cite{Karr1981,Schneider:2010,Schneider:2010b,Schneider:2010c} we could verify
that the found basis sums are also algebraic independent in the ring as sequences; for related work see also~\cite{Singer:08}.
To be more precise, the given $S$-sums evaluated as sequences form a polynomial ring which is a subring of the ring of sequences.

In general, we obtain the following relations
\begin{eqnarray*}
N_D(w)&=&N_S(w)-N_S(w-1)\\
N_{AD}(w)&=&N_A(w)-N_A(w-1)
\end{eqnarray*}
where $N_S(w),N_D(w)$ and $N_{AD}(w)$ are the number of sums, the number algebraic basis sums and the number of basis sums using algebraic and differential relations at weight $w$ respectively.

\begin{table}[t]\centering
\begin{tabular}{|r | r | r | r | r|}
\hline	
weight& sums &a-basis sums & d-basis sums & ad-basis sums\\
\hline
 1&     6 &      6 &     6 &     6\\	
 2&    38 &     23 &    32 &    17\\
 3&   222 &    120 &   184 &    97\\
 4&  1206 &    654 &   984 &   543\\
 5&  6150 &   3536 &  4944 &  2882\\
 6& 29718 &  18280 & 23568 & 14744\\
\hline
\end{tabular}
\caption{\label{SSweighttable2}Number of basis sums up to weight 6 with index set $\{1,-1,1/2,-1/2,2,-2\}.$ Each of the indices $\{1/2,-1/2,2,-2\}$ is allowed to appear just once in each sum.}
\end{table}

\section{Relations between $S$-Sums at Infinity}\label{Sec:InfiniteSSums}

\vspace*{1mm}
\noindent
In this Section we will state several types of relations between the values of $S$-sums
at infinity
\begin{equation}\label{Equ:InfiniteSSums}
\S{m_1,\ldots,m_p}{x_1,\ldots,x_p;\infty}
\end{equation}
(or equivalently between the values of generalized polylogarithms) that extend various
relations for infinite harmonic sums given in~\cite{Blumlein:2009cf}. This enables one
to reduce the expressions that arise during the calculations of, e.g., the (inverse) Mellin
transform, of the differentiation of $S$-sums, or of asymptotic expansions given below. We remark
that the relations of $S$-sums  with $x_i$ being roots of unity have been considered
in~\cite{Ablinger:2011te} already.

\begin{description}
\item[Stuffle relations:]
The first type of relations originates from the algebraic
relations of $S$-sums, see Section \ref{SSalgrel}.
Those relations which contain only convergent $S$-sums remain valid when we
consider them at infinity\footnote{In general, relations containing sums of
logarithmic growth such as $S_1(n)$ can be also formally utilized in asymptotic
expansions. However, in the following only convergent sums are considered.}.
We will refer to these relations as the quasi shuffle or stuffle relations.

\item[Duplication relations:]
The duplication relations from Section \ref{SSduplrel} remain valid if we consider sums
which are finite at infinity, since the limits  $n$ resp. $2n \rightarrow \infty$
agree.

\item[Shuffle relations:]
We can generalize the shuffle relations form
\cite{Vermaseren:1998uu} for harmonic sums to $S$-sums. For convergent sums we have:
\begin{eqnarray*}
&&\S{m_1,\ldots,m_p}{x_1,\ldots,x_p;\infty}\S{k_1,\ldots,k_q}{y_1,\ldots,y_p;\infty}=\\
	&&\hspace{2cm}\lim_{n \rightarrow \infty}\sum_{i=1}^n\frac{y_1^i \S{m_1,\ldots,m_p}{x_1,\ldots,x_p;n-i}\S{k_2,\ldots,k_q}{y_2,\ldots,y_p;i}} {i^{k_1}}.
\end{eqnarray*}
Using
\begin{eqnarray*}
&&\lim_{n \rightarrow \infty}\sum_{i=1}^n\frac{y_1^i \S{m_1,\ldots,m_p}{x_1,\ldots,x_p;n-i}\S{k_2,\ldots,k_q}{y_2,\ldots,y_p;i}} {i^{k_1}}=\\
&&\hspace{1cm}\sum_{a=1}^{k_1} \binom{k_1+m_1-1-a}{m_1-1}\sum_{i=1}^n\frac{x_1^i}{i^{m_1+k_1-a}}\\
&&\hspace{2cm}\sum_{j=1}^i\frac{\left(\frac{y_1}{x_1}\right)^j\S{m_2,\ldots,m_p}{x_2,\ldots,x_p;i-j}\S{k_2,\ldots,k_q}{y_2,\ldots,y_p;j}} {j^a}\\
&&\hspace{1cm}+\sum_{a=1}^{m_1} \binom{k_1+m_1-1-a}{m_1-1}\sum_{i=1}^n\frac{y_1^i}{i^{m_1+k_1-a}}\\
&&\hspace{2cm}\sum_{j=1}^i\frac{\left(\frac{x_1}{y_1}\right)^j\S{k_2,\ldots,k_q}{y_2,\ldots,y_p;i-j}\S{m_2,\ldots,m_p}{x_2,\ldots,x_p;j}}{j^a}
\end{eqnarray*}
we can rewrite the right hand side in terms of $S$-sums. We will refer to these
relations as the shuffle relations since one could also obtain them from the shuffle algebra
of generalized polylogarithms.

\item[Duality relations:]
We can use duality relations of generalized polylogarithms;
compare \cite{Blumlein:2009cf} where these relations are considered for harmonic sums.

\begin{description}
\item[$1-x\rightarrow x$:]
For a generalized polylogarithm $\H{a_1,a_2,\ldots,a_{k-1},a_k}1$ with $a_k\neq 0$ this leads to the relation:
\begin{eqnarray*}
\H{a_1,a_2,\ldots,a_{k-1},a_k}1=\H{1-a_k,1-a_{k-1},\ldots,1-a_2,1-a_1}1.
\end{eqnarray*}
E.g., we get $\H{2}1=\H{-1}1$
which gives in $S$-sum notation the new relation
$$\S{1}{\tfrac{1}{2};\infty}=-\S{-1}\infty.$$
\item[$\frac{1-x}{1+x}\rightarrow x$:] This transform leads, e.g., to the relation
$\H{2}1=-\H{-1}1+\H{-\frac{1}{3}}1$
which gives in $S$-sum notation again the relation
$$\S{1}{\tfrac{1}{2};\infty}=-\S{-1}\infty.$$
\item[$\frac{c-x}{d+x}\rightarrow x$:]For $c,d\in \R,\ d\neq -1$ this transform can be calculated similarly as described in Section~\ref{SS1x1x}.
The transform $\frac{2-x}{1+x}\rightarrow x$ produces, e.g., the relation
$\H{-2}1=\frac{1}{2}(\H{-8}1+\H{-\tfrac{1}{2}}1)$
which gives in $S$-sum notation the new relation
$$
\S{1}{-\tfrac{1}{8};\infty}=-\S{-1}\infty+\S{1}{-\tfrac{1}{2};\infty}.
$$
\end{description}
\end{description}
If one considers a certain class of infinite $S$-sums~\eqref{Equ:InfiniteSSums} (with finite values) with given weight and a finite index set from which the $x_i$ can be chosen, one can try to determine all relations from above that stay within the specified class.  In general, we first have to
extend the alphabet generated by the relations used. Then one
exploits all relevant relations that lead back (after several applications) to the original alphabet.

In the following Subsections we apply this tactic to a series of alphabets. We summarize in Table~\ref{SSdepth4table} how all the so far calculated relations for the particular alphabets contribute in reducing the number of infinite $S$-sums to a smaller set of sums.
For the found constants we searched for new relations using {\tt PSLQ}
\cite{PSLQ}
using representations of up to 2000 digits. No further relations were found.
These constants, which we also call \textit{basis constants}, enlarge the number of constants generated by the
multiple zeta values \cite{Blumlein:2009cf}. They are further extended by other
special numbers related to cyclotomic harmonic sums and their generalization
at infinity, cf.~e.g.~\cite{Broadhurst:1998rz,Ablinger:2011te}.

Subsequently, the found basis constants for the concrete alphabets are given.
In the simpler cases also the relations between the
dependent sums are printed explicitly~\footnote{In \texttt{http://www.risc.jku.at/research/combinat/software/HarmonicSums/} all relations are available up to weight four. The most involved calculation is the alphabet $x_i\in\{1, -1, 1/2, -1/2, 2, -2\}$; the relations are obtained executing
$\texttt{ComputeSSumInfBasis}\left[4, \{1, -1, 1/2, -1/2, 2, -2\}, \texttt{ExtendAlphabet}\to\texttt{True}\right]$  in about 36 hours.}. Given an expression with the corresponding weight and alphabet and given these relations, one can bring the expression to its reduced form, i.e., in terms of the basis constants. Here one needs to replace the $S$-sum on the left hand side of a given relation by its right hand side.

\begin{table}\centering
\begin{tabular}{|c | r | r | r| r|}
\hline	
alphabet& weight & finite sums& dependent sums& basis sums\\
\hline	
									& 1&    1 &   0 &   1  \\
 $\{1,-1\}$								& 2&    4 &   3 &   1  \\
									& 3&   12 &  11 &   1  \\
									& 4&   36 &  35 &   1  \\
\hline	
									& 1&    1 &   0 &   1  \\
 $\{1,\frac{1}{2}\}$							& 2&    4 &   2 &   2  \\
									& 3&   12 &   8 &   4  \\
									& 4&   36 &  24 &  12  \\
\hline	
									& 1&    1 &   0 &   1  \\
 $\{1,\frac{1}{2},2\}$							& 2&    5 &   3 &   2  \\
									& 3&   18 &  14 &   4  \\
									& 4&   66 &  53 &  13  \\
\hline
									& 1&    2 &   1 &   1  \\
  $\{1,-1,\frac{1}{2}\}$						& 2&    9 &   7 &   2  \\
									& 3&   36 &  29 &   7  \\
									& 4&  144 & 113 &  31  \\
\hline	
									& 1&    3 &   1 &   2  \\
  $\{1,-1,\frac{1}{2},-\frac{1}{2}\}$					& 2&   16 &  12 &   4  \\
									& 3&   80 &  63 &  17  \\
									& 4&  400 & 301 &  99  \\
\hline	
									& 1&    3 &   1 &   2  \\
  $\{1,-1,\frac{1}{2},-\frac{1}{2},2,-2\}$				& 2&   20 &  16 &   4  \\
									& 3&  124 & 107 &  17  \\
									& 4&  788 & 678 & 110  \\

\hline	
\end{tabular}
\caption{\label{SSInfinitydepth4table}Number of basis constants  (convergent infinite $S$-sums) for different
alphabets up to weight 4.}
\end{table}

\subsection{Alphabet $\{1,\frac{1}{2}\}:$}

\noindent
\begin{description}
 \item[w=1:]\  \\
	      Basis constant: $S_1\left(\frac{1}{2};\infty \right)=\log(2)$\\
	      Relation: ---

 \item[w=2:]\  \\
	      Basis constants: $S_2(\infty
)=\zeta_2,S_{1,1}\left(\frac{1}{2},\frac{1}{2};\infty \right)$\\
	      Relations:
	      \begin{eqnarray*}
	      S_{1,1}\left(\frac{1}{2},1;\infty \right)&=&\frac{\zeta_2}{2}\\
	      S_2\left(\frac{1}{2};\infty \right)&=&-\frac{1}{2} \log ^2(2)+\frac{\zeta_2}{2}
	      \end{eqnarray*}
 \item[w=3:]\  \\
	      Basis constants: $S_3(\infty
)=\zeta_3,S_{1,2}\left(\frac{1}{2},\frac{1}{2};\infty \right),S_{1,1,1}\left(\frac{1}{2},\frac{1}{2},\frac{1}{2};\infty \right),S_{1,1,1}\left(\frac{1}{2},\frac{1}{2},1;\infty \right)$\\
	      Relations:
	      \begin{eqnarray*}
		S_3\left(\frac{1}{2};\infty \right)&=&\frac{\log ^3(2)}{6}-\frac{1}{2} \log (2) \zeta_2+\frac{7 \zeta_3}{8}\\
		S_{2,1}\left(\frac{1}{2},\frac{1}{2};\infty \right)&=&\frac{\log ^3(2)}{6}+\frac{1}{2} \log (2) \zeta_2-\frac{7 \zeta_3}{6}+3  S_{1,2}\left(\frac{1}{2},\frac{1}{2};\infty \right)\\
		S_{2,1}\left(\frac{1}{2},1;\infty \right)&=&-\frac{1}{2} (\log (2) \zeta_2)+\zeta_3\\
		S_{2,1}\left(1,\frac{1}{2};\infty \right)&=&\frac{\log ^3(2)}{6}+\frac{1}{2} \log (2) \zeta_2+\frac{\zeta_3}{4}\\
		S_{2,1}(\infty )&=&2 \zeta_3\\
		S_{1,2}\left(\frac{1}{2},1;\infty \right)&=&\frac{5 \zeta_3}{8}\\
		S_{1,1,1}\left(\frac{1}{2},1,\frac{1}{2};\infty \right)&=&\frac{1}{2} \log (2) \zeta_2-\frac{5 \zeta(3)}{12}+S_{1,2}\left(\frac{1}{2},\frac{1}{2};\infty \right)\\
		S_{1,1,1}\left(\frac{1}{2},1,1;\infty \right)&=&\frac{3 \zeta_3}{4}
	      \end{eqnarray*}
\end{description}

\subsection{Alphabet $\{1,\frac{1}{2},2\}:$}

\noindent
\begin{description}
 \item[w=1:]\  \\
	      Basis constant: $S_1\left(\frac{1}{2};\infty \right)=\log (2)$\\
	      Relation: ---
	      \begin{eqnarray*}
	      \end{eqnarray*}

 \item[w=2:]\  \\
	      Basis constants: $S_2(\infty
)=\zeta_2,S_{1,1}\left(\frac{1}{2},\frac{1}{2};\infty \right)$\\
	      Relations:
	      \begin{eqnarray*}
	      S_2\left(\frac{1}{2};\infty \right)&=&-\frac{1}{2} \log ^2(2)+\frac{\zeta_2}{2}\\
	      S_{1,1}\left(\frac{1}{2},1;\infty \right)&=&\frac{\zeta_2}{2}\\
	      S_{1,1}\left(\frac{1}{2},2;\infty \right)&=&\frac{3 \zeta_2}{2}
	      \end{eqnarray*}
 \item[w=3:]\  \\
	      Basis constants: $S_3(\infty
)=\zeta_3,S_{1,2}\left(\frac{1}{2},\frac{1}{2};\infty \right),S_{1,1,1}\left(\frac{1}{2},\frac{1}{2},\frac{1}{2};\infty \right),S_{1,1,1}\left(\frac{1}{2},\frac{1}{2},1;\infty \right)$\\
	      Relations:
	      \begin{eqnarray*}
S_3\left(\frac{1}{2};\infty \right)&=&-\frac{1}{2} \zeta_2 \log (2)+\frac{\log ^3(2)}{6}+\frac{7 \zeta_3}{8}\\S_{2,1}\left(\frac{1}{2},\frac{1}{2};\infty \right)&=&\frac{1}{2} \zeta_2 \log (2)+\frac{\log ^3(2)}{6}+3
   S_{1,2}\left(\frac{1}{2},\frac{1}{2};\infty \right)-\frac{7 \zeta_3}{6}\\S_{2,1}\left(\frac{1}{2},1;\infty \right)&=&-\frac{1}{2} \zeta_2 \log (2)+\zeta_3\\S_{2,1}\left(\frac{1}{2},2;\infty \right)&=&\frac{1}{2} (-3) \zeta_2 \log (2)+\frac{21 \zeta
   (3)}{8}\\S_{2,1}\left(1,\frac{1}{2};\infty \right)&=&\frac{1}{2} \zeta_2 \log (2)+\frac{\log ^3(2)}{6}+\frac{\zeta_3}{4}\\S_{2,1}(\infty )&=&2 \zeta_3\\S_{1,2}\left(\frac{1}{2},1;\infty \right)&=&\frac{5 \zeta
   (3)}{8}\\S_{1,2}\left(\frac{1}{2},2;\infty \right)&=&\frac{3}{2} \zeta_2 \log (2)\\S_{1,1,1}\left(\frac{1}{2},\frac{1}{2},2;\infty \right)&=&-\frac{1}{2} \zeta_2 \log (2)+\frac{\log ^3(2)}{6}+\log (2)
   S_{1,1}\left(\frac{1}{2},\frac{1}{2};\infty \right)+2 S_{1,2}\left(\frac{1}{2},\frac{1}{2};\infty \right)+\frac{7 \zeta_3}{24}\\S_{1,1,1}\left(\frac{1}{2},1,\frac{1}{2};\infty \right)&=&\frac{1}{2} \zeta_2 \log
   (2)+S_{1,2}\left(\frac{1}{2},\frac{1}{2};\infty \right)-\frac{5 \zeta_3}{12}\\S_{1,1,1}\left(\frac{1}{2},1,1;\infty \right)&=&\frac{3 \zeta_3}{4}\\S_{1,1,1}\left(\frac{1}{2},1,2;\infty \right)&=&\frac{7 \zeta
   (3)}{4}\\S_{1,1,1}\left(\frac{1}{2},2,\frac{1}{2};\infty \right)&=&\frac{3}{2} \zeta_2 \log (2)-\frac{\zeta_3}{4}\\S_{1,1,1}\left(\frac{1}{2},2,1;\infty \right)&=&\frac{3}{2} \zeta_2 \log (2)+\frac{7 \zeta_3}{4}
	      \end{eqnarray*}
\end{description}

\subsection{Alphabet $\{1,-1,\frac{1}{2}\}:$}

\noindent
\begin{description}
 \item[w=1:]\  \\
	      Basis constant: $S_{-1}(\infty )=-\log (2)$\\
	      Relation:
	      \begin{eqnarray*}
	       S_1\left(\frac{1}{2};\infty \right)= -S_{-1}(\infty )
	      \end{eqnarray*}

 \item[w=2:]\  \\
	      Basis constants: $S_{2}(\infty )=\zeta_2,S_{1,1}\left(\frac{1}{2},\frac{1}{2};\infty \right)$\\
	      Relations:
	      \begin{eqnarray*}
	      S_{-2}\left(\infty \right)&=&-\frac{\zeta_2}{2}\\
	      S_2\left(\frac{1}{2};\infty \right)&=&\frac{\zeta_2}{2}-\frac{\log ^2(2)}{2}\\
	      S_{-1,-1}(\infty )&=&\frac{\zeta_2}{2}+\frac{\log ^2(2)}{2}\\
	      S_{1,1}\left(-1,\frac{1}{2};\infty \right)&=&-\frac{\zeta_2}{2}-\frac{\log ^2(2)}{2}+2 S_{1,1}\left(\frac{1}{2},\frac{1}{2};\infty \right)\\
	      S_{-1,1}(\infty )&=&-\frac{\zeta_2}{2}+\frac{\log ^2(2)}{2}\\
	      S_{1,1}\left(\frac{1}{2},-1;\infty \right)&=&-\frac{1}{2} \log ^2(2)-S_{1,1}\left(\frac{1}{2},\frac{1}{2};\infty\right)\\
	      S_{1,1}\left(\frac{1}{2},1;\infty \right)&=&\frac{\zeta_2}{2}\\
	      S_{1,1}\left(\frac{1}{2},2;\infty \right)&=&\frac{3 \zeta_2}{2}
	      \end{eqnarray*}
 \item[w=3:]\  \\
	      Basis constants: $S_3(\infty )=\zeta_3,S_{1,2}\left(-1,\frac{1}{2};\infty \right),S_{1,1,1}\left(-1,\frac{1}{2},\frac{1}{2};\infty \right), S_{1,1,1}\left(-1,\frac{1}{2},1;\infty\right),\\
	      S_{1,1,1}\left(\frac{1}{2},-1,\frac{1}{2};\infty \right),S_{1,1,1}\left(\frac{1}{2},\frac{1}{2},-1;\infty \right),S_{1,1,1}\left(\frac{1}{2},\frac{1}{2},\frac{1}{2};\infty \right)$\\
	      Relations:
	      \begin{eqnarray*}
S_{-3}(\infty )&=&-\frac{1}{4} (3 \zeta_3)
\\
S_3\left(\frac{1}{2};\infty \right)&=&\frac{\log (2)^3}{6}-\frac{\log (2) \zeta_2}{2}+\frac{7 \zeta_3}{8}
\\
S_{-2,-1}(\infty )&=&\frac{3 \log (2) \zeta_2}{2}-\frac{5
   \zeta_3}{8}
\\
S_{2,1}\left(-1,\frac{1}{2};\infty \right)&=&-\frac{\log (2)^3}{18}-\frac{\log (2) \zeta_2}{3}+\frac{7 \zeta_3}{72}+\frac{1}{3} S_{1,2}\left(-1,\frac{1}{2};\infty \right)
\\
S_{-2,1}(\infty )&=&-\frac{1}{8} (5
   \zeta_3)
\\
S_{2,1}\left(\frac{1}{2},-1;\infty \right)&=&\frac{4 \log (2)^3}{9}-\frac{\log (2) \zeta_2}{3}-\frac{4 \zeta_3}{9}-\frac{2}{3} S_{1,2}\left(-1,\frac{1}{2};\infty \right)
\\
S_{2,1}\left(\frac{1}{2},\frac{1}{2};\infty \right)&=&-\frac{1}{2}
   (\log (2) \zeta_2)+\zeta_3+S_{1,2}\left(-1,\frac{1}{2};\infty \right)
\\
S_{2,1}\left(\frac{1}{2},1;\infty \right)&=&-\frac{1}{2} (\log (2) \zeta_2)+\zeta_3
\\
S_{2,-1}(\infty )&=&-\frac{1}{2} (3 \log (2)
   \zeta_2)+\frac{\zeta_3}{4}
\\
S_{2,1}\left(1,\frac{1}{2};\infty \right)&=&\frac{\log (2)^3}{6}+\frac{\log (2) \zeta_2}{2}+\frac{\zeta_3}{4}
\\
S_{2,1}(\infty )&=&2 \zeta_3
\\
S_{-1,-2}(\infty )&=&-\log (2) \zeta_2+\frac{13
   \zeta_3}{8}
\\
S_{-1,2}(\infty )&=&\frac{\log (2) \zeta_2}{2}-\zeta_3
\\
S_{1,2}\left(\frac{1}{2},-1;\infty \right)&=&-\frac{1}{24} (13 \zeta_3)
\\
S_{1,2}\left(\frac{1}{2},\frac{1}{2};\infty \right)&=&-\frac{\log (2)^3}{18}-\frac{\log (2)
   \zeta_2}{3}+\frac{13 \zeta_3}{18}+\frac{1}{3} S_{1,2}\left(-1,\frac{1}{2};\infty \right)
\\
S_{1,2}\left(\frac{1}{2},1;\infty \right)&=&\frac{5 \zeta_3}{8}
\\
S_{-1,-1,-1}(\infty )&=&-\frac{\log (2)^3}{6}-\frac{\log (2)
   \zeta_2}{2}-\frac{\zeta_3}{4}
\\
S_{1,1,1}\left(-1,-1,\frac{1}{2};\infty \right)&=&\frac{17 \log (2)^3}{18}
+\frac{13 \log (2) \zeta_2}{6}-\frac{61 \zeta_3}{36}
-4 \log (2) S_{1,1}\left(\frac{1}{2},\frac{1}{2};\infty \right)
\NN\\ &&
-\frac{8}{3}
   S_{1,2}\left(-1,\frac{1}{2};\infty \right)
\\
S_{-1,-1,1}(\infty )&=&-\frac{\log (2)^3}{6}-\frac{\log (2) \zeta_2}{2}
+\frac{7 \zeta_3}{4}
\\
S_{1,1,1}\left(-1,\frac{1}{2},-1;\infty \right)&=&-\frac{1}{3} \left(2 \log (2)^3\right)-3 \log (2)
   \zeta_2+\frac{85 \zeta_3}{24}
+4 \log (2) S_{1,1}\left(\frac{1}{2},\frac{1}{2};\infty \right)
\NN\\ &&
+4 S_{1,2}\left(-1,\frac{1}{2};\infty \right)
\\
S_{-1,1,-1}(\infty )&=&-\frac{\log (2)^3}{6}+\frac{\log (2)
   \zeta_2}{2}+\frac{\zeta_3}{8}
\\
S_{1,1,1}\left(-1,1,\frac{1}{2};\infty \right)&=&\frac{5 \log (2)^3}{18}-\frac{5 \log (2) \zeta_2}{6}+\frac{61 \zeta_3}{72}+\frac{4}{3} S_{1,2}\left(-1,\frac{1}{2};\infty \right)
\\
S_{-1,1,1}(\infty
   )&=&-\frac{\log (2)^3}{6}+\frac{\log (2) \zeta_2}{2}-\frac{7 \zeta_3}{8}
\\
S_{1,1,1}\left(\frac{1}{2},-1,-1;\infty \right)&=&\frac{7 \log (2)^3}{18}+\frac{4 \log (2) \zeta_2}{3}-\frac{35 \zeta_3}{36}-\log (2)
   S_{1,1}\left(\frac{1}{2},\frac{1}{2};\infty \right)
\NN\\ &&
-\frac{4}{3} S_{1,2}\left(-1,\frac{1}{2};\infty \right)
\\
S_{1,1,1}\left(\frac{1}{2},-1,1;\infty \right)&=&\frac{2 \log (2)^3}{9}
+\frac{\log (2) \zeta_2}{3}-\frac{155 \zeta_3}{144}-\frac{1}{3}
   S_{1,2}\left(-1,\frac{1}{2};\infty \right)
\NN\\ &&
-\frac{1}{2} S_{1,1,1}\left(-1,\frac{1}{2},1;\infty \right)
\\
S_{1,1,1}\left(\frac{1}{2},\frac{1}{2},1;\infty \right)&=&\frac{\log (2)^3}{18}-\frac{7 \log (2) \zeta_2}{6}+\frac{217 \zeta_3}{144}
+\log (2)
   S_{1,1}\left(\frac{1}{2},\frac{1}{2};\infty \right)
\NN\\ &&
+\frac{2}{3} S_{1,2}\left(-1,\frac{1}{2};\infty \right)
+\frac{1}{2} S_{1,1,1}\left(-1,\frac{1}{2},1;\infty \right)
\\
S_{1,1,1}\left(\frac{1}{2},1,-1;\infty \right)&=&-\frac{1}{9} \left(2
   \log (2)^3\right)-\frac{5 \log (2) \zeta_2}{6}+\frac{11 \zeta_3}{36}+\frac{1}{3} S_{1,2}\left(-1,\frac{1}{2};\infty \right)
\\
S_{1,1,1}\left(\frac{1}{2},1,\frac{1}{2};\infty \right)&=&-\frac{\log (2)^3}{18}+\frac{\log (2) \zeta_2}{6}+\frac{11
   \zeta_3}{36}+\frac{1}{3} S_{1,2}\left(-1,\frac{1}{2};\infty \right)
\\S_{1,1,1}\left(\frac{1}{2},1,1;\infty \right)&=&\frac{3 \zeta_3}{4}
	      \end{eqnarray*}
\end{description}

\subsection{Alphabet $\{1,-1,\frac{1}{2},-\frac{1}{2}\}:$}

\noindent
\begin{description}
 \item[w=1:]\  \\
	      Basis constant: $S_{-1}(\infty )=-\log (2),S_1\left(-\frac{1}{2};\infty \right)$\\

 \item[w=2:]\  \\
	      Basis constants: $S_{2}(\infty )=\zeta_2,S_{1,1}\left(-\frac{1}{2},\frac{1}{2};\infty \right),S_{1,1}\left(\frac{1}{2},-\frac{1}{2};\infty \right),S_{1,1}\left(\frac{1}{2},\frac{1}{2};\infty \right)$\\
 \item[w=3:]\  \\
	      Basis constants: $ \displaystyle S_{3}(\infty )=\zeta_3,
S_3\left(-\frac{1}{2};\infty \right),
S_{1,2}\left(-\frac{1}{2},-\frac{1}{2};\infty \right),
S_{1,2}\left(-\frac{1}{2},\frac{1}{2};\infty \right),
\\
S_{1,2}\left(\frac{1}{2},-\frac{1}{2};\infty \right),
S_{2,1}\left(-\frac{1}{2},\frac{1}{2};\infty \right),
S_{1,1,1}\left(-1,-\frac{1}{2},-\frac{1}{2};\infty \right),
S_{1,1,1}\left(-1,-\frac{1}{2},\frac{1}{2};\infty \right),
\\
S_{1,1,1}\left(-1,\frac{1}{2},-\frac{1}{2};\infty \right),
S_{1,1,1}\left(-\frac{1}{2},-1,-\frac{1}{2};\infty \right),
S_{1,1,1}\left(-\frac{1}{2},-\frac{1}{2},\frac{1}{2};\infty \right),
\\
S_{1,1,1}\left(-\frac{1}{2},\frac{1}{2},-\frac{1}{2};\infty \right),
S_{1,1,1}\left(-\frac{1}{2},\frac{1}{2},\frac{1}{2};\infty \right),
S_{1,1,1}\left(-\frac{1}{2},\frac{1}{2},1;\infty \right),
\\
S_{1,1,1}\left(\frac{1}{2},-\frac{1}{2},-\frac{1}{2};\infty \right),
S_{1,1,1}\left(\frac{1}{2},-\frac{1}{2},\frac{1}{2};\infty \right),
S_{1,1,1}\left(\frac{1}{2},-\frac{1}{2},1;\infty \right)$\\
\end{description}

\subsection{Alphabet $\{1,-1,\frac{1}{2},-\frac{1}{2},2,-2\}:$}

\noindent
\begin{description}
 \item[w=1:]\  \\
	      Basis constant: $S_{-1}(\infty )=-\log(2),S_1\left(-\frac{1}{2};\infty \right)$\\

 \item[w=2:]\  \\
	      Basis constants: $ \displaystyle S_{2}(\infty
)=\zeta_2,S_{1,1}\left(\frac{1}{2},\frac{1}{2};\infty \right),S_{1,1}\left(-\frac{1}{2},\frac{1}{2};\infty \right),S_{1,1}\left(\frac{1}{2},-\frac{1}{2};\infty \right)$\\
 \item[w=3:]\  \\
	      Basis constants: $\displaystyle S_3(\infty
)=\zeta_3,S_3\left(-\frac{1}{2};\infty
\right),S_{1,2}\left(-\frac{1}{2},-2;\infty \right),S_{1,2}\left(-\frac{1}{2},\frac{1}{2};\infty \right),
\\
S_{1,2}\left(\frac{1}{2},-\frac{1}{2};\infty
   \right),S_{2,1}\left(-\frac{1}{2},\frac{1}{2};\infty \right),S_{1,1,1}\left(-1,-\frac{1}{2},-\frac{1}{2};\infty \right),S_{1,1,1}\left(-1,-\frac{1}{2},\frac{1}{2};\infty \right),
\\
S_{1,1,1}\left(-1,\frac{1}{2},-\frac{1}{2};\infty
   \right),S_{1,1,1}\left(-\frac{1}{2},-1,-\frac{1}{2};\infty \right),S_{1,1,1}\left(-\frac{1}{2},-\frac{1}{2},\frac{1}{2};\infty \right),
\\
S_{1,1,1}\left(-\frac{1}{2},\frac{1}{2},-\frac{1}{2};\infty \right),
S_{1,1,1}\left(-\frac{1}{2},\frac{1}{2},\frac{1}{2};\infty \right),
S_{1,1,1}\left(-\frac{1}{2},\frac{1}{2},1;\infty \right),
\\
S_{1,1,1}\left(\frac{1}{2},-\frac{1}{2},-\frac{1}{2};\infty \right),
S_{1,1,1}\left(\frac{1}{2},-\frac{1}{2},\frac{1}{2};\infty \right),
S_{1,1,1}\left(\frac{1}{2},-\frac{1}{2},1;\infty \right)$.\\
\end{description}

\section{Asymptotic Expansion of $S$-Sums}\label{Sec:AsymptoticExp}

\vspace*{1mm}
\noindent
In this Section we seek expansions\footnote{For further computer algebra aspects concerning the calculation of expansions we refer, e.g., to~\cite{WZ:85,SS:99,H:06,FS:09} and references therein.
} of $S$-sums generalizing ideas for harmonic sums given in~\cite{Blumlein:2009fz}. We say that a function $f:\R \rightarrow \R$ is expanded in an asymptotic series
\cite{Wittaker}
$$
f(x) \sim \sum_{n=1}^{\infty}{\frac{a_n}{x^n}}, \ x \rightarrow \infty,
$$
where $a_n$ are constants from $\R$, if for all $N\geq 0$
$$
R_N(x)=f(x)-\sum_{n=0}^N{\frac{a_n}{x^n}}\in o\left(\frac{1}{x^N}\right), \ x \rightarrow \infty;
$$
note that for a function $g(x)$ we use the notion $g(x)\in o(G(x))\Leftrightarrow \lim\limits_{x\to \infty} \left |\frac{g(x)}{G(x)}\right| =0$. When only the first coefficients of such an expansion are given (calculated), we often write $f(x)\sim a_0+a_1x+a_2x^2+\dots+a_n x^n$ and cut off the remainder term.

We remark that the subsequent Sections provide algorithms to expand $S$-sums which come
close to the class of $\bar{S}$-sums (for a definition see~\eqref{Equ:SBar}). For the
general class of $S$-sums the expansion problem is not handled completely. However, for
current calculations in QCD (see, e.g., Section~\ref{Sec:QCDExample}) the presented
methods are completely sufficient.

\subsection{Asymptotic Expansions of Harmonic Sums $\S{1}{c;n}$ with $c\geq 1$}
\label{SSexpandS1b}

\noindent
\textbf{Version 1:} The asymptotic representation of the harmonic numbers $\S{1}{n}$
is given by \cite{Nielsen1906b}
\begin{equation}\label{Equ:HarmonicNumberExp}
S_1(n)\sim
\gamma+\log(n)+\frac{1}{2\,n}-\sum_{k=1}^{\infty}{\frac{B_{2k}}{2\,k\,n^{2k}}}~,
\end{equation}
where $\gamma$ is the Euler-Mascheroni constant and
$B_n$ are the Bernoulli numbers \cite{Bernoulli1713,Saalschuetz1893,Nielsen1923},
$$
B_n=\sum_{k=0}^n{\binom{n}{k}} B_k, \ \ \ B_0=1.
$$
To compute the asymptotic expansion of $\S{1}{c;n}$ with $c>1$ one may represent it using the Euler-Maclaurin formula \cite{Euler1738,MacLaurin1742}:
\begin{eqnarray*}
\SS1cn&=& \bar{S}_{1}(c;n)+\S1n
\end{eqnarray*}
where
\begin{eqnarray*}
\bar{S}_{1}(c;n)&:=&\sum_{i=1}^n\frac{c^i-1}{i}\\ &=&-\tfrac{1}{2}\left[\log(c)+\frac{c^{n+1}-1}{n+1}\right]-\gamma-\log(\log(c))-\log(n+1)\\
		    & &+\textnormal{Ei}(\log(c)(n+1))+\sum_{j=1}^m{\frac{B_{2j}}{(2j)!}\left.\left[\frac{d^{2j-1}}{di^{2j-1}}\left(\frac{c^i-1}{i}\right)\right]\right|_0^{n+1}}
+R_k(c,n)~.
\end{eqnarray*}
Here $\textnormal{Ei}(z)$ is the exponential integral \cite{Euler1768,Nielsen1906a},
$$
\int_0^x\frac{c^t-1}{t}dt=\textnormal{Ei}(x\log(c))-\gamma-\log(x)-\log(\log(c)),
$$
and $R_k(c,n)$ stand for the rest term. The derivatives are given in closed form by
\begin{eqnarray*}
\frac{d^{j}} {di^{j}}\left(\frac{c^i-1}{i}\right)&=&\frac{c^i}{i}\sum_{k=0}^j\frac{\log^{j-k}(c)}{i^k}{(-1)}^k\frac{\Gamma(j+1)}{\Gamma(j+1-k)}-{(-1)}^j\frac{\Gamma(j+1)}{i^{j+1}}\\
\lim_{i\rightarrow0}\frac{d^{j}}{di^{j}}\left(\frac{c^i-1}{i}\right)&=&\frac{\log^{j+1}(c)}{j+1}.
\end{eqnarray*}
The asymptotic representation for $\S{1}{c;n}$ with $c>1$ can now be derived using
$$
\textnormal{Ei}(t) \sim \exp(t) \sum_{k=1}^{\infty}\frac{\Gamma(k)}{t^k}.
$$

\noindent\textbf{Version 2:} A more suitable approach to compute the expansion of $\S{1}{c;n}$ with $c>1$, and which is implemented in the \texttt{HarmonicSums} package, is as follows. Differentiating the integral representation
$$
\frac{\S{1}{c;n}}{c^n}=\sum_{i=1}^n\frac{\frac{1}{c^{n-i}}}{i}=\int_0^1{\frac{x^n-\frac{1}{c^n}}{x-\frac{1}{c}}}dx
$$
with respect to $n$ we get
\begin{eqnarray*}
 \frac{\partial}{\partial n}\frac{\S{1}{c;n}}{c^n}&=&\int_0^1\frac{x^n\H0x+\frac{1}{c^n}\H0c}{x-\frac{1}{c}}dx\\
    &=&\int_0^1\frac{x^n(\H0x+\H0c)}{x-\frac{1}{c}}dx-\H0c \int_0^1\frac{x^n-\frac{1}{c^n}}{x-\frac{1}{c}}dx.
\end{eqnarray*}
Hence
\begin{eqnarray}\label{expdiffrel}
\int_0^1\frac{x^n(\H0x+\H0c)}{x-\frac{1}{c}}dx&=&\frac{\partial}{\partial n}\frac{\S{1}{c;n}}{c^n}+\H0c\frac{\S{1}{c;n}}{c^n}.
\end{eqnarray}
We can now compute the asymptotic expansion of the integral using integration by parts. Let
\begin{eqnarray*}
\int_0^1\frac{x^n(\H0x+\H0c)}{x-\frac{1}{c}}dx&\sim&\sum_{i=0}^\infty\frac{a_i}{n^i}
\end{eqnarray*}
with $a_i\in \R$ be this expansion and let $\sum_{i=0}^\infty\frac{b_i}{n^i}$
with $b_i\in \R$ be the asymptotic expansion of $\frac{\S{1}{c;n}}{c^n}.$ We note that the differentiation of this expansion yields $\sum_{i=2}^\infty\frac{-(i-1)b_{i-1}}{n^i}$ \ie
\begin{eqnarray*}
\frac{\partial}{\partial n}\frac{\S{1}{c;n}}{c^n}&\sim&\sum_{i=2}^\infty\frac{-(i-1)b_{i-1}}{n^i}.
\end{eqnarray*}
Plugging these expansions into equation (\ref{expdiffrel}) and comparing coefficients yields the recurrence relation
\begin{eqnarray*}
 b_n&=&\frac{a_n}{\H0c}+\frac{(n-1)}{\H0c}b_{n-1} \textnormal{ for } n\geq2
\end{eqnarray*}
with the initial values  $b_0=\frac{a_0}{\H0c}$ and $b_1=\frac{a_1}{\H0c}$.
Hence we get for example
\begin{eqnarray*}
\S{1}{2;n}&\sim&2^n\left(\frac{14174522}{n^{10}}+\frac{1091670}{n^9}+\frac{94586}{n^8}+\frac{9366}{n^7}+\frac{1082}{n^6}+\frac{150}{n^5}+\frac{26}{n^4}+\frac{6}{n^3}+\frac{2}{n^2}+\frac{2}{n}\right)\\
\S{1}{3;n}&\sim&3^n \left(\frac{566733}{4 n^{10}}+\frac{17295}{n^9}+\frac{38001}{16 n^8}+\frac{1491}{4 n^7}+\frac{273}{4 n^6}+\frac{15}{n^5}+\frac{33}{8 n^4}+\frac{3}{2 n^3}+\frac{3}{4 n^2}+\frac{3}{2 n}\right).
\end{eqnarray*}

%

\subsection{Computation of Asymptotic Expansions by repeated partial Integration}
\label{Sec:PartialInt}

\noindent
Next, we deal with generalized polylogarithms of the form
\begin{equation}\label{Equ:MellinSpecialCase}
\M{\frac{\H{m_1,m_2,\ldots,m_k}x}{c \pm x}}{n} \ \textnormal{ and } \M{\frac{\H{b_1,b_2,\ldots,b_k}{1-x}}{c \pm x}}{n}
\end{equation}
with $\abs{c}\geq1$ where $m_i~\in~\R~\setminus(0,1]$ or $b_i~\in~\R~\setminus~(-1,0)~\cup~(0,1))$ with $b_l\neq 0$ using repeated integration by parts.
The first two lemmas show that the arguments of the Mellin transform under consideration are analytic. Thus iterated integration by parts is possible for $c>1$ using the already known integral representation. The third lemma gives an integral representation for the special case $c=1$ such that the method of repeated partial integration is also possible. The proofs are omitted here.

\begin{lemma}
Let $\H{m_1,m_2,\ldots,m_k}x$ be a generalized polylogarithm with $m_i~\in~\R~\setminus(0,1]$ for $1~\leq~i~\leq k.$ Then
$$\H{m_1,m_2,\ldots,m_k}x,\quad \frac{\H{m_1,m_2,\ldots,m_k}x}{c+x}\quad\text{ and }\quad
\frac{\H{m_1,m_2,\ldots,m_k}x-\H{m_1,m_2,\ldots,m_k}1}{c-x}$$
are analytic for $\abs{c}\geq1$ and $x \in (0,1].$
\label{SSanalytic1}
\end{lemma}

\begin{lemma}
Let $\H{m_1,m_2,\ldots,m_k}x$ be a generalized polylogarithm with $m_i~\in~\R~\setminus~((-1,0)~\cup~(0,1))$ for $1~\leq~i~\leq~k,$ and $m_k~\neq~0.$ Then
$$\frac{\H{m_1,m_2,\ldots,m_k}{1-x}}{c+x}\quad\text{ and }
\quad\frac{\H{m_1,m_2,\ldots,m_k}{1-x}}{c-x}$$
are analytic for $\abs{c}\geq1$ and $x \in (0,1].$
\label{SSanalytic2}
\end{lemma}

\noindent For $c>1$, we obtain a suitable integral representation for partial integration. E.g., we get

\small
\begin{align*} \M{\frac{\H{-2}x}{x+2}}n=&\int_0^1x^n\frac{\H{-2}x}{x+2}dx=\int_0^1x^{n-1}\frac{x\cdot\H{-2}x}{x+2}dx\\
=&\frac{x^n}{n}\frac{x\cdot\H{-2}x}{x+2}\Biggr|_0^1-\int_0^1\frac{x^n}{n}\frac{2 \H{-2}x+x}{(x+2)^2} dx\\
=&\H{-2}1\frac{3}{n}-\frac{1}{n}\int_0^1x^{n-1}\frac{x(2 \H{-2}x+x)}{(x+2)^2} dx\\
=&\H{-2}1\frac{3}{n}-\frac{1}{n}\left(\frac{x^n}{n}\frac{x(2 \H{-2}x+x)}{(x+2)^2}\Biggr|_0^1-\int_0^1\frac{x^n}{n}\frac{6 x-2 (x-2) \H{-2}x}{(x+2)^3} dx\right)\\  =&\H{-2}1\left(\frac{3}{n}-\frac{2}{9n^2}\right)-\frac{1}{9n^2}+\frac{1}{n^2}\int_0^1x^{n-1}\frac{x(6 x-2 (x-2) \H{-2}x)}{(x+2)^3} dx\\
=&\cdots=\H{-2}1\left(\frac{3}{n}-\frac{2}{9n^2}+\frac{2}{27 n^3}+\frac{2}{27 n^4}\right)-\frac{1}{9 n^2}+\frac{2}{9 n^3}-\frac{20}{81 n^4}+O\left(\frac{1}{n^5}\right).
\end{align*}
\normalsize

\noindent For the special case $c=1$ one can use the integral representation of the following lemma to calculate the expansion with partial integration as demonstrated above.

\begin{lemma}
Let $\H{\ve m}x=\H{m_1,m_2,\ldots,m_k}x$ and $\H{\ve b}x=\H{b_1,b_2,\ldots,b_l}x$ be a generalized polylogarithm with
$m_i~\in~\R~\setminus(0,1]$ for $1~\leq~i~\leq~k$ and $b_i~\in~\R~\setminus~((-1,0)~\cup~(0,1))$ for $1~\leq~i~\leq~l$ where $b_l\neq 0.$ Then we have
$$
\M{\frac{\H{\ve m}x}{1-x}}{n} =\int_0^1{\frac{x^n(\H{\ve m}x-\H{\ve m}1)}{1-x}dx}-\int_0^1{\frac{\H{\ve m}x-\H{\ve m}1}{1-x}}dx-\S{1}n\H{\ve m}{1},
$$
and
$$
\M{\frac{\H{\ve b}{1-x}}{1-x}}{n}=\int_0^1{\frac{x^n\H{\ve b}{1-x}}{1-x}dx}-\H{0,\ve b}1
$$
where
$$\int_0^1{\frac{\H{\ve m}x-\H{\ve m}1}{1-x}}dx, \ \H{\ve m}{1} \textnormal{ and } \H{0,\ve b}1 $$
are finite constants.
\label{SSexpansioncorrection1}
\end{lemma}

\noindent Summarizing, we are able to calculate the asymptotic expansion of~\eqref{Equ:MellinSpecialCase}
using the method of repeated integration by parts.

\subsection{Computation of Asymptotic Expansions of $S$-Sums}
\label{SSExpansion}

\noindent
Subsequently, we derive an algorithm that calculates the asymptotic expansion of
$S$-sums $\S{a_1,\ldots,a_k}{b_1,\ldots,b_k;n}$ with $b_i\in[-1,1]$ and
$b_i\neq0$ using the technologies of the previous Subsections.

Since these sums fall into the class of $\bar{S}$-sums, it follows by Section~\ref{SSInvMellin} that they can be represented using Mellin transforms of the form
$\M{\frac{\H{m_1,m_2\ldots,m_k}x}{c\pm x}}{n}$
with $m_i\in \R \setminus ((-1,0)\cup(0,1))$ and $\abs{c}\geq1$. Conversely,
given such a Mellin transform, it can be expressed in terms of $S$-sums of the above type with $b_i\in[-1,1]$ and $b_i\neq0$.

In addition, it is straightforward to see that this subclass is closed under the quasi
shuffle product~\eqref{SSsumproduct}, i.e., the product can be expressed again as a
linear combination of $S$-sums of this subclass over the rational numbers.

We are now ready to present an algorithm which
computes asymptotic expansions of $S$-sums
$\S{a_1,a_2,\ldots,a_k}{b_1,b_2,\ldots,b_k;n}$ with $b_i\in [-1,1]$ and $b_i\neq 0$.
The algorithm tries to attack the sums directly using the previous Sections. If it fails one uses the transformation $x\to 1-x$ from Section~\ref{SSbxx} twice and reduces the problems to $S$-sums that are less complicated than the input sum. Recursive application will finally produce the expansion of the original $S$-sum. To be more precise, the algorithm can be described as follows.

\begin{itemize}
	\item If  $\S{a_1,a_2,\ldots,a_k}{b_1,b_2,\ldots,b_k;n}$ has trailing ones, \ie $a_k=b_k=1$, we first extract them such that we end up in a univariate
	polynomial in $\S1n$ with coefficients in the $S$-sums without trailing ones.
   Expand the powers of $\S{1}n$ using~\eqref{Equ:HarmonicNumberExp}.
  Now apply the
	following items to each of the $S$-sums without trailing ones:
	\item Let $\S{a_1,a_2,\ldots,a_k}{b_1,b_2,\ldots,b_k;n}$ with $a_k\neq 1\neq b_k$, and compute $\frac{\H{m_1,m_2,\ldots,m_l}x}{c+sx}$ such that
	in the inverse Mellin transform the most complicated $S$-sum is $\S{a_1,a_2,\ldots,a_k}{b_1,b_2,\ldots,b_k;n}$, i.e., express
	$\S{a_1,a_2,\ldots,a_k}{b_1,b_2,\ldots,b_k;n}$ as
	\begin{equation}\label{SSasyalg1}		 \S{a_1,a_2,\ldots,a_k}{b_1,b_2,\ldots,b_k;n}=\M{\frac{\H{m_1,m_2,\ldots,m_l}x}{c+sx}}{n}+T
	\end{equation}
where $T$ is an expression in $S$-sums
$\S{a'_1,a'_2,\ldots,a'_{k'}}{b'_1,b'_2,\ldots,b'_{k'};n}$ with $b'_i\in [-1,1]$ and
$b'_i\neq 0$ (which are less complicated than $\S{a_1,a_2,\ldots,a_k}{b_1,b_2,\ldots,b_k;n}$) and constants. Note that $s=\pm1$, $c\geq1$ and $m_i\in \R \setminus ((-1,0)\cup(0,1))$ and that $m_l\neq1$ (see the beginning of this Subsection and Lemma~\ref{Lemma:MostComplicated}).
\item We proceed by expanding $\M{\frac{\H{m_1,m_2,\ldots,m_l}x}{c+sx}}{n}$ in (\ref{SSasyalg1}):
	\begin{description}
		\item[all $m_i\neq 1$:] Expand $\M{\frac{\H{m_1,m_2,\ldots,m_l}x}{c+sx}}{n}$ directly; see Section~\ref{Sec:PartialInt}.
		\item[not all $m_i\neq 1$:]
		\begin{itemize}\item[]		

			\item Transform $x\rightarrow 1-x$ in $\H{\ve m}x$ as described in
Section~\ref{SSbxx}
								 and expand all products, i.e., we obtain
				\begin{equation}\label{SSasyalg2}
					\M{\frac{\H{\ve m}x}{c+sx}}{n}=\sum_{i=1}^pd_i\M{\frac{\H{\ve b_i}{1-x}}{c+sx}}{n}+d \ \textnormal{  with } d,d_i\in\R.
				\end{equation}
				Note that due to Remark~\ref{SSbxxRemark} the components of the $\ve b_i$ are in $\set R\setminus (0,1).$
			 \item For each Mellin transform $\M{\frac{\H{b_1,\ldots,b_j}{1-x}}{c+sx}}{n}$ do
				\begin{description}\item[]
					\item[$b_j\neq 0:$] Expand $\int_0^1{\frac{x^n\Hma{\ve b}{1-x}}{c+sx}}$; see Section~\ref{Sec:PartialInt}.

					\item[$b_j=0:$] As described in Section~\ref{SSbxx}					
					transform back $1-x\rightarrow x$ in $\H{\ve b}{1-x}$ and expand all products, i.e., we obtain
						\begin{equation}\label{SSasyalg3}
							\M{\frac{\H{\ve b}{1-x}}{c+sx}}{n}=\sum_{i=1}^pe_i\M{\frac{\H{\ve g_i}{x}}{c+sx}}{n}+e
						\end{equation}
						with $e,e_i\in\R$. Finally, perform the Mellin transforms $\M{\frac{\Hma{\ve g_i}{x}}{c+sx}}{n}$, i.e., express it in terms of $S$-sums and constants in terms of infinite $S$-sums. Note that these harmonic sums are less complicated (see below).
				\end{description}
		\end{itemize}
	\end{description}
	\item Replace $\M{\frac{\Hma{m_1,m_2,\ldots,m_l}x}{c+sx}}{n}$ in (\ref{SSasyalg1}) by the result of this process.
	\item For all harmonic sums that remain in (\ref{SSasyalg1}) apply the above points. Since these harmonic sums are less complicated (see below) this process will terminate.
\end{itemize}

Concerning termination the following remarks are in place. Since $a_k\neq 1$ in  (\ref{SSasyalg1}), we know due to Lemma~\ref{Lemma:MostComplicated} that $m_l\neq 1$ in  (\ref{SSasyalg1}).
If not all $m_i\neq 1$ in  (\ref{SSasyalg1}), we have to transform $x\rightarrow 1-x$ in $\H{m_1,m_2,\ldots,m_l}x$ as described in Section~\ref{SSbxx}.
Due to Remark~\ref{SSbxxRemark} we see that the single
generalized polylogarithm at argument $x$ with weight $l,$ which will emerge, will not
have
trailing zeroes since $m_l\neq 1$. Hence we can compute the asymptotic expansion of the Mellin transform containing this generalized polylogarithm using repeated integration by parts. We might be able to compute the asymptotic
expansion of some (or all) of the other Mellin transforms by repeated integration by
parts as well, but if we fail we know at least that due to
Remark~\ref{Lemma:MostComplicated} the $S$-sums which will appear
in~\eqref{SSasyalg3} are less complicated (they are even of lower weight) than
$\S{a_1,a_2,\ldots,a_k}{b_1,b_2,\ldots,b_k;n}$ of (\ref{SSasyalg1}) and
hence this algorithm will eventually terminate.

Applying the described algorithm to $\textnormal{S}_{2,1}\hspace{-0.2em}\left(\tfrac{1}{3},\tfrac{1}{2};n\right)$ produces the following asymptotic expansion:
\small
\begin{align*}
\textnormal{S}_{2,1}\hspace{-0.2em}\left(\tfrac{1}{3},\tfrac{1}{2};n\right)\sim&
-\textnormal{S}_{1,2}\hspace{-0.2em}\left(\tfrac{1}{6},3;\infty \right)
  +\textnormal{S}_1\hspace{-0.2em}
\left(\tfrac{1}{2};\infty \right)
\left[-\textnormal{S}_2\hspace{-0.2em}\left(\tfrac{1}{6};\infty \right)
+\textnormal{S}_2\hspace{-0.2em}\left(\tfrac{1}{3};\infty \right)+\tfrac{1}{6^n}
\left(-\frac{1464}{625n^5}\right.\right.\\
&\left.\left.+\frac{126}{125 n^4}-\frac{12}{25 n^3}+\frac{1}{5n^2}
+33\frac{ 2^{n-1}}{n^5}-9\frac{2^{n-1}}{n^4}+3\frac{2^{n-1}}{n^3}
-\frac{2^{n-1}}{n^2}\right)\right]\\
&+\frac{1}{6^n} \left(-\frac{4074}{3125 n^5}+\frac{366}{625 n^4}-\frac{42}{125 n^3}+\frac{6}{25 n^2}-\frac{1}{5 n}\right)
   \textnormal{S}_2\hspace{-0.2em}\left(\tfrac{1}{2};\infty \right)\\
&+\textnormal{S}_1\hspace{-0.2em}\left(\tfrac{1}{6};\infty \right) \textnormal{S}_2\hspace{-0.2em}\left(\tfrac{1}{2};\infty \right)+\textnormal{S}_3\hspace{-0.2em}\left(\tfrac{1}{2};\infty
   \right)+\frac{1}{6^n} \left(\frac{642}{125 n^5}-\frac{28}{25 n^4}+\frac{1}{5 n^3}\right)\\
&+\frac{1}{6^n} \left(\frac{2037}{3125 n^5}-\frac{183}{625
   n^4}+\frac{21}{125 n^3}-\frac{3}{25 n^2}+\frac{1}{10 n}\right) \zeta_2+\frac{1}{6^n} \left(-\frac{2037}{3125 n^5}\right.\\
&\left.+\frac{183}{625
   n^4}-\frac{21}{125 n^3}+\frac{3}{25 n^2}-\frac{1}{10 n}\right) \log ^2(2)+\frac{1}{6^n} \left(\frac{1464}{625 n^5}-\frac{126}{125 n^4}+\frac{12}{25
   n^3}\right.\\
&\left.-\frac{1}{5 n^2}\right) \log (2),
\end{align*}
\normalsize

with $\text{S}_m(\tfrac{1}{k},\infty) = \text{Li}_m\left(\tfrac{1}{k}\right),~~m,k
\in \mathbb{N}_+, m \geq 2$.

\subsection{Further extensions}

\noindent
The following lemma  extends the presented algorithm to $\bar{S}$-sums $\S{a_1,\ldots,a_k}{b_1,\ldots,b_k;n}$ provided that if $\abs{b_1}>1,$ we have that $\abs{b_i}\leq 1$  and if $b_i=1$ then $a_i>1$ for $2\leq i \leq k$. In this case, we obtain the required form given in~\eqref{SSasyalg1} except that $0<c<1$ and that the occurring $\bar{S}$-sums in $T$ satisfy again the constrains from above. Then the following lemma yields a suitable integral representation of $\M{\frac{\H{m_1,m_2,\ldots,m_l}x}{c+sx}}{n}$ and repeated integration by parts produces its asymptotic expansion. In addition the occurring $\bar{S}$-sums in $T$ can be handled by recursion of the extended method.

\begin{lemma}\label{Lemma:ExtExpansion}
Let $\H{m_1,m_2,\ldots,m_k}x$ be a generalized polylogarithm with $m_i~\in~\R~\setminus~(-1, 1]~\cup~\{0\}$ for $1\leq i\leq k$ and $0<c<1.$ Then
\begin{eqnarray*}
&&\hspace{-1.5cm}\M{\frac{\H{m_1,m_2,\ldots,m_k}x}{c-x}}{n}=\\		 &&\frac{1}{c^n}\int_0^1{\frac{x^n\left(\H{m_1,m_2,\ldots,m_k}x-\H{m_1,m_2,\ldots,m_k}c\right)}{c-x}dx}\\
					&& + \H{m_1,c,m_2,\ldots,m_k}c+\H{m_1,m_2,c,\ldots,m_k}c+\cdots+\H{m_1,m_2,\ldots,m_k,c}c\\
					&& + \H{m_2,\ldots,m_k}c\H{0,m_1-c}{1-c} + \H{m_3,\ldots,m_k}c\H{0,m_1-c,m_2-c}{1-c}\\
					&& +\cdots+ \H{m_k}c\H{0,m_1-c,\ldots,m_{k-1}-c}{1-c} + \H{0,m_1-c,\ldots,m_k-c}{1-c} \\
					&&- S_{1,\{\frac{1}{c}\}}(n)\H{m_1,m_2,\ldots,m_k}c.
\end{eqnarray*}
\label{SSexpansioncorrection2}
\end{lemma}
\begin{proof}We have
\begin{eqnarray*}
\M{\frac{\H{m_1,m_2,\ldots,m_k}x}{c-x}}{n}&=&\int_0^1{\frac{((\frac{x}{c})^n-1)\H{m_1,m_2,\ldots,m_k}x}{c-x}dx}=\\
			 &=&\int_0^1{\frac{(\frac{x}{c})^n(\H{m_1,m_2,\ldots,m_k}x-\H{m_1,m_2,\ldots,m_k}c)}{c-x}dx}\\
			 &&+\int_0^1{\frac{(\frac{x}{c})^n\H{m_1,m_2,\ldots,m_k}c-\H{m_1,m_2,\ldots,m_k}x}{c-x}dx}\\
			 &=&\int_0^1{\frac{(\frac{x}{c})^n(\H{m_1,m_2,\ldots,m_k}x-\H{m_1,m_2,\ldots,m_k}c)}{c-x}dx}\\	 &&+\underbrace{\H{m_1,m_2,\ldots,m_k}c\int_0^1{\frac{(\frac{x}{c})^n-1}{c-x}dx}}_{A:=}\\	 &&-\underbrace{\int_0^1{\frac{\H{m_1,m_2,\ldots,m_k}x-\H{m_1,m_2,\ldots,m_k}c}{c-x}dx}}_{B:=}.
\end{eqnarray*}
with
\begin{eqnarray*}
A&=&- \S{1}{\frac{1}{c};n}\H{m_1,m_2,\ldots,m_k}c\\
B&=&\underbrace{\int_0^c{\frac{\H{m_1,m_2,\ldots,m_k}x-\H{m_1,m_2,\ldots,m_k}c}{c-x}dx}}_{B_1:=}\\	 &&\hspace{2.5cm}+\underbrace{\int_c^1{\frac{\H{m_1,m_2,\ldots,m_k}x-\H{m_1,m_2,\ldots,m_k}c}{c-x}dx}}_{B_2:=}\\
\end{eqnarray*}
where
\begin{eqnarray*}
B_1&=&-\H{m_1,c,m_2,\ldots,m_k}c-\cdots-\H{m_1,m_2,\ldots,m_k,c}c\\
B_2&=&\int_0^{1-c}{\frac{\H{m_1,m_2,\ldots,m_k}{x+c}-\H{m_1,m_2,\ldots,m_k}c}{-x}dx}\\
			&=&-\int_0^{1-c}{\frac{\H{m_2,\ldots,m_k}c\H{m_1-c}{x} +\cdots + \H{m_1-c,\ldots,m_k-c}{x}}{x}dx}\\
			&=&-\H{m_2,\ldots,m_k}c\H{0,m_1-c}{1-c} - \H{m_3,\ldots,m_k}c\H{0,m_1-c,m_2-c}{1-c}\\
			&&-\cdots- \H{m_k}c\H{0,m_1-c,\ldots,m_{k-1}-c}{1-c} - \H{0,m_1-c,\ldots,m_k-c}{1-c};
\end{eqnarray*}
thus the lemma follows.
\end{proof}

\noindent A typical example for this extension is, e.g.,
\begin{align*}
\S{1,1}{3,\tfrac{1}{3};n}\sim&\frac{1331}{240 n^5}-\frac{63}{32 n^4}+\frac{23}{24 n^3}-\frac{5}{8 n^2}+\frac{1}{2 n}+\textnormal{S}_1\left(\tfrac{1}{3};\infty \right) \textnormal{S}_1(3;n)\\
&-\frac{1}{2} \textnormal{S}_1\left(\tfrac{1}{3};\infty \right){}^2-
\textnormal{S}_2\left(\tfrac{1}{3};\infty \right),
\end{align*}
with $\text{S}_1(\tfrac{1}{m},\infty) = \ln(m) - \ln(m-1), m \in \mathbb{N}\backslash
\{1\}$.

More generally, this lemma enables one to deal also with $S$-sums which are outside of the $\bar{S}$-sum case. To be more precise, our strategy works if the given $S$-sum $\S{a_1,\ldots,a_k}{b_1,\ldots,b_k;n}$ can be written in the form~\eqref{SSasyalg1} where $\M{\frac{\H{m_1,m_2,\ldots,m_l}x}{c+sx}}{n}$ can be expanded directly by Lemma~\ref{Lemma:ExtExpansion}. In addition we need that $T$ consists of $S$-sums which are less complicated than $\S{a_1,\ldots,a_k}{b_1,\ldots,b_k;n}$ and which can be handled by iterative application of our expansion strategy (and using Lemma~\ref{Lemma:ExtExpansion} if needed).

\noindent
E.g., our expansion technique produces
\begin{align*}
\textnormal{S}&_{1,1,1}\left(-\tfrac{1}{2},2,\tfrac{1}{4};n \right)
\sim
\textnormal{S}_{1,1,1}\left(-\tfrac{1}{2},2,\tfrac{1}{4};\infty \right)
+\left(-\tfrac{1}{27 n^3}+\tfrac{1}{9 n^2}-\tfrac{1}{6 n}\right)
\left(-\tfrac{1}{2}\right)^n \textnormal{S}_2\left(\tfrac{1}{4};\infty \right)\\
   &+\left(-\frac{1}{2}\right)^n
\left[\left(\frac{2}{27 n^3}-\frac{2}{9
n^2}+\frac{1}{3 n}\right) \log ^2(2)+\left(-\frac{2}{27 n^3}+\frac{2}{9 n^2}
-\frac{1}{3 n}\right) \log (3)\right] \textnormal{S}_1(2;n)\\
   &+\frac{\left(-\frac{1}{4}\right)^n}{15 n^3}+\left(-\frac{2}{27 n^3}
+\frac{2}{9 n^2}-\frac{1}{3 n}\right) \left(-\frac{1}{2}\right)^n \log ^3(2)\\
   &+\left[(-1)^n \left(\frac{1}{3 n^2}-\frac{2}{9
n^3}\right)+\left(-\frac{1}{2}\right)^n \left(\frac{1}{27 n^3}-\frac{1}{9 n^2}\right)
+\frac{(-1)^n \left(\frac{1}{2}\right)^n}{6 n}\right] \log ^2(2)\\
   &+\left(\frac{2}{27 n^3}-\frac{2}{9 n^2}+\frac{1}{3 n}\right) \left(-\frac{1}{2}\right)^n \log (2) \log (3)+(-1)^n \left(\frac{2}{9 n^3}-\frac{1}{3 n^2}\right) \log (3).
\end{align*}

\section{Application from Quantum Chromodynamics}\label{Sec:QCDExample}

\vspace*{1mm}
\noindent
In the following we concentrate on the discussion of the emergence of generalized
harmonic sums in calculations of massless and massive Wilson coefficients in
deep-inelastic scattering in QCD.
Up to the 2--loop level neither for the anomalous
dimensions nor for the Wilson coefficients generalized harmonic sums occur,
cf.~\cite{Moch:1999eb,Blumlein:2005im,Blumlein:2006rr,Blumlein:2007dj}.
They are also absent in case of the 3--loop anomalous dimensions
\cite{Moch:2004pa,Vogt:2004mw}. In intermediate results of the calculation of the
massless 3--loop Wilson coefficients they emerge, e.g. for a graph called
${\rm LA}_{27}\text{box}(N)$ in  Ref.~\cite{Vermaseren:2005qc}. The generalized
sums~\eqref{Equ:SSumDef} have numerator--weights over the alphabet
$x_i\in\{1,-1,2,-2\}$ up to weight $w = 6$. Examples are:
\begin{description}
\item[]
$
{\displaystyle
\text{S}_1(2;n),
\text{S}_{1,1,1,1,2}(2,1,1,1,1;n),
\text{S}_{2,1,3}(2,1,-1;n),
\text{S}_{3,2,1}(-2,-1,1;n),~\text{etc.}
}
$
\end{description}
All of these sums contain besides the numerator weights $i_k \in \{-1,1\}$ only
either
the weight $2$ or $-2$ once. For these cases one always may represent the
corresponding sums such that
\begin{eqnarray}
\text{S}_{n_1,...n_m}(\pm 2,i_1,...,i_{m-1};n) = \int_0^1 dx \frac{(\pm 2x)^n
-1}{\pm 2x -1} \sum_j c_j H_{\vec{a_j}}(x),
\noindent
\end{eqnarray}
where $H_{\vec{a_j}}(x)$ are usual harmonic polylogarithms.
Since the corresponding expressions are weighted by a factor $2^{-n}$,
the individual expressions are transforms by the non-singular kernels $((\pm 1)^n
x^n-1/2^n)/(\pm 2x - 1)$. It even turns out that all generalized sums cancel in
the final result, \cite{Vermaseren:2005qc}. In this
way, the massless Wilson coefficients up to 3--loop order do not depend on
generalized harmonic sums.

In Ref.~\cite{Ablinger:2010ty} the $O(n_f T_F^2 C_{A,F})$ contributions
to the 3-loop massive Wilson coefficient contributing to the structure function
$F_2(x,Q^2)$ at high virtualities $Q^2 \gg m^2$ were calculated. Here $m$ denotes
the heavy quark mass. In this calculation most of the generated nested sums
were calculated directly using the package {\tt Sigma}~
\cite{Schneider:2005,Schneider:2005b,Schneider:2005c,Schneider:2007,Schneider:2008}.
In this way individual Feynman diagrams were split into many terms being summed
individually. Here the contributing alphabet of numerator weights was
$x_i\in\{1,-1,\tfrac{1}{2}, -\tfrac{1}{2}, 2, -2\}$. Examples of generalized sums are~:
\begin{description}
\item[]
$
{\displaystyle
\text{S}_1\left(\tfrac{1}{2};n\right),
\text{S}_2(-2;n),
\text{S}_{2,1}(-1, 2;n),
\text{S}_{3,1}\left(-2,-\tfrac{1}{2};n\right),
\text{S}_{1,1,1,2}\left(-1, \tfrac{1}{2}, 2,-1;n\right),
\text{S}_{2,3}\left(-2,-\tfrac{1}{2};n\right),}\\
{\displaystyle
\hspace*{-8mm}
\text{S}_{2,2,1}\left(-1,-\tfrac{1}{2}, 2;n\right),
\text{S}_{1,1,1,1,1}\left(2,\tfrac{1}{2},1,1,1;n \right), \text{etc.}
}
$
\end{description}
The highest contributing weight was {\it w = 5}. In the final result all generalized
harmonic sums canceled for the $O(n_f T_F^2 C_{A,F})$--terms, even on a
diagram-by-diagram basis.

This brought up the idea to combine all sums contributing to each diagram
leading to C.~Schneider's package {\tt SumProduction} \cite{Schneider:2013}. Through
this a much smaller amount of nested sums, however, with  very large  summands, have
to be calculated, \cite{Blumlein:2012vq,Blumlein:2012hg}. We have
applied this method
for the gluonic $O(n_f T_F^2 C_{A,F})$ contributions in \cite{Ablinger:2012ph}.

In case of the massive 3--loop operator matrix elements generalized harmonic
sums do not vanish always. For
a series of ladder diagrams with six
massive
propagators they contribute for scalar integrals up to weight {\it w = 4},
\cite{Ablinger:2012qm}. This is likewise also the case for a series of Benz-type
diagrams, \cite{Ablinger:2012sm}. Examples are:
\begin{description}
\item[]
\small
$
{\displaystyle
\text{S}_1\left(2;n\right),
\text{S}_{1,2}\left(\tfrac{1}{2},1;n\right),
\text{S}_{1,1,1}\left(\tfrac{1}{2},1,1;n\right),
\text{S}_{1,1,2}\left(\tfrac{1}{2},1,1;n\right),
\text{S}_{1,1,2}\left(2,\tfrac{1}{2},1;n\right)},
\text{S}_{1,1,1,1}\left(2,\tfrac{1}{2},1,1;n\right), \text{etc.}
$
\normalsize
\end{description}
Representing results for individual Feynman diagrams containing also $S$--sums
sometimes involves divergent  $S$--sum expressions in individual terms as $n
\rightarrow \infty$. On the other hand, combinations of different terms
contributing
to the corresponding diagram are always convergent, as also the whole diagram,
see
\cite{Ablinger:2012qm,Ablinger:2012ej,ABSW13}. We therefore limit the
consideration to
convergent combinations of terms containing $S$--sums here. Let us consider an
example
of Ref.~\cite{Ablinger:2012qm}. Integral $\hat{I}_4(n)$, Eq.~(3.18), contains the
following terms growing like  $2^n$ for $n \rightarrow \infty$ ~:
\begin{align*}
T_1(n) &= 2^n \frac{n^3 - 4 n^2 -25 n -28)}{(n+1)^2(n+2)(n+3)^2} \left[
2 S_{1,2}\left(\tfrac{1}{2},1,n\right)
+ S_{1,2}\left(\tfrac{1}{2},1,1;n\right) - 2 \zeta_3 \right]\\
T_2(n) &= -\frac{1}{2} \frac{n+5}{(n+1)(n+3)} \left[
S_{1,1,1,1}\left(2,\tfrac{1}{2},1,1;n\right) + 2
S_{1,1,2}\left(2,\tfrac{1}{2},1;n\right)
- 2 \zeta_3 S_1(2;n)\right].
\end{align*}
One may represent these sums in terms of Mellin transforms,
cf.~Eq.~(3.36)--(3.40),
\cite{Ablinger:2012qm}. There it turns out that each of the $S$--sums contributing to $T_1(n)$ can be expressed in the form
\begin{equation*}
r_{1,i} \zeta_3 + \frac{1}{2^n} F_{1,i}(n),
\end{equation*}
where $r_{1,i}\in\set Q$ and $F_{1,i}(n)$ are convergent integrals for $n \rightarrow \infty$.
Replacing the $S$-sums with these alternative representations in $T_1(n)$ shows that the $\zeta_3$ part vanishes. In addition, rewriting the introduced integral representations back to $S$-sums produces finite $S$-sums. Similarly, each of the $S$-sums in $T_2(n)$ may be expressed in the form
\begin{equation*}
r_{2,i} \zeta_3 S_1(2;n) + F_{2,i}(n),
\end{equation*}
where $r_{2,i}\in\set Q$ and $F_{2,i}(n)$ are convergent integrals. Again, replacing the $S$-sums with these representations in $T_2$ cancels the $S_1(2;n)\zeta_3$ part; in addition, the introduced integrals can be transformed to an expression in terms of absolutely convergent $S$-sums. The remainder
terms of $\hat{I}_4(n)$ are even nested harmonic sums. Due to this only absolutely convergent
$S$-sum contribute to
$\hat{I}_4(n)$ via $T_1(n)$ and $T_1(n)$. A decomposition of this or a similar kind is
expected
in case of physical expressions in general. As the above example shows a
suitable
decomposition of the contributing $S$-sums may be necessary.

In Refs.~\cite{Ablinger:2012ej,ABSW13} we also considered graphs with an internal
4-leg local operator insertion and five massive lines. Here generalized harmonic sums
of up to weight {\it w = 5} occur. Examples are~:
\begin{description}
\item[]
$
{\displaystyle
\text{S}_{1,4}\left(\tfrac{1}{2},-2;n\right),
\text{S}_{1,2,2}\left(\tfrac{1}{2},-2,-1;n\right),
\text{S}_{1,2,2}\left(\tfrac{1}{2},-2,1;n\right),
\text{S}_{1,2,2}\left(\tfrac{1}{2},-2,-1;n\right), \text{etc.}
}
$
\end{description}
For this diagram also generalized harmonic sums nested with binomial-- and
inverse--binomial expressions contribute.

Generalized harmonic sums are a necessary asset both in higher loop and multi-leg
Feynman diagram calculations. They are connected to the generalized harmonic
polylogarithms via a Mellin transform. We have worked out their various relations and
considered special new numbers associated to these quantities in the limit $n
\rightarrow \infty$ of the generalized harmonic sums and special arguments of the
generalized harmonic polylogarithms, generalizing the multiple zeta values.
The relations discussed generate basis representations in the respective algebras.
Note also that special classes of higher transcendental functions are contained in the
generalized harmonic sums for $n \rightarrow \infty$.

In cases that generalized harmonic sums occur in physical results an efficient way
of implementation consists in working in Mellin space, as being the case for the
nested harmonic sums (see also the comments on page~\pageref{RemarksOnQCDAndXSpace}). Having determined the singularity structure of the problem,
and knowing the shift relations $n \rightarrow n+1$ of the corresponding quantities
in analytic form, the knowledge of the asymptotic expansion at $|n| \gg 1,~~n
\in \mathbb{C}$ is sufficient for the analytic continuation of anomalous dimensions
and Wilson coefficients, similar to the case of the nested harmonic sums.
The asymptotic representation is requested to be free of exponential divergences.
Logarithmic terms $\sim \ln^k(n)$ with $k$ related to the number of loops may
contribute, however.
The structure
of the singularities may become more general, however, and has to be determined for
the respective problem under investigation. A single numerical contour integral around
the singularities yields the inverse Mellin-transform from $n$-- to $x$--space.
cf.~\cite{Blumlein:2012bf}.
The techniques outlined above also allow to perform the inverse Mellin-transform
analytically for applications in $x$-space directly. For this purpose the corresponding
basis-functions still need to be represented numerically at sufficient precision,
unlike the case in $n$--space.

\section{Appendix: Available commands of the package {\tt HarmonicSums.m}}\label{App:HarmonicSums}

\vspace*{1mm}
\noindent
In the following we list some of the main commands of the code {\tt HarmonicSums}, a package implemented in Mathematica, and illustrate
their use by examples, with emphasis on the case of $S$--sums.
The {\tt HarmonicSums.m} package and the conditions of
use are found under {\tt  www.risc.jku.at/research/combinat/software/HarmonicSums/}.

\medskip

\noindent
{\tt S[$a_1$,...,$a_m$,$n$]:} defines the harmonic sum $\S{a_1,\ldots,a_m}{n} $ of depth $m$ with upper summation limit $n$ and $a_i\in \N$, while {\tt CS[$\{\{a_1,b_1,c_1\}$,...,$\{a_m,b_m,c_m\}\}$,$\{s_1$,...,$s_m\}$,$n$]} defines cyclotomic harmonic sums $S_{(a_1,b_1,c_1),\dots(a_m,b_m,c_m)}(s_1,\dots,s_m;n)$ with $a_i,c_i\in\set N$, $b_i\in\set N\cup\{0\}$ and $s_i\in\{-1,1\}$; see~\cite[(2.1)]{Ablinger:2011te}.

\medskip

\noindent
{\tt S[$a_1$,...,$a_m$,\{$x_1$,...,$x_m$\},$n$]:} defines the $S$-sum
$\S{a_1,\ldots,a_m}{x_1,\ldots,x_m;n}$
of depth $m$ with upper summation limit $n$ and\footnote{Obviously, the $x_i$ must be chosen from a subfield of $\set R$ which is computable, in particular which can be treated within the computer algebra system Mathematica.} $a_i\in \N,$ $x_i\in \R^*$.
\medskip

\noindent
{\tt Z[$a_1$,...,$a_m$,\{$x_1$,...,$x_m$\},$n$]:} defines the Z-sum
$\textnormal{Z}_{a_1,\ldots,a_m}\left(x_1,\ldots,x_m;n\right) $ of depth $m$
with upper summation
limit $n$ and $a_i\in \N,$  $x_i\in \R^*$; see~\eqref{Equ:ZSum}.

\medskip

\noindent
{\tt H[$m_1$,...,$m_n$,$x$]:}
defines the harmonic polylogarithm $\H{m_1,\ldots,m_n}x$ ($m_i\in\{-1,0,1\}$)
or a generalized polylogarithm $\H{m_1,\ldots,m_n}x$ ($m_i\in\R$), while {\tt H[$\{a_1,b_1\},...,\{a_n,b_n\},x$]}
defines a cyclotomic harmonic polylogarithm~\cite{Ablinger:2011te} with $a_i,b_i\in \N$.

\medskip

\noindent
{\tt SToZ[S[$a_1$,...,$a_m$,\{$x_1$,...,$x_m$\},$n$]]:} converts the $S$-sum
$\S{a_1,\ldots,a_m}{x_1,\ldots,x_m;n}$ to Z-sums.

\smallskip

\noindent
{Example:}
\begin{eqnarray}
&& {\tt SToZ[S[1, 1, \{1/2, -1/4\}, n]} =
Z_{2}\left(-\frac{1}{8},n\right)+Z_{1,1}\left(\frac{1}{2},-\frac{1}{4},n\right)
\nonumber
\end{eqnarray}

\medskip

\noindent
{\tt ZToS[Z[$a_1$,...,$a_m$,\{$x_1$,...,$x_m$\},$n$]]:} converts the Z-sum
$\textnormal{Z}_{a_1,\ldots,a_m}\left(x_1,\ldots,x_m;n\right)$ to $S$-sums.

\smallskip

\noindent
{Example:}
\begin{eqnarray}
&& {\tt ZToS[Z[1,1, \{1/2, -1/4\}, n]} =
-S_{2}\left(-\frac{1}{8},n\right)+S_{1,1}\left(\frac{1}{2},-\frac{1}{4},n\right)
\nonumber
\end{eqnarray}

\medskip

\noindent
{\tt ReduceToBasis[$expr$]:} reduces occurring harmonic sums, $S$-sums or cyclotomic
sums to basis sums
which are (in the setting of their quasi-shuffle algebra) algebraic independent over each other. In general, tables (available at the \texttt{HarmonicSums} homepage)  are exploited for this reduction. If sums are not covered by the given tables, one can set the option
{\tt Dynamic $\rightarrow$ Automatic}: then for such sums (outside of the tables) the reduction is calculated online. Using the setting {\tt Dynamic $\rightarrow$ True} the reduction for all sums is calculated online.

\smallskip

\noindent
{Example:}
\begin{eqnarray}
&& {\tt ReduceToBasis[S[1, 2, 2, n],n,Dynamic \rightarrow True]} = \nonumber\\
&& S_2(n)
   \left(S_3(n)-S_{2,1}(n)\right)-S_{2,3}(n)-S_{4,1}(n)+S_{2,2,1}(n)+\frac{1}
   {2} S_1(n) \left(S_2(n){}^2+S_4(n)\right)+S_5(n)
\nonumber
\end{eqnarray}

\medskip

\noindent
{\tt ReduceToHBasis[$expr$]:}
reduces occurring harmonic polylogarithms, generalized polylogarithms or cyclotomic polylogarithms to basis
polylogarithms which are algebraic independent over each other by using given tables. If the option
{\tt Dynamic $\rightarrow$ True} is set, the reduction is calculated online; this will
help if certain polylogarithms
are not stored in the table.

\smallskip

\noindent
{Example:}
\begin{eqnarray}
&& {\tt ReduceToHBasis[H[1, 0, -1, 1, x]]} =
\nonumber\\ &&
H_{-1,1}(x) \left(H_0(x) H_1(x)-H_{0,1}(x)\right)+H_{-1,1}(x)
   H_{0,1}(x)-H_1(x) \left(H_{-1}(x)
   H_{0,1}(x)-H_{0,-1,1}(x) \right.
\nonumber\\ &&
\left.-H_{0,1,-1}(x)\right)-H_1(x) \left(H_0(x)
   H_{-1,1}(x)-H_{-1}(x) H_{0,1}(x)+H_{0,1,-1}(x)\right)-2
   H_{0,-1,1,1}(x)
\nonumber\\ &&
-H_{0,1,-1,1}(x)
\nonumber
\end{eqnarray}

\medskip

\noindent
{\tt TransformToSSums[$expr$]:}
returns an expression which transforms indefinite nested sums to harmonic sums,
$S$-sums or cyclotomic sums
whenever possible.

\smallskip

\noindent
Example:
\begin{eqnarray}
&&{\tt expr} =
(n-1)^2 \sum_{i=4}^n \frac{(-1/3)^i}{(i-3)^2}
   + 4 \zeta_3 (n-2)^2 (n-1)^2 \sum_{i=1}^n \frac{(-1)^i}{i-2}
   + n^2(n+1)^2 \sum_{i=2}^n \frac{(-1)^i}{i-1}
\nonumber\\ &&
   + \sum_{i=1}^n \frac{1}{2+i} \sum_{j=1}^i \frac{(-3)^j}{(1+3j)^2}
   -\frac{3}{2} (n-2)^2 (1+n)^3 \zeta_3 \sum_{i=4}^n \frac{1}{i-1} \sum_{j=4}^i
   \frac{(-1)^j}{(j-3)^2} + \sum_{i=1}^n \frac{(-1)^i}{i^2} S_1(i)
\nonumber\\ &&
   + \sum_{i=1}^n \frac{(-3)^i}{(3i+1)^2} S_1(i)
\nonumber\\ &&
\nonumber\\ &&
{\tt TransformToSSums[expr]}= \nonumber\\
&&
\frac{1}{2} S_{(1,0,1),(3,1,2)}(1,-3;n)
+S_{(3,1,2),(1,0,1)}(-3,1;n)+2 (n-2)^2 (n+1)^3 S_{-3}(\infty ) \times
\nonumber\\ &&
\Biggl[-S_{1,-2}(n)-\frac{(-1)^n}{n^3}-\frac{(-1)^n}{n^2}+S_{-3}(n)
+\frac{S_{-2}(n)}{n}+S_{-2}(n)-2 S_{-1}(n)+\frac{(-1)^{n-1}}{n-1}
+\frac{2 (-1)^n}{n}
\nonumber\\ &&
-\frac{(-1)^{n-1}}{(n-1)^2}-1\Biggr]
+S_{-2,1}(n)+(n-2)^2
\Biggl[-\frac{1}{27} S_2\left(-\frac{1}{3};n\right)+\frac{(-1)^n
3^{-n-3}}{n^2}+\frac{(-1)^{n-2} 3^{-n-1}}{(n-2)^2}
\nonumber\\ &&
+\frac{(-1)^{n-1}
3^{-n-2}}{(n-1)^2}\Biggr]+n^2 (n+1)^3 \left(\frac{(-1)^n}{n^3}-S_{-3}(n)\right)+4
(n-2)^2 (n-1)^2
   \Biggl[S_{-1}(n)-\frac{(-1)^{n-1}}{n-1}
\nonumber\\ &&
-\frac{(-1)^n}{n}\Biggr] S_3(\infty )
\nonumber
\end{eqnarray}
Here $S_{(a_1,b_1,c_1),...,(a_m,b_m,c_m)}(x_1,...,x_m;n)$ are generalized
cyclotomic sums, cf.~Ref.~\cite{Ablinger:2011te}.

\medskip

\noindent
{\tt ReduceSums[$expr$]:}
returns an expression which transforms indefinite nested sums to harmonic sums,
$S$-sums
or cyclotomic sums whenever possible and reduces occurring harmonic sums, $S$-sum or
cyclotomic
sums to basis sums which are algebraic independent over each other by using given tables. If the
option {\tt Dynamic $\rightarrow$ True} is set, the reduction is calculated online;
this will
help if certain sums are not stored in the table.

\smallskip

\noindent
Example:
\begin{eqnarray}
&& {\tt ReduceSums[expr]}=
\nonumber\\ &&
\frac{1}{2} \Biggl\{
S_{(1,0,1),(3,1,2)}(1,-3;n)+2 S_{(3,1,2),(1,0,1)}(-3,1;n) -3 (n-2)^2 (n+1)^3
S_3(\infty) \Biggl[S_{-2,1}(n)
\nonumber\\ &&
-\frac{(-1)^n}{n^3}-\frac{(-1)^n}{n^2}-2 S_{-1}(n)+S_{-2}(n)
   \left(-S_1(n)+\frac{1}{n}+1\right)-\frac{(-1)^n}{n-1}
+\frac{2 (-1)^n}{n}
+\frac{(-1)^n}{(n-1)^2}
\nonumber\\ &&
-1\Biggr]
+\frac{3^{-n-3} \left(2
(-1)^n \left(7 n^4-12
   n^3+10 n^2-12 n+4\right)-2\ 3^n n^2 \left(n^2-3 n+2\right)^2
   S_2\left(-\frac{1}{3};n\right)\right)}{(n-1)^2 n^2}
\nonumber\\ &&
+2 n^2 (n+1)^3
   \left(\frac{(-1)^n}{n^3}-S_{-3}(n)\right)+\frac{8 (n-2)^2 (n-1) \left((n-1) n
S_{-1}(n)+(-1)^n\right)
   S_3(\infty )}{n}
\nonumber\\ &&
+2 S_{-2,1}(n) \Biggr\}
\nonumber
\end{eqnarray}

\medskip

\noindent
{\tt LinearHExpand[$expr$]:}
expands all the products of harmonic polylogarithms, generalized polylogarithms and cyclotomic
polylogarithms at same arguments.

\smallskip

\noindent
{Example:}
\begin{eqnarray}
&& {\tt LinearHExpand[H[0, 1, x]*H[-1, 0, -1/2, x]]} = \nonumber\\
&& H_{-1,0,-\frac{1}{2},0,1}(x)
+2 H_{-1,0,0,-\frac{1}{2},1}(x)
+2 H_{-1,0,0,1,-\frac{1}{2}}(x)
  +H_{-1,0,1,0,-\frac{1}{2}}(x)
  +H_{0,-1,0,-\frac{1}{2},1}(x)
\nonumber\\ &&
   +H_{0,-1,0,1,-\frac{1}{2}}(x)+H_{0,-1,1,0,-\frac{1}{2}}(x)
   +H_{0,1,-1,0,-\frac{1}{2}}(x)
\nonumber
\end{eqnarray}

\medskip

\noindent
{\tt LinearExpand[$expr$]:}
expands all the products of harmonic sums, $S$-sums and cyclotomic sums at same arguments.

\smallskip

\noindent
{Example:}
\begin{eqnarray}
&& {\tt LinearExpand[S[1, 2, n]*S[-1, 2, {1/2, 1/3}, n]]} = \nonumber\\
&& S_{0,4}\left(\frac{1}{2},\frac{1}{3};n\right)
-S_{-1,1,4}\left(\frac{1}{2},1,\frac{1}{3};n\right)
-S_{-1,3,2}\left(\frac{1}{2},\frac{1}{3},1;n\right)
-S_{0,2,2}\left(\frac{1}{2},\frac{1}{3},1;n\right)
\nonumber\\ &&
-S_{0,2,2}\left(\frac{1}{2},1,\frac{1}{3};n\right)
-S_{1,-1,4}\left(1,\frac{1}{2},\frac{1}{3};n\right)
-S_{1,1,2}\left(1,\frac{1}{2},\frac{1}{3};n\right)
+S_{-1,1,2,2}\left(\frac{1}{2},1,\frac{1}{3},1;n\right)
\nonumber\\ &&
+S_{-1,1,2,2}\left(\frac{1}{2},1,1,\frac{1}{3};n\right)
+S_{-1,2,1,2}\left(\frac{1}{2},\frac{1}{3},1,1;n\right)
+S_{1,-1,2,2}\left(1,\frac{1}{2},\frac{1}{3},1;n\right)
\nonumber\\ &&
+S_{1,-1,2,2}\left(1,\frac{1}{2},1,\frac{1}{3};n\right)
+S_{1,2,-1,2}\left(1,1,\frac{1}{2},\frac{1}{3};n\right)
\nonumber
\end{eqnarray}

\medskip

\noindent
{\tt InvMellin[$expr,n,x$]:}
calculates the inverse Mellin-Transform of {\tt expr.}

\smallskip

\noindent
{Example:}
\begin{eqnarray}
&&{\tt ReduceConstants[InvMellin[S[2, 2, 1, \{1, 1, 1/2\}, n], n, x],
 ToKnownConstants \rightarrow True]} = \nonumber\\
&&
-\frac{2^{-n}
   S_{2,2}\left(\frac{1}{2},1;\infty \right)}{x-2}-\frac{2^{-n-1}
   \left(\zeta_2-\ln(2)^2\right) H_{0,2}(x)}{x-2}-\frac{\ln(2)
   2^{-n} H_{0,2,0}(x)}{x-2}
+\frac{\ln(2) H_{0,1,0}(x)}{x-1}
\nonumber\\ &&
+\frac{2^{-n} H_{0,2,0,2}(x)}{x-2}
+\frac{2^{-n} S_4\left(\frac{1}{2};\infty\right)}{x-2}
  +\frac{2^{-n} H_0(x)\left(\frac{\ln(2)^3}{6}
 -\frac{\ln(2) \zeta_2}{2}
  +\frac{7\zeta_3}{8}\right)}{x-2}
  +\frac{5 \zeta_3 H_0(x)}{8 (x-1)}
\nonumber\\ &&
  -\frac{H_0(x)\left(\frac{\ln(2)^3}{6}
  -\frac{\ln(2) \zeta_2}{2}
  +\frac{7 \zeta_3}{8}\right)}{x-1}
  -\frac{5\ 2^{-n-3} \zeta_3 H_0(x)}{x-2}
  +\frac{\ln(2) 2^{1-n}\left(\frac{\ln(2)^3}{6}
  -\frac{\ln(2) \zeta_2}{2}
  +\frac{7 \zeta_3}{8}\right)}{2-x}
\nonumber\\ &&
  +\frac{2^{-n-2}\left(\zeta_2 -\ln(2)^2\right)^2}{x-2}
  +\frac{2 \ln(2) \zeta_3}{x-1}
+ \delta(1-x) \Biggl[
   \ln(2) \left(\frac{1}{4}
   \left(\zeta_2-\ln(2)^2\right)^2
   -2 \left(
   -S_{1,3}\left(\frac{1}{2},1;\infty\right)
\right. \right. \nonumber\\ && \left. \left.
   +S_4\left(\frac{1}{2};\infty \right)
+\ln(2) \left(
   \frac{\ln(2)^3}{6}
   -\frac{\ln(2) \zeta_2}{2}
   +\frac{7\zeta_3}{8}\right)\right)\right)
   -S_{1,4}\left(\frac{1}{2},1;\infty\right)
   -S_{3,2}\left(\frac{1}{2},1;\infty\right)
\nonumber\\ &&
   +S_{1,2,2}\left(\frac{1}{2},1,1;\infty\right)
   +S_5\left(\frac{1}{2};\infty \right)\Biggr]
\nonumber
\end{eqnarray}
For the command \texttt{ReduceConstants} we refer to page~\pageref{Equ:ReduceConstants}. Note that some of the constants may be further identified by
\begin{eqnarray}
S_k\left(z, \infty\right) = \text{Li}_k\left(z\right), |z| \leq 1~.
\nonumber
\end{eqnarray}

\medskip

\noindent
{\tt Mellin[$expr,x,n$]:}
calculates the Mellin-Transform of {\tt expr.}

\smallskip

\noindent
{Example:}
\begin{eqnarray}
&& {\tt Mellin[H[2, 1, -3, x], x, n]~~/.~~S[1, \{-1/3\}, Infinity] \rightarrow -Log[4/3]} =
\nonumber\\
&&
-\frac{2^{n+1} S_{1,2}\left(\frac{1}{2},-\frac{2}{3};\infty\right)}{n+1}
+\frac{S_{1,2}\left(\frac{1}{2},-\frac{2}{3};\infty\right)}{n+1}
-\frac{2^{n+1} S_{2,1}\left(1,-\frac{1}{3};\infty\right)}{n+1}
+\frac{S_{2,1}\left(1,-\frac{1}{3};\infty\right)}{n+1}
\nonumber\\ &&
+\frac{2^{n+1}S_{1,1,1}\left(\frac{1}{2},2,-\frac{1}{3};\infty\right)}{n+1}
-\frac{S_{1,1,1}\left(\frac{1}{2},2,-\frac{1}{3};\infty\right)}{n+1}
-\frac{S_{1,1}\left(-3,-\frac{1}{3};n\right)}{(n+1)^2}
-\frac{S_3\left(-\frac{1}{3};\infty\right)}{n+1}
\nonumber\\ &&
-\frac{2^{n+1}S_{1,1,1}\left(\frac{1}{2},-3,-\frac{1}{3};n\right)}{n+1}
-\frac{2^{n+1}\log \left(\frac{4}{3}\right)S_{1,1}\left(\frac{1}{2},-3;n\right)}{n+1}
+\frac{2^{n+1} \log\left(\frac{4}{3}\right)S_{1,1}\left(\frac{1}{2},1;n\right)}{n+1}
\nonumber\\ &&
+\frac{2^{n+1}S_3\left(-\frac{1}{3};\infty\right)}{n+1}
+\frac{(-1)^n 3^{n+1}S_1\left(-\frac{1}{3};n\right)}{(n+1)^3}
-\frac{\log\left(\frac{4}{3}\right) S_1(-3;n)}{(n+1)^2}
+\frac{\log\left(\frac{4}{3}\right) S_1(n)}{(n+1)^2}
\nonumber\\ &&
-\frac{1}{(n+1)^4}+\frac{(-1)^n 3^{n+1} \log \left(\frac{4}{3}\right)}{(n+1)^3}
+\frac{\log\left(\frac{4}{3}\right)}{(n+1)^3}
\nonumber
\end{eqnarray}

\medskip

\noindent
{\tt DifferentiateSSum[$expr,n$]:}
differentiates {\tt expr} w.r.t.\ $n$.

\noindent
{Example:}
\begin{eqnarray}
&& {\tt  DifferentiateSSum[S[4, 1, \{1, -1/2\}, n], n]]}\nonumber\\ %
&& =
H_0\left(\frac{1}{2}\right)
   S_{4,1}\left(1,-\frac{1}{2};n\right)-S_{4,2}\left(1,-\frac{1}{2};n\right)-
   4 S_{5,1}\left(1,-\frac{1}{2};n\right)+H_{-1,0}\left(\frac{1}{2}\right)
   S_4(n) \nonumber\\
&&
   +H_{\frac{3}{2}}\left(\frac{1}{2}\right)%
   H_{0,0,0,1,0}(1)+H_{-1,0,0,0,-1,0}\left(\frac{1}{2}\right)%
\nonumber %
\end{eqnarray}
The constants may be further reduced to known multiple zeta values and/or cyclotomic
numbers applying the command  {\tt ReduceConstants[$expr$]} (see page~\pageref{Equ:ReduceConstants}).

\medskip

\noindent
{\tt RemoveLeading1[H[$m_1$,...,$m_n$,$x$]]:}
extracts leading ones of $\H{m_1,\ldots,m_n}x$.

\smallskip

\noindent
{Example:}
\begin{eqnarray}
&& {\tt RemoveLeading1[H[1, 1, -1, 0, 1/2, x]]} \nonumber\\
&& =
\frac{1}{2} H_1(x){}^2 H_{-1,0,\frac{1}{2}}(x)
-H_1(x) H_{-1,0,\frac{1}{2},1}(x)-H_1(x) H_{-1,0,1,\frac{1}{2}}(x)
-H_1(x) H_{-1,1,0,\frac{1}{2}}(x) \nonumber\\ &&
+H_{-1,0,\frac{1}{2},1,1}(x)+H_{-1,0,1,\frac{1}{2
   },1}(x)+H_{-1,0,1,1,\frac{1}{2}}(x)+H_{-1,1,0,\frac{1}{2},1}(x)+H_{-1,1,0,
   1,\frac{1}{2}}(x)
\nonumber\\ &&
+H_{-1,1,1,0,\frac{1}{2}}(x) \nonumber
\end{eqnarray}

\medskip

\noindent
{\tt RemoveTrailing0[H[$m_1$,...,$m_n$,$x$]]:}
extracts trailing zeroes of $\H{m_1,\ldots,m_n}x$.

\smallskip

\noindent
{Example:}
\begin{eqnarray}
&& {\tt RemoveTrailing0[H[1/2, -1/3, -1, 0, 0, x]]} \nonumber\\
&& =
\frac{1}{2} H_0(x){}^2 H_{\frac{1}{2},-\frac{1}{3},-1}(x)
-H_0(x) H_{0,\frac{1}{2},-\frac{1}{3},-1}(x)
-H_0(x) H_{\frac{1}{2},-\frac{1}{3},0,-1}(x)
-H_0(x) H_{\frac{1}{2},0,-\frac{1}{3},-1}(x)
\nonumber\\ &&
+H_{0,0,\frac{1}{2},-\frac{1}{3},-1}(x
   )+H_{0,\frac{1}{2},-\frac{1}{3},0,-1}(x)+H_{0,\frac{1}{2},0,-\frac{1}{3},-
   1}(x)
+H_{\frac{1}{2},-\frac{1}{3},0,0,-1}(x)+H_{\frac{1}{2},0,-\frac{1}{3}
   ,0,-1}(x)
   \nonumber\\ &&
+H_{\frac{1}{2},0,0,-\frac{1}{3},-1}(x)
\nonumber
\end{eqnarray}

\medskip

\noindent
{\tt RemoveLeadingIndex[H[$m_1$,...,$m_n$,$x$]]:} extracts the leading index
$m_1$ of $\H{m_1,\ldots,m_n}x$. \\
\noindent{\tt RemoveLeadingIndex[$expr,b$]} removes the
leading index $b$ of all polylogarithms in $expr$.

\smallskip

\noindent
{Examples~:}
\begin{eqnarray}
&& {\tt RemoveLeadingIndex[H[1/2, 1/2, -1, 1, 1/3, x]]} = \nonumber\\
&&
\frac{1}{2} H_{\frac{1}{2}}(x){}^2 H_{-1,1,\frac{1}{3}}(x)
-H_{\frac{1}{2}}(x)
   H_{-1,\frac{1}{2},1,\frac{1}{3}}(x)
-H_{\frac{1}{2}}(x)
   H_{-1,1,\frac{1}{3},\frac{1}{2}}(x)
-H_{\frac{1}{2}}(x)
   H_{-1,1,\frac{1}{2},\frac{1}{3}}(x)
\nonumber\\ &&
+H_{-1,\frac{1}{2},\frac{1}{2},1,\frac{
   1}{3}}(x)+H_{-1,\frac{1}{2},1,\frac{1}{3},\frac{1}{2}}(x)+H_{-1,\frac{1}{2
   },1,\frac{1}{2},\frac{1}{3}}(x)+H_{-1,1,\frac{1}{3},\frac{1}{2},\frac{1}{2
   }}(x)
+H_{-1,1,\frac{1}{2},\frac{1}{3},\frac{1}{2}}(x)
\nonumber\\ &&
+H_{-1,1,\frac{1}{2},
   \frac{1}{2},\frac{1}{3}}(x)
\nonumber\\ &&
\nonumber\\ &&
{\tt RemoveLeadingIndex[H[b, 1/2, -1, 1, -1/3, x], b]} = \nonumber\\ &&
H_{\frac{1}{2},-1,1,-\frac{1}{3}}(x)
   H_b(x)-H_{\frac{1}{2},-1,1,-\frac{1}{3},b}(x)-H_{\frac{1}{2},-1,1,b,
   -\frac{1}{3}}(x)-H_{\frac{1}{2},-1,b,1,-\frac{1}{3}}(x)-H_{\frac{1}{2},b,-1,1,
   -\frac{1}{3}}(x)
\nonumber
\end{eqnarray}

\medskip

\noindent
{\tt RemoveTrailingIndex[H[$m_1$,...,$m_n$,$x$]]:} extracts the trailing index $m_n$ of
$\H{m_1,\ldots,m_n}x$.\\
{\tt RemoveTrailingIndex[$expr,b$]} removes the trailing index
$b$ of all polylogarithms in $expr$.

\smallskip

\noindent
{Examples~:}
\begin{eqnarray}
&& {\tt RemoveTrailingIndex[H[-1, 1, 1/3, 1/2, 1/2, x]]} = \nonumber\\
&& \frac{1}{2} H_{\frac{1}{2}}(x){}^2 H_{-1,1,\frac{1}{3}}(x)
   -H_{\frac{1}{2}}(x)H_{-1,\frac{1}{2},1,\frac{1}{3}}(x)
   -H_{\frac{1}{2}}(x)H_{-1,1,\frac{1}{2},\frac{1}{3}}(x)
   -H_{\frac{1}{2}}(x)H_{\frac{1}{2},-1,1,\frac{1}{3}}(x)
\nonumber\\ &&
   +H_{-1,\frac{1}{2},\frac{1}{2},1,\frac{
   1}{3}}(x)+H_{-1,\frac{1}{2},1,\frac{1}{2},\frac{1}{3}}(x)+H_{-1,1,\frac{1}
   {2},\frac{1}{2},\frac{1}{3}}(x)+H_{\frac{1}{2},-1,\frac{1}{2},1,\frac{1}{3
   }}(x)+H_{\frac{1}{2},-1,1,\frac{1}{2},\frac{1}{3}}(x)
\nonumber\\ &&
   +H_{\frac{1}{2},\frac{1}{2},-1,1,\frac{1}{3}}(x)
\nonumber\\
&& {\tt RemoveTrailingIndex[H[-1, 1, 1/3, 1/2, b, x],b]} = \nonumber\\
&& H_{-1,1,\frac{1}{3},\frac{1}{2}}(x)
   H_b(x)-H_{-1,1,\frac{1}{3},b,\frac{1}{2}}(x)-H_{-1,1,b,\frac{1}{3},\frac{
   1}{2}}(x)-H_{-1,b,1,\frac{1}{3},\frac{1}{2}}(x)-H_{b,-1,1,\frac{1}{3},
   \frac{1}{2}}(x)
\nonumber
\end{eqnarray}

\medskip

\noindent
{\tt SRemoveLeading1[S[$a_1$,...,$a_m$,\{$x_1$,...,$x_m$\},$n$]]:} extracts leading ones of
$\S{a_1,\ldots,a_m}{x_1,\ldots,x_m;n}$.

\smallskip

\noindent
{Example:}
\begin{eqnarray}
&& {\tt SRemoveLeading1[S[1, 1, 1, -2, 4, \{1, 1/3, 1/2, 1/2, 1/2\}, n]]} =\nonumber\\
&&S_1(n)
   S_{1,1,-2,4}\left(\frac{1}{3},\frac{1}{2},\frac{1}{2},\frac{1}{2};n\right)
   +S_{1,1,-2,5}\left(\frac{1}{3},\frac{1}{2},\frac{1}{2},\frac{1}{2};n\right
   )+S_{1,1,-1,4}\left(\frac{1}{3},\frac{1}{2},\frac{1}{2},\frac{1}{2};n\right)
\nonumber\\
&&
   +S_{1,2,-2,4}\left(\frac{1}{3},\frac{1}{2},\frac{1}{2},\frac{1}{2};n\right)
   +S_{2,1,-2,4}\left(\frac{1}{3},\frac{1}{2},\frac{1}{2},\frac{1}{2};n\right)
   -S_{1,1,-2,1,4}\left(\frac{1}{3},\frac{1}{2},\frac{1}{2},1,\frac{1}{2}
   ;n\right)
\nonumber\\
&&
-S_{1,1,-2,4,1}\left(\frac{1}{3},\frac{1}{2},\frac{1}{2},\frac{1}
   {2},1;n\right)-S_{1,1,1,-2,4}\left(\frac{1}{3},\frac{1}{2},1,\frac{1}{2},\frac{1}{2};n\right)
\nonumber\\
&&
-S_{1,1,1,-2,4}\left(\frac{1}{3},1,\frac{1}{2},\frac{1}{2},\frac{1}{2};n\right)
\nonumber
\end{eqnarray}

\medskip

\noindent
{\tt SRemoveTrailing1[S[$a_1$,...,$a_m$,\{$x_1$,...,$x_m$\},$n$]]:} extracts trailing ones of
$\S{a_1,\ldots,a_m}{x_1,\ldots,x_m;n}$.

\smallskip

\noindent
{Example:}
\begin{eqnarray}
&& {\tt SRemoveTrailing1[S[-2, 4, 1, 1, 1, \{1/3, 1/2, 1/2, -1/2, 1\}, n]]} = \nonumber\\
&&
   S_1(n) S_{-2,4,1,1}\left(\frac{1}{3},\frac{1}{2},\frac{1}{2},-\frac{1}{2};n\right)
   +S_{-2,4,1,2}\left(\frac{1}{3},\frac{1}{2},\frac{1}{2},-\frac{1}{2};n\right)
   +S_{-2,4,2,1}\left(\frac{1}{3},\frac{1}{2},\frac{1}{2},-\frac{1}{2};n\right)
\nonumber\\ &&
   +S_{-2,5,1,1}\left(\frac{1}{3},\frac{1}{2},\frac{1}{2},-\frac{1}{2};n\right)
   +S_{-1,4,1,1}\left(\frac{1}{3},\frac{1}{2},\frac{1}{2},-\frac{1}{2};n\right)
   -S_{-2,1,4,1,1}\left(\frac{1}{3},1,\frac{1}{2},\frac{1}{2},-\frac{1}{2};n\right)
\nonumber\\ &&
   -S_{-2,4,1,1,1}\left(\frac{1}{3},\frac{1}{2},\frac{1}{2},1,-\frac{1}{2};n\right)
   -S_{-2,4,1,1,1}\left(\frac{1}{3},\frac{1}{2},1,\frac{1}{2},-\frac{1}{2};n\right)
\nonumber\\ &&
   -S_{1,-2,4,1,1}\left(1,\frac{1}{3},\frac{1}{2},\frac{1}{2},-\frac{1}{2};n\right)
\nonumber
\end{eqnarray}

\medskip

\noindent
{\tt SRemoveTrailingIndex[S[$a_1$,...,$a_m$,\{$x_1$,...,$x_m$\},$n$]]:}
extracts the trailing index $(a_m,x_m)$ of
$\S{a_1,\ldots,a_m}{x_1,\ldots,x_m;n}$.

\smallskip

\noindent
{Examples~:}
\begin{eqnarray}
&&{\tt SRemoveTrailingIndex[S[a_1,...,a_m,\{x_1,...,x_m\},n]]}  = \nonumber\\
&&
S_{a_3}\left(x_3; n\right) S_{a_1, a_2}\left(x_1, x_2; n\right)
+ S_{a_1, a_2 + a_3}\left(x_1, x_2 x_3, n\right)
+ S_{a_1 + a_3, a_2}\left(x_1 x_3, x_2; n\right)
\nonumber \\ &&
- S_{a_1, a_3, a_2} \left(x_1, x_3, x_2; n\right)
- S_{a_3, a_1, a_2} \left(x_3, x_1, x_2; n\right)
\nonumber
\end{eqnarray}

\medskip

\noindent
{\tt TransformH[H[$m_1$,...,$m_n$,$y$],$x$]:}
performs several transforms on the argument of a generalized
polyloarithm.

\smallskip

\noindent
{Examples~:}
\begin{eqnarray}
&& {\tt TransformH[H[-1, 2, 0, 1/x], x]} = \nonumber \\
&&
\left(H_{0,\frac{1}{2}}(1)+H_{0,2}(1)\right)
   \left(-H_{-1}(x)+H_0(x)+\ln(2)\right)
   -H_{-1,0,0}(x)
   -H_{-1,\frac{1}{2},0}(x)
   +H_{0,0,0}(x)
   \nonumber\\ &&
   +H_{0,\frac{1}{2},0}(x)
   +H_{-1,\frac{1}{2},0}(1)
   +H_{-1,2,0}(1)-H_{0,\frac{1}{2},0}(1)+\frac{3 \zeta_3}{4}
\nonumber \\
&& {\tt TransformH[H[-1, 2, 0, 1-x], x]} =
-H_{2,0}(1)H_2(x)-H_{2,-1,1}(x)+H_{-1,2,0}(1)
\nonumber\\
&& {\tt TransformH[H[-1, 2, 0,  (1-x)/(1+x)],x]} = \nonumber \\
&& H_{2,0}(1) \left(-H_{-1}(x)\right)+H_{-1,-1,-1}(x)+H_{-1,-1,1}(x)-H_{-1,-\frac{1}{3},-1}(x)-H_{-1,-\frac{1}{3},1}(x)+H_{-1,2,0}(1)
\nonumber
\end{eqnarray}

\medskip

\noindent
{\tt HarmonicSumsSeries[$expr, \{x, c, ord\}$]:}
tries to compute the series expansion of $expr$ with respect to
$x$ about $c\in\R\cup\infty$ up to order $ord\in\N$. $expr$ may contain harmonic sums,
cyclotomic sums, and their generalizations as well as their associated
polylogarithms.

\smallskip

\noindent
Examples~:
\begin{eqnarray}
&& {\tt HarmonicSumsSeries[H[-2/5, 3/24, 0, 1/4, x], \{x, 0, 10\}]} = \nonumber\\
&&
   \frac{40x^3}{3}
   +35 x^4
   +\frac{2026 x^5}{9}
   +\frac{20081 x^6}{18}
   +\frac{12632239 x^7}{1890}
   +\frac{1205544979 x^8}{30240}
   +\frac{119322150247 x^9}{476280}
\nonumber\\ &&
   +\frac{188972192341 x^{10}}{117600}
\nonumber\\
&& {\tt HarmonicSumsSeries[S[1, -2, 1, n], \{n, 0, 4\}]} = \nonumber\\
&&
 n  \left(-\frac{\zeta_2 \zeta_3}{2}
   -\frac{5 \zeta_5}{8}\right)
+n^2 \left(\frac{41
   \zeta_2^3}{168}+\frac{\zeta_3^2}{4}\right)
   +n^3
   \left(S_{-5,-2}(\infty )+S_{-4,-3}(\infty )-S_{-4,1,-2}(\infty
   )\right.
\nonumber\\ && \left.
+\frac{\zeta_2^2 \zeta_3}{32}-\frac{29 \zeta_2
   \zeta_5}{32}-\zeta_7\right)
+n^4 \left(-S_{-6,-2}(\infty )-S_{-5,-3}(\infty )+S_{-5,1,-2}(\infty
   )-\frac{31 \ln(2) \zeta_2 \zeta_5}{16} \right.
\nonumber\\ && \left.
+S_8(\infty )+s_6
   \zeta_2+\frac{23 \zeta_2^4}{70}-\frac{3 \zeta_2
   \zeta_3^2}{4}+\frac{155 \zeta_3 \zeta_5}{128}\right)
\nonumber \\
&& {\tt HarmonicSumsSeries[S[1, -2, 1, n], \{n, Infinity, 10\}]} = \nonumber\\
&&\frac{1}{8} \Biggl[(-1)^n
   \left(
 \frac{2}{n^3}
-\frac{5}{n^4}
+\frac{4}{n^5}
+\frac{7}{n^6}
-\frac{16}{n^7}
-\frac{27}{n^8}
+\frac{108}{n^9}
+\frac{187}{n^{10}}
\right)
   \bar{\ln}(n)
\nonumber\\ &&
+ \zeta_3 \left(
-\frac{5}{2 n}
+\frac{5}{12 n^2}
-\frac{1}{24 n^4}
+\frac{5}{252 n^6}
-\frac{1}{48 n^8}
+\frac{5}{132 n^{10}}\right)
\nonumber\\ &&
-\frac{3 \zeta_2^2}{5}
+(-1)^n
   \left(
+\frac{3}{n^4}
-\frac{20}{3 n^5}
-\frac{19}{12 n^6}
+\frac{1511}{60 n^7}
+\frac{59}{8 n^8}
-\frac{118529}{630 n^9}
-\frac{130213}{2520 n^{10}}
\right)
-5 \zeta_3 \bar{\ln}({n})\Biggr] + O\left(\frac{1}{n^{10}}\right)
\nonumber \\
&&
{\tt HarmonicSumsSeries[S[1, -2, 1, \{1/2, 1/3, -1\}, n], \{n, Infinity, 3\}]} =
\nonumber\\
&&
\frac{3}{2} S_{1,1}\left(-\frac{1}{6},2;\infty \right)
-\frac{3}{2} S_{1,1}\left(-\frac{1}{6},6;\infty \right)
+\frac{3}{2} \ln(2) S_1\left(-\frac{1}{3};\infty \right)
-\frac{3}{2} \ln(2) S_1\left(-\frac{1}{6};\infty \right)
\nonumber\\ &&
+\frac{3}{2} \ln(2)S_1\left(\frac{1}{6};\infty \right)
-\frac{9\ 2^{-n}S_1\left(-\frac{1}{3};\infty \right)}{n^3}
+\frac{3\ 2^{-n} S_1\left(-\frac{1}{3};\infty \right)}{n^2}
-\frac{3\ 2^{-n-1} S_1\left(-\frac{1}{3};\infty \right)}{n}
\nonumber\\ &&
-\frac{3}{2}S_1\left(-\frac{1}{3};\infty \right) S_1\left(-\frac{1}{6};\infty
   \right)
-\frac{1011}{800} S_1\left(-\frac{1}{6};\infty \right)
-\frac{3}{2}
   S_2\left(-\frac{1}{3};\infty \right)
-\frac{7 \ln(2) 2^{-n}
   3^{2-n}}{125 n^3}
\nonumber\\ &&
+\frac{\ln(2) 2^{-n} 3^{2-n}}{25 n^2}-\frac{7}{25}
   \ln(2) 2^{-n-1} 3^{1-n}-\frac{1}{5} \ln(2) 2^{-n-1} 3^{-n}
   n-\frac{\ln(2) 2^{-n-1} 3^{1-n}}{5 n}-\frac{21
   \ln(2)}{32}
\nonumber\\ &&
+\frac{79 \left(-\frac{1}{3}\right)^n 2^{-n-5}}{343
   n^3}+\frac{63\ 2^{-n-4}}{n^3}-\frac{21\ 2^{-n-4}}{n^2}-\frac{(-1)^n
   2^{-n-4} 3^{1-n}}{49 n^2}+\frac{1}{7} \left(-\frac{1}{3}\right)^n
   2^{-n-3}+\frac{\left(-\frac{1}{3}\right)^n 2^{-n-5}}{7 n}
\nonumber\\ &&
+\frac{21\
   2^{-n-5}}{n}-\frac{3 \zeta_2}{4}+\frac{9}{280} + O\left(\frac{1}{n^4}\right),
\nonumber
\end{eqnarray}

\noindent
with $\bar{\ln}(b) = \ln(n) + \gamma_E$.\\

\medskip

\noindent
{\tt ComputeSSumBasis[$w,\{x_1,x_2,...\},n$]:} computes a basis representation of
$S$-sums with indices $x_i$
at weight $w$ using algebraic, differential and half integer relations.

\smallskip

\noindent
Example:
\begin{eqnarray}
&&{\tt ComputeSSumBasis[2, \{1, 1/2, 1/3\}, n]} = \nonumber\\
&&\Biggl\{\Biggl\{S_{-2}(n),
                  S_2(n),
                  S_{-1,1}(n),
                  S_2\left(\frac{1}{3};n\right),
                  S_2\left(\frac{1}{2};n\right),
                  S_{1,1}\left(1,\frac{1}{3};n\right),
                  S_{1,1}\left(1,\frac{1}{2};n\right)\Biggr\},
\nonumber\\ &&
   \Biggl\{S_{1,-1}(n)\to
   -S_{-1,1}(n)+S_{-2}(n)+S_{-1}(n) S_1(n),
\nonumber\\ &&
S_{1,1}(n)\to \frac{1}{2}
   S_1(n){}^2+\frac{S_2(n)}{2},
\nonumber\\ &&
S_{1,1}\left(\frac{1}{3},1;n\right)\to
   -S_{1,1}\left(1,\frac{1}{3};n\right)+S_1(n)
   S_1\left(\frac{1}{3};n\right)+S_2\left(\frac{1}{3};n\right),
\nonumber\\ &&
S_{1,1}\left(
   \frac{1}{2},1;n\right)\to -S_{1,1}\left(1,\frac{1}{2};n\right)+S_1(n)
   S_1\left(\frac{1}{2};n\right)+S_2\left(\frac{1}{2};n\right),
\nonumber\\ &&
S_{-1,-1}(n)
   \to \frac{1}{2} S_{-1}(n){}^2+\frac{S_2(n)}{2}\Biggr\}\Biggr\}
\nonumber
\end{eqnarray}

\medskip

\noindent
{\tt ComputeHLogBasis[$w,x,$Alphabet $\rightarrow\{a_1,...,a_k\}$]:} computes a basis representation of
generalized polylogarithms with indices $a_i$ at weight $w$ using algebraic relations.

\smallskip

\noindent
Example:
\begin{eqnarray}
&&{\tt ComputeHLogBasis[2,x,Alphabet \rightarrow \{1/2,-1/3,1\}]} =
\nonumber\\ &&
\Biggl\{\Biggl\{H_{-1,-\frac{1}{3}}(x),H_{-1,\frac{1}{2}}(x),H_{-1,1}(x),
H_{-\frac{1}{3},\frac{1}{2}}(x),H_{-\frac{1}{3},1}(x),
H_{\frac{1}{2},1}(x)\Biggr\},
\nonumber\\ &&
\Biggl\{H_{1,1}(x)\to \frac{1}{2} H_1(x){}^2,
   \nonumber\\ &&
   H_{1,\frac{1}{2}}(x)\to
   H_{\frac{1}{2}}(x) H_1(x)-H_{\frac{1}{2},1}(x),
   \nonumber\\ &&
   H_{1,-\frac{1}{3}}(x)\to
   H_{-\frac{1}{3}}(x) H_1(x)-H_{-\frac{1}{3},1}(x),
   \nonumber\\ &&
   H_{1,-1}(x)\to H_{-1}(x)
   H_1(x)-H_{-1,1}(x),
   \nonumber\\ &&
   H_{\frac{1}{2},\frac{1}{2}}(x)\to \frac{1}{2}
   H_{\frac{1}{2}}(x){}^2,
   \nonumber\\ &&
   H_{\frac{1}{2},-\frac{1}{3}}(x)\to
   H_{-\frac{1}{3}}(x)
   H_{\frac{1}{2}}(x)-H_{-\frac{1}{3},\frac{1}{2}}(x),
   \nonumber\\ &&
   H_{\frac{1}{2},-1}(x) \to H_{-1}(x)
   H_{\frac{1}{2}}(x)-H_{-1,\frac{1}{2}}(x),
   \nonumber\\ &&
   H_{-\frac{1}{3},-\frac{1}{3}}(x) \to \frac{1}{2}
   H_{-\frac{1}{3}}(x){}^2,
   \nonumber\\ &&
   H_{-\frac{1}{3},-1}(x)\to
   H_{-1}(x) H_{-\frac{1}{3}}(x)-H_{-1,-\frac{1}{3}}(x),
   \nonumber\\ &&
   H_{-1,-1}(x)\to
   \frac{1}{2} H_{-1}(x){}^2\Biggr\}\Biggr\}
\nonumber
\end{eqnarray}

\medskip

\noindent
{\tt ComputeSSumInfBasis[$w,\{x_1,..,x_k\}$]:}\\ computes relations of $S$-sums at
infinity and weight $w$ with
indices in $\{x_1,...x_k\}$.

\smallskip

\noindent
Example:
\begin{eqnarray}
&& {\tt ComputeSSumInfBasis[2, \{1/2, -1, -1/3, 1\}]} = \nonumber\\
&& \Biggl\{\Biggl\{S_{-2}(\infty ),S_2\left(-\frac{1}{3};\infty \right)\Biggr\},
\nonumber\\ &&
\Biggl\{S_{1,1}\left(1,\frac{1}{2};\infty \right)\to
   S_1(\infty ) S_1\left(\frac{1}{2};\infty \right)-\frac{1}{2}
   S_{-1}(\infty ){}^2,
\nonumber\\ &&
S_{1,1}\left(1,-\frac{1}{3};\infty \right)\to
   S_1(\infty ) S_1\left(-\frac{1}{3};\infty \right)-\frac{1}{2}
   S_1\left(-\frac{1}{3};\infty \right){}^2,
\nonumber\\ &&
S_{1,1}(\infty )\to \frac{1}{2}
   S_1(\infty ){}^2-S_{-2}(\infty ),
\nonumber\\ &&
S_{1,-1}(\infty )\to S_{-1}(\infty )
   S_1(\infty )-\frac{1}{2} S_{-1}(\infty
   ){}^2,
\nonumber\\ &&
S_{1,1}\left(\frac{1}{2},1;\infty \right)\to -S_{-2}(\infty
   ),
\nonumber\\ &&
S_{1,1}\left(-\frac{1}{3},1;\infty \right)\to \frac{1}{2}
   S_1\left(-\frac{1}{3};\infty \right){}^2+S_2\left(-\frac{1}{3};\infty
   \right),
\nonumber\\ &&
S_{-1,1}(\infty )\to \frac{1}{2} S_{-1}(\infty
   ){}^2+S_{-2}(\infty ),
\nonumber\\ &&
S_{-1,-1}(\infty )\to \frac{1}{2} S_{-1}(\infty
   ){}^2-S_{-2}(\infty ),
\nonumber\\ &&
S_2\left(\frac{1}{2};\infty \right)\to -\frac{1}{2}
   S_{-1}(\infty ){}^2-S_{-2}(\infty ),
\nonumber\\ &&
S_2(\infty )\to -2 S_{-2}(\infty
   )\Biggr\}\Biggr\}
\nonumber
\end{eqnarray}

\medskip

\noindent
{\tt ReduceConstants[$expr$]:} reduces\label{Equ:ReduceConstants} (as much as possible) occurring harmonic sums,
$S$-sums or cyclotomic sums
at infinity and harmonic polylogarithms, generalized polylogarithms and cyclotomic
polylogarithms at constants to a set of basis constants using given tables. Via
the option {\tt ToKnownConstants $\rightarrow$ True} a special set of basis constants is used, \ie $\log(2),\zeta_2,\zeta_3,\ldots$
For the multiple zeta values we use the basis
given before in \cite{Vermaseren:1998uu} for weights up to {\sf w = 6}. The
relations were calculated newly and agree with those in \cite{Vermaseren:1998uu}.
Extended tables (for higher weights) both on the harmonic sums at $N \rightarrow \infty$ and the
the harmonic polylogarithms at $x = 1$ were given in association with Refs.~\cite{
Vermaseren:1998uu,Remiddi:1999ew,Blumlein:2009cf} and are available using {\tt
FORM} \cite{Vermaseren:2000nd} codes. It is straightforwardly possible to generate
more extensive tables for special numbers as harmonic sums, cyclotomic
harmonic sums and classes of generalized harmonic sums at infinity and likewise
special values of the corresponding harmonic polylogarithms at a given argument
over respective bases using {\tt HarmonicSums}.

\smallskip

\noindent
Example:
\begin{eqnarray}
&& {\tt ReduceConstants[S[4, 1, -1, Infinity], ToKnownConstants \rightarrow True]}
= \nonumber\\ &&
 \text{Li}_4\left(\frac{1}{2}\right) \zeta_2
+\frac{\ln(2)^4 \zeta_2}{24}
-\frac{5 \ln(2)^2 \zeta_2^2}{8}
+\frac{3 \ln(2) \zeta_2 \zeta_3}{2}
+\frac{93 \ln(2) \zeta_5}{32}
-\frac{5 \text{s}_6}{2}
-\frac{149 \zeta_2^3}{168}
+\frac{49 \zeta_3^2}{64}
\nonumber
\end{eqnarray}

\noindent
Also tables of sums and polylogarithms at general argument are available.
The harmonic sums and polylogarithms are tabulated up to $w = 6$, for the
cyclotomic harmonic sums over the alphabet $\{\pm 1)^i/i^k,(\pm 1)^i/(2i+1)^k\}$
and the alphabet for the special generalized sums discussed in the present paper
both to weight $w = 4$. Depending on
the weight and length of the alphabet the corresponding calculations may become
more demanding. The corresponding tables of relations grow correspondingly.
If the option {\tt Dynamic $\rightarrow$ True} is set, the reduction is
calculated online. This will help if certain sums are not stored in the tables
yet.

\medskip

\noindent
{\tt HToSinf[$expr$]:} transforms generalized polylogarithms and cyclotomic
polylogarithms at 1 to $S$-sums and
cyclotomic sums at infinity.

\smallskip

\noindent
Example:
\begin{eqnarray}
&& {\tt HToSinf[H[1,2,-4,3,1]]} = \nonumber\\ &&
-S_{1,3}\left(1,\frac{1}{3};\infty
   \right)-S_{2,2}\left(\frac{1}{2},\frac{2}{3};\infty
   \right)-S_{3,1}\left(-\frac{1}{4},-\frac{4}{3};\infty
   \right)+S_{1,1,2}\left(1,\frac{1}{2},\frac{2}{3};\infty
   \right)
\nonumber\\ &&
+S_{1,2,1}\left(1,-\frac{1}{4},-\frac{4}{3};\infty
   \right)
+S_{2,1,1}\left(\frac{1}{2},-\frac{1}{2},-\frac{4}{3};\infty
   \right)-S_{1,1,1,1}\left(1,\frac{1}{2},-\frac{1}{2},-\frac{4}{3};\infty
   \right)+S_4\left(\frac{1}{3};\infty \right)
\nonumber
\end{eqnarray}

\medskip

\noindent
{\tt SinfToH[$expr$]:} transforms $S$-sums and cyclotomic sums at infinity to
generalized polylogarithms and
cyclotomic polylogarithms at 1.

\smallskip

\noindent
Example:
\begin{eqnarray}
&& {\tt SinfToH[S[1, 2, 3, \{1/2, -1/3, 1/4\}, Infinity]} = \nonumber\\ &&
H_{0,0,-6,0,0,-24}(1)-H_{0,0,0,0,0,-24}(1)+H_{2,0,-6,0,0,-24}(1)-H_{2,0,0,0,
   0,-24}(1)
\nonumber
\end{eqnarray}

\medskip

\noindent
{\tt HToS[$expr$]:} computes the power series expansion of generalized polylogarithms and cyclotomic
polylogarithms without trailing zeroes.

\smallskip

\noindent
Example:
\begin{eqnarray}
&&{\tt HToS[H[1, 2, -4, 3, x]} = \nonumber
\\
&&
-\sum_{\tau_1 = 1}^\infty
\frac{2^{-\tau _1} x^{\tau _1}
   S_2\left(\frac{2}{3};\tau _1\right)}{\tau _1^2}
+
\sum_{\tau_1 = 1}^\infty \frac{2^{-\tau _1} x^{\tau
   _1} S_{1,1}\left(-\frac{1}{2},-\frac{4}{3};\tau _1\right)}{\tau
   _1^2}
-\sum_{\tau_1 = 1}^\infty
\frac{x^{\tau _1}
   S_3\left(\frac{1}{3};\tau _1\right)}{\tau _1}
\nonumber\\ &&
+\sum_{\tau_1 = 1}^\infty \frac{x^{\tau _1}
   S_{1,2}\left(\frac{1}{2},\frac{2}{3};\tau _1\right)}{\tau _1}
+\sum_{\tau_1 = 1}^\infty \frac{x^{\tau _1}
   S_{2,1}\left(-\frac{1}{4},-\frac{4}{3};\tau _1\right)}{\tau
   _1}
-\sum_{\tau_1 = 1}^\infty \frac{x^{\tau _1}
   S_{1,1,1}\left(\frac{1}{2},-\frac{1}{2},-\frac{4}{3};\tau _1\right)}{\tau
   _1}
\nonumber\\ &&
-\sum_{\tau_1 = 1}^\infty \frac{4^{-\tau _1} (-x)^{\tau
   _1} S_1\left(-\frac{4}{3};\tau _1\right)}{\tau _1^3}
+\sum_{\tau_1 = 1}^\infty\frac{3^{-\tau
   _1} x^{\tau _1}}{\tau _1^4}
\nonumber
\end{eqnarray}

\medskip

\noindent
{\tt SToH[$expr$]:} takes (part of) a power series expansion of a (generalized/cyclotomic) polylogarithm and
finds the (generalized/cyclotomic) polylogratithm(s) it originates form. It is the inverse of {\tt HToS}.

\smallskip

\noindent
Example:

\begin{eqnarray}
&& {\tt SToH\left[\sum_{\tau_1=1}^\infty
4^{-\tau _1} (-x)^{\tau _1}
   S_1\left(-\frac{4}{3};\tau _1\right)\frac{1}{\tau _1^3}\right]} = -H[0, 0, -4,
3, x] + H[0, 0, 0, 3, x] \nonumber
\end{eqnarray}

\medskip

\vspace*{5mm}\noindent
{\bf Acknowledgment.}~This work has been supported in part by DFG
Sonderforschungsbereich Transregio 9, Computergest\"utzte Theoretische
Teilchenphysik, by the Austrian Science Fund (FWF) grant P20347-N18,
and by the EU Network {\sf LHCPHENOnet} PITN-GA-2010-264564.

\let\oldbibliography\thebibliography
\renewcommand{\thebibliography}[1]{%
  \oldbibliography{#1}%
  \setlength{\itemsep}{0pt}%
}

\end{document}